\documentclass[11pt,letterpaper]{article}
\pdfoutput=1 
\usepackage[utf8]{inputenc}

\usepackage[left=1.00in, right=1.00in, top=1.00in, bottom=1.00in]{geometry}
\usepackage{amsfonts}
\usepackage{amsmath}
\usepackage{amssymb}
\usepackage{amsthm}
\usepackage{float}
\usepackage[font=small,labelfont=bf]{caption} 
\allowdisplaybreaks

\usepackage{footmisc}
\usepackage{booktabs}
\usepackage{ltxtable}
\usepackage{complexity}
\usepackage{MnSymbol}
\usepackage{array,multirow}

\usepackage{bold-extra} 

\usepackage{hyperref}
\usepackage[svgnames]{xcolor}
\hypersetup{colorlinks={true},urlcolor={blue},linkcolor={DarkBlue},citecolor=[named]{DarkGreen}}

\usepackage{microtype}
\usepackage[capitalise,nameinlink,noabbrev]{cleveref}

\usepackage{tikz}
\usetikzlibrary{arrows.meta}

\usepackage{doi}

\usepackage{graphicx}

\usepackage{enumerate}

\theoremstyle{definition}
\newtheorem{definition}{Definition}

\newtheorem{observation}{Observation}

\theoremstyle{plain}
\newtheorem{theorem}{Theorem}[section]
\newtheorem{lemma}[theorem]{Lemma}

\newtheorem{claim}{Claim}

\theoremstyle{remark}
\newtheorem{remark}{Remark}
\newtheorem*{remark*}{Remark}

\Crefname{claim}{Claim}{Claims}

\newcommand{\vallow}{\underline{\alpha}}
\newcommand{\valhigh}{\overline{\alpha}}
\newcommand{\val}{\alpha}
\newcommand{\valstar}{\alpha^*}
\newcommand{\dval}{d\alpha}

\newcommand{\equitwo}{\alpha_2^=}
\newcommand{\equithree}{\alpha_3^=}
\newcommand{\equifour}{\alpha_4^=}
\newcommand{\nquant}{\mathfrak{m}}

\newcommand{\flabel}{\text{first-label}}
\newcommand{\slabel}{\text{second-label}}
\newcommand{\boost}{\text{boost}}

\newcommand{\eps}{\varepsilon}
\DeclareMathOperator{\odeg}{out-deg}
\DeclareMathOperator{\ideg}{in-deg}
\DeclareMathOperator{\outgoing}{out}
\DeclareMathOperator{\incoming}{in}

\def\fp/{\textup{\textsf{FP}}}
\def\p/{\textup{\textsf{P}}}
\def\np/{\textup{\textsf{NP}}}
\def\conp/{\textup{\textsf{co-NP}}}
\def\fnp/{\textup{\textsf{FNP}}}
\def\tfnp/{\textup{\textsf{TFNP}}}
\def\ptfnp/{\textup{\textsf{PTFNP}}}
\def\ppa/{\textup{\textsf{PPA}}}
\def\ppad/{\textup{\textsf{PPAD}}}
\def\ppads/{\textup{\textsf{PPADS}}}
\def\ppp/{\textup{\textsf{PPP}}}
\def\pwpp/{\textup{\textsf{PWPP}}}
\def\pls/{\textup{\textsf{PLS}}}
\def\cls/{\textup{\textsf{CLS}}}
\def\ppadpls/{\textup{$\textsf{PPAD} \cap \textsf{PLS}$}}
\def\ppapls/{\textup{$\textsf{PPA} \cap \textsf{PLS}$}}
\def\eopl/{\textup{\textsf{EOPL}}}
\def\sopl/{\textup{\textsf{SOPL}}}
\def\ueopl/{\textup{\textsf{UEOPL}}}
\def\fixp/{\textup{\textsf{FIXP}}}
\def\bu/{\textup{\textsf{BU}}}
\def\bbu/{\textup{\textsf{BBU}}}
\def\linearfixp/{\textup{\textsf{Linear-FIXP}}}
\def\pspace/{\textup{\textsf{PSPACE}}}

\title{Envy-Free Cake-Cutting for Four Agents}

\author{
\begin{tabular}{cc}
& \\
\textbf{Alexandros Hollender} & \textbf{Aviad Rubinstein}\\
\small{University of Oxford, UK} & \small{Stanford University, USA}\\
\href{mailto:alexandros.hollender@cs.ox.ac.uk}{\small{\texttt{alexandros.hollender@cs.ox.ac.uk}}} & \href{mailto:aviad@cs.stanford.edu}{\small{\texttt{aviad@cs.stanford.edu}}}\\
& \\
\end{tabular}
}

\date{}

\begin{document}

\maketitle

\setcounter{page}{0}
\thispagestyle{empty}

\begin{abstract}
In the envy-free cake-cutting problem we are given a resource, usually called a cake and represented as the $[0,1]$ interval, and a set of $n$ agents with heterogeneous preferences over pieces of the cake. The goal is to divide the cake among the $n$ agents such that no agent is envious of any other agent. Even under a very general preferences model, this fundamental fair division problem is known to always admit an exact solution where each agent obtains a connected piece of the cake; we study the complexity of {\em finding} an approximate solution, i.e., a connected $\varepsilon$-envy-free allocation.

For {\em monotone} valuations of cake pieces, Deng, Qi, and Saberi (2012) gave an efficient ($\poly(\log(1/\varepsilon))$ queries) algorithm for three agents and posed the open problem of four (or more) monotone agents. Even for the special case of {\em additive} valuations, Brânzei and Nisan (2022) conjectured an $\Omega(1/\eps)$ lower bound on the number of queries for four agents. We provide the first efficient algorithm for finding a connected $\varepsilon$-envy-free allocation with four monotone agents.

We also prove that as soon as valuations are allowed to be {\em non-monotone}, the problem becomes hard: it becomes \ppad/-hard, requires $\poly(1/\varepsilon)$ queries in the black-box model, and even $\poly(1/\varepsilon)$ {\em communication complexity}. This constitutes, to the best of our knowledge, the first intractability result for any version of the cake-cutting problem in the communication complexity model.
\end{abstract}

\newpage

\section{Introduction}

The field of fair division studies ways of dividing and allocating a resource among different agents, while ensuring that the division is \emph{fair} in some sense. The perhaps most famous, and certainly most well-studied problem in fair division is \emph{cake-cutting}. In this problem, first introduced by Steinhaus \cite{Steinhaus1948-fair-division}, the resource is a ``cake''. The cake serves as a metaphor for modeling settings where the resource is divisible -- namely it can be divided arbitrarily -- and heterogeneous -- different parts of the resource have different attributes and can thus be valued differently by different agents. Examples where the resource can be modeled as a cake, include, among other things, division of land, time, and other natural resources. The problem has been extensively studied in mathematics and economics \cite{brams1996fair,robertson1998cake} and more recently in computer science \cite{procaccia2013cake}.

More formally, the cake is usually modeled as the interval $[0,1]$, and there are $n$ agents, each with their own preferences over different pieces of the cake. The preferences of each agent $i$ are represented by a valuation function $v_i$ which assigns a value to each piece of the cake. The goal is to divide the cake into $n$ pieces $A_1, \dots, A_n$ and assign one piece to each agent in a \emph{fair} manner. Different notions of fairness have been proposed and studied. Steinhaus \cite{Steinhaus1948-fair-division}, with the help of Banach and Knaster, originally studied a fairness notion known as \emph{proportionality}. An allocation where agent $i$ is assigned piece $A_i$ is proportional, if $v_i(A_i) \geq v_i([0,1])/n$ for all agents $i$. While proportionality ensures that each agent obtains a fair share of the total cake, some agent might be unhappy, because they prefer a piece $A_j$ allocated to some other agent to their own piece $A_i$.

A stronger notion of fairness which tries to address this is \emph{envy-freeness}. An allocation is envy-free, if no agent is envious of another agent's piece. Formally, we require that $v_i(A_i) \geq v_i(A_j)$ for all agents $i,j$. The problem of envy-free cake-cutting was popularized by Gamow and Stern \cite{gamow1958puzzle} and has since been studied extensively. Stromquist \cite{Stromquist80-existence} and Woodall \cite{Woodall80-existence} have shown that an envy-free allocation always exists under some mild continuity assumptions on the valuations. Moreover, this envy-free allocation is also \emph{connected} (or \emph{contiguous}), meaning that every agent is allocated a piece that is a single interval. In other words, the cake is only cut at $n-1$ positions.

The existence proofs of Stromquist and Woodall, as well as the more recent proof of Su \cite{Su99-rental-harmony}, all rely on tools such as Brouwer's fixed point theorem, or Sperner's lemma. These powerful tools prove the existence of envy-free allocations, but they do not yield an efficient algorithm for finding one. Given the practical importance of fair division problems such as cake-cutting, this is a serious drawback to the existence result. Mathematicians and economists have tried to address this by proposing so-called \emph{moving-knife} protocols to solve the problem. Unfortunately, the definition of a moving-knife protocol is mostly informal and thus not well-suited for theoretical investigation (see \cite{brams1997moving} for a discussion). In more recent years, research in theoretical computer science has started studying these questions in formal models of computation. However, despite extensive efforts on the envy-free cake-cutting problem, its complexity is still poorly understood.

In this paper we study the envy-free cake-cutting problem with \emph{connected pieces}. A simple algorithm is known for two agents: the \emph{cut-and-choose} protocol. The first agent cuts the cake in two equal parts, according to its own valuation, and then the second agent picks its favorite piece, and the first agent receives the remaining piece. This simple algorithm can be implemented in the standard Robertson-Webb query model \cite{WoegingerS07-cake}, and yields an envy-free allocation. For three players or more, Stromquist \cite{stromquist2008envy} has shown that no algorithm using Robertson-Webb queries exists for finding envy-free allocations. In order to bypass this impossibility result, Brânzei and Nisan \cite{BranzeiN22-cake-query} propose to study the complexity of finding approximate envy-free allocations instead. An allocation is $\eps$-envy-free, if $v_i(A_i) \geq v_i(A_j) - \eps$ for all agents $i,j$. In this setting, an $\eps$-envy-free allocation can be found by brute force using $\poly(1/\eps)$ queries when the number of agents is constant. An \emph{efficient} algorithm is one that instead only uses $\poly(\log(1/\eps))$ queries.

With this relaxation in place, the problem with three agents can be solved efficiently~\cite{DQS12}. As shown by Brânzei and Nisan~\cite{BranzeiN22-cake-query}, moving-knife algorithms due to Barbanel and Brams \cite{barbanel2004cake} and Stromquist~\cite{Stromquist80-existence} can also be efficiently simulated in the query model. All of these algorithms require that the valuations be \emph{monotone}, a standard assumption in many works on the topic. A valuation function $v_i$ is monotone if $v_i(A) \geq v_i(B)$, whenever $A$ is a superset of $B$.

For four agents, the problem has remained open, even for moving-knife procedures. Barbanel and Brams \cite{barbanel2004cake} provided a moving-knife procedure for four agents but using 5 cuts, improving upon an existing result of Brams et al~\cite{brams1997moving} which used $11$ cuts, and asked whether the minimal number of cuts could be achieved. The problem of finding an algorithm for connected $\eps$-envy-free cake-cutting with four agents was explicitly posed by Deng et al.~\cite{DQS12}. More recently it was conjectured by Brânzei and Nisan~\cite{BranzeiN22-cake-query} that such an algorithm might not exist:

\begin{quote}
``We conjecture that unlike equitability, which remains logarithmic in $1/\eps$ for any number of players, computing a contiguous $\eps$-envy-free allocation for $n = 4$ players [...] will require $\Omega(1/\eps)$ queries.''
\end{quote}

\paragraph*{Our contribution.} Our first contribution is to disprove this conjecture by providing an efficient algorithm that finds an $\eps$-envy-free allocation using $O(\log^3(1/\eps))$ value queries.\footnote{The conjecture of Brânzei and Nisan~\cite{BranzeiN22-cake-query} was stated in a slightly different model (compared to the one in \cite{DQS12}) where one does not assume any upper bound on the Lipschitz-constant of the valuations, but instead allows the added power of cut queries from the Robertson-Webb model (as well as restricting attention to additive valuations). We show that a modification of our algorithm also applies to their model, thus disproving their conjecture; see \cref{sec:RW-algo}.} As for existing approaches for three agents \cite{DQS12,Stromquist80-existence,barbanel2004cake,BranzeiN22-cake-query}, our algorithm also relies on the monotonicity assumption.

In the second part of our work, we investigate whether monotonicity is necessary for obtaining efficient algorithms. We prove that this is indeed the case, in a very strong sense. Namely, we show that the \emph{communication complexity} of finding an $\eps$-envy-free allocation with four \emph{non-monotone} agents is $\Omega(\poly(1/\eps))$. To the best of our knowledge, this is the first intractability result for any version of the cake-cutting problem in the communication model. For the case of agents with identical non-monotone valuations, our reduction also yields an $\Omega(\poly(1/\eps))$ query lower bound, as well as a \ppad/-hardness result in the standard Turing machine model.

Our hardness results improve upon existing lower bounds by Deng et al.~\cite{DQS12} in two ways:
\begin{itemize}
    \item Our lower bounds apply to valuation functions, whereas the lower bounds of Deng et al.\ only apply to a much more general class of preference functions, where the value of an agent for a piece can depend on how the whole cake is divided.
    \item Whereas their lower bounds apply only to the query and Turing machine models, we show hardness in the communication model.
\end{itemize}
We note that the communication model was recently formalized and studied by Brânzei and Nisan~\cite{BranzeiN19-cake-communication}.

\paragraph*{Open problems and future directions.}

In \cref{tab:query,tab:communication,tab:turing} we summarize the current state-of-the-art results for connected $\eps$-envy-free cake-cutting, including our results, in three natural models of computation for this problem. In all of these tables, upper bounds appearing in a certain cell also continue to apply for any cell to the left or above. Similarly, lower bounds also apply to any cell to the right or below. This is due to the fact that ``monotone'' is a special case of ``general'', as well as to the (surprisingly) non-trivial fact that the problem with $n$ agents can be efficiently reduced to the problem with $n+1$ agents; see \cref{app:sec:reduction-to-more-agents}.

\begin{table}[t]
    \centering
    {
    \renewcommand{\arraystretch}{1.4}
    \newcolumntype{C}{>{\centering\arraybackslash}p{2.5cm}}
    \begin{tabular}{|r|C|C|C|C|}
        \hline
        \bf valuations & $\boldsymbol{n=2}$ & $\boldsymbol{n=3}$ & $\boldsymbol{n=4}$ & $\boldsymbol{n \geq 5}$ \\
        \hline
        \bf monotone & \multirow{2}{*}{$\Theta(\log (1/\eps))$} & $O(\log^2 (1/\eps))$ & $O(\log^3 (1/\eps))$ & \textbf{?} \\
        \cline{1-1}\cline{3-5}
        \bf general & & \textbf{?} & $\Theta(\poly (1/\eps))$ & $\Theta(\poly (1/\eps))$ \\
        \hline
    \end{tabular}
    }
    \caption{\textbf{Query complexity bounds for $\boldsymbol{\eps}$-envy-free cake-cutting.} Here ``$\Theta(\poly (1/\eps))$'' denotes that there is a polynomial upper bound and a (possibly different) polynomial lower bound. For $n=2$ the upper bound is obtained by using binary search to simulate the cut-and-choose protocol, while the lower bound is a simple exercise. The upper bound for $n=3$ was shown by Deng et al.~\cite{DQS12}. Alternatively, this bound can also be obtained by simulating the Barbanel-Brams~\cite{barbanel2004cake} or Stromquist~\cite{Stromquist80-existence} moving-knife protocols using value queries. (We note that for additive valuations with the Robertson-Webb query model, Brânzei and Nisan~\cite{BranzeiN22-cake-query} have proved a $\Theta(\log(1(\eps))$ bound for $n=3$.) The remaining results in the table are proved in this paper (\cref{sec:algo,sec:cake-hardness}). All the lower bounds also hold for agents with identical valuations.}
    \label{tab:query}
\end{table}

\begin{table}[t]
    \centering
    {
    \renewcommand{\arraystretch}{1.4}
    \newcolumntype{C}{>{\centering\arraybackslash}p{2.5cm}}
    \begin{tabular}{|r|C|C|C|C|}
        \hline
        \bf valuations & $\boldsymbol{n=2}$ & $\boldsymbol{n=3}$ & $\boldsymbol{n=4}$ & $\boldsymbol{n \geq 5}$ \\
        \hline
        \bf monotone & \multirow{2}{*}{$O(\log (1/\eps))$} & $O(\log (1/\eps))$ & $O(\log^2 (1/\eps))$ & \textbf{?} \\
        \cline{1-1}\cline{3-5}
        \bf general & & \textbf{?} & $\Theta(\poly (1/\eps))$ & $\Theta(\poly (1/\eps))$ \\
        \hline
    \end{tabular}
    }
    \caption{\textbf{Communication complexity bounds for $\boldsymbol{\eps}$-envy-free cake-cutting.} For $n=2$ the bound is obtained by the cut-and-choose protocol. For $n=3$ the bound follows by a simple generalization of the same bound shown for additive valuations by Brânzei and Nisan~\cite{BranzeiN22-cake-query}. For $n=4$ the upper bound for monotone valuations follows from our algorithm together with some tools introduced in \cite{BranzeiN22-cake-query}. We provide proof sketches for these two bounds in \cref{app:communication}. The lower bounds are proved in this paper (\cref{sec:cake-hardness}).}
    \label{tab:communication}
\end{table}

\begin{table}[t]
    \centering
    {
    \renewcommand{\arraystretch}{1.4}
    \newcolumntype{C}{>{\centering\arraybackslash}p{2.5cm}}
    \begin{tabular}{|r|C|C|C|C|}
        \hline
        \bf valuations & $\boldsymbol{n=2}$ & $\boldsymbol{n=3}$ & $\boldsymbol{n=4}$ & $\boldsymbol{n \geq 5}$ \\
        \hline
        \bf monotone & \multirow{2}{*}{\textsf{P}} & \textsf{P} & \textsf{P} & \textbf{?} \\
        \cline{1-1}\cline{3-5}
        \bf general & & \textbf{?} & \ppad/ & \ppad/ \\
        \hline
    \end{tabular}
    }
    \caption{\textbf{Computational complexity of $\boldsymbol{\eps}$-envy-free cake-cutting.} Here ``\ppad/'' denotes that the problem is \ppad/-complete. The problem lies in \ppad/ for all $n$. As before, the results for $n=2$ and $n=3$ follow from cut-and-choose and the work of Deng et al.~\cite{DQS12}, respectively. The remaining results in the table are proved in this paper (\cref{sec:algo,sec:cake-hardness}), and the \ppad/-hardness results also hold for agents with identical valuations.}
    \label{tab:turing}
\end{table}

The following questions are particularly interesting:
\begin{itemize}
    \item What is the complexity of the problem for five agents with monotone valuations? Can the problem be solved efficiently, or can we show a lower bound? Can we at least show a lower bound for some larger number of players?
    \item What is the complexity of the problem for three agents with general valuations?
    \item Can any of the insights used in our algorithm for four players be used to tackle the problem of proving existence of EFX allocations for four agents in the indivisible goods setting?
\end{itemize}

\paragraph*{Further related work.}

The envy-free cake-cutting model has also been extensively studied without requiring that the pieces allocated to the agents be connected. The celebrated works of Aziz and Mackenzie~\cite{AzizM20} have shown that an exact envy-free allocation can be found in the Robertson-Webb query model for any number of players, albeit with a prohibitive number of cuts. It is an open question whether this can be improved, as only a lower bound of $\Omega(n^2)$ due to Procaccia~\cite{procaccia2013cake} is currently known.

Efficient algorithms for envy-free cake cutting have been obtained for some relatively large constant approximation factors~\cite{Arunachaleswaran19,GoldbergHS20-cake,BarmanK22-approx-EF-cake}, or under stronger assumptions, e.g.~\cite{Branzei15,BarmanR22}.
Our focus in this work is on (approximately) {\em envy-free} cake cutting, but other objectives have been considered such as equitable~\cite{ProcacciaW17}, (Nash) Social Welfare~\cite{AumannDH13,Arunachaleswaran19}, and extensions to group fairness~\cite{Segal-HaleviN19,Segal-HaleviS21,SegalHaleviS23}.
In this work we are only concerned with the algorithmic question of finding an envy free allocation, but several works have also considered the important question of incentives, e.g.~\cite{MosselT10,MayaN12,ChenLPP13,AumannDH14,BeiCHTW17,OrtegaS22,Tao22}.

\section{Technical Overview}

We begin with a technical overview of the algorithm in \cref{sub:overview-alg}.
The most novel technical contribution of our work is in proving a communication complexity lower bound for non-monotone agents. In \cref{sub:overview-hardness} we contrast with previous work and explain the first step in the reduction, a novel communication variant of the {\sc End-of-Line} problem. In \cref{sub:overview-embed} we give an overview of our embedding of the discrete {\sc Intersection End-of-Line} into the inherently continuous {\sc EF-Cake-Cutting} problem. 

\subsection{Overview of the algorithm}\label{sub:overview-alg}
The core idea for our algorithm is to define a special invariant, parameterized by $\val \in [0,1]$ that satisfies the following useful desiderata:
\begin{itemize}
    \item For any $\val \in [0,1]$ we can efficiently find a partition satisfying the invariant, if one exists.
    \item Starting from any partition satisfying the invariant, there is a continuous path in the space of partitions where the invariant holds and $\val$ increases monotonically. The partition at the end of this path yields an EF allocation.
    \item The invariant is guaranteed to hold at Agent 1's equipartition\footnote{Unless the equipartition is already envy-free, in which case we are done.} (i.e., a partition where she values all pieces of the cake equally --- this equipartition can also be found efficiently). It is guaranteed {\em not} to hold for $\val = 1$.
\end{itemize}

Together, these suggest a simple algorithm for finding $\varepsilon$-EF allocations: Set $\vallow$ to be the $\val$ corresponding to Agent 1's equipartition, and set $\valhigh = 1$. Then we can use binary search to find a (locally) maximal $\val$ where the invariant holds, and output the corresponding partition.

The key is of course in identifying such a nice invariant. In \cref{sec:proof-existence} we identify such an invariant and use it to present a new proof of the existence of envy-free allocations for four monotone agents. The invariant is defined as the OR of the following two conditions:

\begin{description}
\item[Condition A:] Agent 1 is indifferent between its three favorite pieces, and the remaining piece is (weakly) preferred by (at least) two of the three other agents.
\item[Condition B:] Agent 1 is indifferent between its two favorite pieces, and the two remaining pieces are each (weakly) preferred by (at least) two of the three other agents.
\end{description}
The value $\val$ is the value that Agent 1 has for its favorite piece(s) in the division.

We then use these insights to obtain an efficient algorithm in \cref{sec:algo}. In \cref{sec:RW-algo} we also present a variant of the algorithm that works in the Robertson-Webb model (i.e., for additive valuations without the bounded Lipschitzness assumption).

\subsection{Technical highlight: communication of \textsc{End-of-Line}}\label{sub:overview-hardness}

\subsubsection*{Background: totality, \textsc{End-of-Line}, and lifting gadgets}

The first major obstacle for proving hardness of cake cutting is that the problem is total, i.e., there always exists an (exactly) envy-free allocation. 
Since the proof of existence uses Sperner's Lemma / Brouwer's fixed point theorem, the natural candidate for understanding the complexity of actually finding the solution is the {\sc End-of-Line} problem.

\begin{definition}[{\sc End-of-Line}~\cite{Papadimitriou94}] \hfill

\begin{description}
\item[Input] 
A directed graph $G = (V,E)$, described as functions (or circuits in the computational model) $S,P$ that, for each vertex, return its list of outgoing, incoming edges (respectively); a special vertex $0$ with a single outgoing edge and no incoming edges.

\item[Output] One of the following:
\begin{itemize}
\item A vertex (different from $0$) with in-degree $\neq$ out-degree.
\item An inconsistency:  pair of vertices $u,v$ such that $S(u) = v$ but $P(v) \ne u$ (or vice versa). 
\item A vertex with more than one outgoing or more than one incoming edges.\footnote{Usually $S,P$ are hardcoded to have exactly one neighbor; this representation will be more convenient for our purposes.}
\end{itemize}
\end{description}
\end{definition}

It is known that {\sc End-of-Line} is hard in the query complexity model. Furthermore, while {\sc End-of-Line}  is inherently a discrete graph problem, using known techniques, e.g., due to~\cite{DQS12}, we can embed it as a continuous cake-cutting instance (albeit for general {\em preferences}; our extension to non-monotone {\em valuations} is non-trivial).
The instances resulting from this reduction require high query complexity even when all the agents' preferences are identical --- i.e., they do not at all capture the difficulty from communication between agents with different utilities. 

Fortunately, there is a very powerful machinery for {\em lifting} query complexity lower bounds to communication complexity. These lifting theorems (sometimes also called ``simulation theorems'') replace each bit in the query problem with a small {\em lifting gadget} whose input is distributed between the parties of the communication problem. This must be done in a special way to ensure that the problem remains hard for communication complexity. (E.g., naively partitioning the edges of the {\sc End-of-Line} instance results in a communication-easy problem!)

\paragraph{The main dilemma.} On one hand, those special lifting gadgets are inherently discrete, so even if we have a query-hard cake cutting instance, if we try to directly lift it to communication complexity, the resulting information structure looks nothing like agents' valuations over intervals of cake. On the other hand, if we first lift the {\sc End-of-Line} problem, the new communication problem loses the special structure of  {\sc End-of-Line}, and it is not clear how to embed it as a continuous, total problem.

\subsubsection*{Previous approaches (and why they fail for cake cutting)}

The same conundrum was previously encountered in works on communication complexity of Nash equilibria and Brouwer fixed points~\cite{RW16, GR18, BR20,GSP21}. These works used very different approaches for this issue, but ultimately they all consider Brouwer's fixed point on a hypercube of some dimension, and then go to higher dimensions%
\footnote{\cite{GR18} obtain tight bounds by reducing the ``dimension'' of their instance by encoding a single point in the Brouwer hypercube across the different actions in the support of a player's mixed strategy at equilibrium. It would be very interesting if there is an analogue of this approach for cake cutting.} to embed the lifting gadgets.
In fact, while the query variant of Brouwer is already hard in 2-D~\cite{HPV89}, none of these papers give non-trivial lower bounds on communication complexity for any constant dimension. 

For cake cutting, we would really like to keep the dimension as small as possible:
First, in the valuations model, the input is a function from the cake intervals; each interval is defined by only two cuts, i.e., it is inherently a 2-D object.  Furthermore, even if we try to get away with higher dimension by simultaneously considering more cuts in the partition in the preferences model, the number of cuts (and hence the number of agents) should scale with the dimension. But our goal is to show hardness for as few agents as possible.

\subsubsection*{Step I: A new total communication problem: \textsc{Intersection End-of-Line}}

Instead of working with the {\sc End-of-Line} instance and the lifting gadgets separately, we present a new communication variant of  {\sc End-of-Line} that naturally combines {\sc End-of-Line} with the flagship hard problem of communication complexity: {\sc SetDisjointness} (equivalently, ``{\sc SetIntersection}''):

\begin{definition}[{\sc Intersection End-of-Line}] \hfill

\begin{description}
\item[Input] 
Each party receives a superset of the edges; we say that an edge is {\em active} in the {\sc End-of-Line} instance $G = (V,E)$ if it is in the intersection of all the supersets. Edges that only appear in the supersets of some parties are called {\em inactive}.

\item[Output] A solution to the {\sc End-of-Line} instance defined by the active edges.
\end{description}
\end{definition}

We prove that {\sc Intersection End-of-Line} requires communication polynomial in the input size, even under strong assumptions on the inputs. In particular, by starting from a number-on-forehead lifting theorem, we can get instances where every inactive edge only appears in one party's inputs. This is helpful for embedding as a cake cutting problem, because for any cake-partition most of the parties will correctly evaluate it and can guarantee a conflict in case it does not correspond to a solution to the original {\sc End-of-Line} instance.%
\footnote{While capturing the intuition, this is not exactly true. In some points in our reduction 2 out of 4 parties do not have sufficient information. But in those special points there is additional structure that allows us to guarantee a conflict between the 2 remaining parties.}

\begin{lemma}

{\sc Intersection End-of-Line} with $k \ge 3$ parties requires $\poly(|G|)$ communication complexity,
even in the special case where the instances satisfy the following promises: 
\begin{enumerate}
\item[0.] Every node has at most one active incoming edge, and at most one active outgoing edge.
\item Every inactive edge is only included in at most one party's set of edges.
\item No vertex has both an $i$-inactive and a $j$-inactive incident edge, for some $i \neq j$. (An edge is said to be {\em $i$-inactive}, if it is inactive, but included in Party $i$'s superset.)
\item Every vertex is assigned to one party (the vertex-party assignment is publicly known): Only that party's superset  may have more than one edge coming into that vertex. Similarly, only that party's superset may have more than one edge going out of that vertex.
\end{enumerate}
\end{lemma}

We expect that {\sc Intersection End-of-Line} will find other applications. For example, it can be used to show that various (total) communication variants of Brouwer's fixed point problem are hard even in 2-D.

\subsubsection*{A few words about the hardness of \textsc{Intersection End-of-Line}}

We give a very brief overview of the proof of the lemma (see \cref{sub:hard-intersect-EoL} for details).
Our starting point is a number-on-forehead lifting of {\sc End-of-Line}, i.e., the parties need to solve a small communication problem (the ``lifting gadget'') to compute $S(v)$ and $P(v)$ for each $v$. 

For every possible edge $(u \rightarrow v)$, and every possible vector of inputs to $v$'s and $u$'s lifting gadgets corresponding to edge $(u \rightarrow v)$, we construct a sequence of vertices that incrementally build this vector. At each vertex in the sequence, the edge to the next vertex is determined by an input on a single party's forehead. That party's superset includes all the edges that could correspond to this input, while all the other parties only keep the edge that corresponds to the true input. When we take the intersection of the parties supersets of edges, we're left with just the path from $u$ to $v$.

\subsection{Overview of reduction to Communication Complexity of Cake Cutting}\label{sub:overview-embed}
We now give an overview of Steps II and III of our hardness result for Cake Cutting in the communication model. 
This is probably the most technically involved part of our paper.

\subsubsection*{Step II: Embedding \textsc{Intersection End-of-Line} in the plane}
Our next step is to embed {\sc Intersection End-of-Line} à-la Sperner-coloring in 2D. 
Our approach is inspired by the work of~\cite{CD09} that embed {\sc End-of-Line} as a 2D Sperner coloring problem, but the final embedding is far more delicate in order to accommodate the double information restriction imposed by our setting: (i) first, because the input to {\sc Intersection End-of-Line} is distributed, parties disagree about which path goes through any point in the plane; (ii) in the next step, we want to turn our embedding into {\em valuations} of cake pieces, but each piece only ``sees'' one or two of the three cuts defining the partition. 

Our embedding assigns a color for each grid point in the discretized unit square; colors of other points are defined as weighted averages of the corner of their cells. Each color consists of two binary labels $(\pm 1, \pm 1)$, and a solution to the Sperner problem consists  of a grid cell where two of the four corners have opposite labels in both coordinates.  
In the communication problem, the embedding for each player is different because they have a different view of the graph, but we guarantee, roughly, that in at least one coordinate a majority of the agents see the same sign.

Modulo using the somewhat non-standard four colors (rather than three) for a 2D Sperner, at a high level our embedding is similar to previous work (e.g.~\cite{HPV89,CD09}): Embedded {\sc End-of-Line} paths are colored by three stripes: $(+1,+1),(-1,+1),(-1-1)$, with the background (where embedded paths do not pass) colored by the default $(+1,-1)$; intuitively this ensures that Sperner solutions (i.e., $(+1,+1),(-1,-1)$ or $(+1,-1),(-1,+1)$ in the same grid cell) can only happen at endpoints of embedded paths.

\paragraph{A special intersection gadget.}
Since we're working in 2D, we also have to worry about intersection of embedded edges/paths; at a high level, we use the intersection gadget of~\cite{CD09}.
However, the details are extremely subtle, especially where 1 party sees a crossing of embedded edges, while for the other 3 parties only one of those edges exist. We design custom variants of both the intersection gadget and simple horizontal paths (we don't need to modify vertical paths) that guarantee a very specific desideratum. See \cref{sub:special-embedding} for the full details.

\subsubsection*{Step III: How to evaluate a piece of cake}
The final step in our construction is to define the valuations of cake pieces. 
The value of every piece is always fairly close to its length, so we can assume wlog that in any $\varepsilon$-EF partition, the lengths are not too far from equal; in particular it is easy to tell whether a piece is the second or third piece based on its endpoints.

Given cuts $(\ell,m,r)$ we want to ensure that $\ell \approx 1-r$; then we use $\ell\approx 1-r$ to encode the $x$ coordinate  of the unit square, and $m$ to encode the $y$ coordinate. The challenges are that (i) the values of the external pieces cannot depend on the $y$ coordinate of the embedding (this is why our embedding treats horizontal and vertical paths differently); and (ii) it is difficult to enforce  $\ell \approx 1-r$ without creating spurious envy-free solutions. 

Since the first and fourth piece do not ``see'' the cut $m$ that encodes the $y$ coordinate, their values are simpler.
The value of the fourth piece is always exactly its length.  The value of the first piece is usually exactly its length (this helps ensure that $\ell \approx 1-r$), except when it receives a special boost (discussed below) on carefully chosen vertical strips.

The second and third piece ``see'' both coordinates of the embedding, so we can slightly adjust their values based on the labels of the corresponding grid cells; specifically, for a sufficiently small $\gamma>0$, the value of the second piece is equal to the value of the first piece plus $\gamma\cdot \flabel$ and the value of the third piece is equal to the value of the fourth piece plus $\gamma \cdot \slabel$. This guarantees, for example, that in any grid cell where all parties agree that the first label is positive (resp.~negative), the second piece is over-demanded (resp.~under-demanded) so the corresponding partition cannot be $\varepsilon$-EF. 

When only three of the parties agree, a positive first label will still cause the second piece to be over-demanded, but with a negative first label the fourth party may find the second piece acceptable, so it would not be under-demanded. So the riskiest regions are ones where three parties share a negative first label and neutral second label (resp.~neutral first and negative second). In this case we want to ensure that the fourth party does not want the second piece (resp.~the third). 

In some carefully chosen vertical strips, the first (and thus also the second piece) receive a {\em boost} of $\pm \beta$ (we set $\beta = 8\gamma$). Specifically, in strips where agent $i$ thinks that the remaining three agents may have a neutral second label, she gives the first and second pieces a negative boost in case the remaining three agents have a negative first label. Similarly, in strips where agent $i$ thinks that the remaining three agents may have a neutral first label, she gives the first and second pieces a positive boost in case the remaining three agents have a negative second label. Note that because of the careful construction of the crossing gadgets, agent $i$ knows based only on the $x$ coordinate where the other 3 parties may see a path she does not, and in what direction it goes. This allows her to correctly implement the boost on the first piece which doesn't ``see'' the $y$ coordinate.

The full details of the construction can be found in \cref{sec:cake-hardness}.

\section{Preliminaries}

We consider a resource, called ``the cake'', which is modeled as the interval $[0,1]$. There are $n$ agents, and each agent has a valuation function defined over intervals of the cake. Formally, each agent $i \in [n]$ has a valuation function $v_i: [0,1]^2 \to [0,1]$, where $v_i(a,b)$ represents the value that agent $i$ has for interval $[a,b]$. For this to be well-defined, we require that $v_i(a,b) = 0$ whenever $b \leq a$. Furthermore, we always assume that valuations $v_i$ are continuous.

Next, we define some additional properties of valuation functions that we will sometimes assume, depending on the setting. A valuation $v$ is:
\begin{itemize}
    \item \emph{monotone}: $v(a,b) \leq v(a',b')$, whenever $[a,b] \subseteq [a',b']$.
    \item \emph{hungry}: $v(a,b) < v(a',b')$, whenever $[a,b] \subsetneq [a',b']$.
    \item \emph{$\delta$-strongly-hungry}, for some $\delta > 0$: $v(a',b') \geq v(a,b) + \delta(b'-b) + \delta(a-a')$, whenever $[a,b] \subseteq [a',b']$.
\end{itemize}
Note that hungriness corresponds to strict monotonicity. Furthermore, a valuation that is $\delta$-strongly-hungry for some $\delta > 0$ is in particular hungry.

For our computational results (except in \cref{sec:RW-algo}), we will assume Lipschitz-continuity of the valuations. A valuation $v$ is Lipschitz-continuous with Lipschitz constant $L$, if, for all $a,b,a',b' \in [0,1]$:
$$|v(a,b) - v(a',b')| \leq L (|a-a'| + |b-b'|).$$

A (connected) allocation of the cake to $n$ agents is a division of the cake $[0,1]$ into $n$ intervals $A_1, \dots, A_n$ such that $A_i$ is assigned to agent $i$, the $A_i$'s are all pairwise disjoint,\footnote{To be more precise, here we assume that all $A_i$ are closed intervals and we say that $A_i$ and $A_j$ are disjoint if their intersection has zero measure.} and $\bigcup_i A_i = [0,1]$. In the setting with four agents, we will denote a division of the cake by its three cuts $(\ell,m,r)$, where $0 \leq \ell \leq m \leq r \leq 1$.

\paragraph*{\textbf{Envy-freeness.}}
An allocation $(A_1, \dots, A_n)$ is said to be \emph{envy-free} if, for all agents $i$, we have $v_i(A_i) \geq v_i(A_j)$ for all $j$. For $\eps \in [0,1]$, we say that the allocation is \emph{$\eps$-envy-free} if, for all agents $i$, we have $v_i(A_i) \geq v_i(A_j) - \eps$ for all $j$. An envy-free allocation is guaranteed to exist in the setting we consider~\cite{Stromquist80-existence,Woodall80-existence,Su99-rental-harmony}.

\paragraph*{\textbf{Normalization.}} For our computational results (except in \cref{sec:RW-algo}), without loss of generality, we will assume that the valuation functions are $1$-Lipschitz-continuous, i.e., $L=1$. If $L > 1$, we can replace $v_i$ by $v_i/L$, and $\eps$ by $\eps/L$.

\paragraph*{\textbf{Query complexity.}}
In this black-box model, we can query the valuation functions of the agents. A (value) query consists of the endpoints of an interval $[x,y]$, and the agent responds with its value for that interval, i.e., $v_i(x,y)$. The running time of an algorithm consists of the number of queries. The problem of computing an $\eps$-envy-free allocation can be solved using $\poly(1/\eps)$ queries by brute force~\cite{BranzeiN22-cake-query}. We say that an algorithm is efficient if it uses $\poly(\log(1/\eps))$ queries.

The problem has traditionally been studied in an extended query model called the Robertson-Webb model~\cite{WoegingerS07-cake}. In this model, in addition to the value queries, there is a second type of query called a \emph{cut} query. On cut query $(x,\alpha)$, where $x$ is a position on the cake and $\alpha$ is a value, the agent responds by returning a position $y$ on the cake such that $v_i(x,y) = \alpha$ (or responds that there is no such $y$). Although this model is usually studied with additive valuations, it can easily be extended to the setting of monotone valuations.\footnote{The natural extension of the Robertson-Webb model to monotone valuations should also allow for a cut query in the other direction, i.e., on cut query $(\alpha,y)$, the agent responds by returning a position $x$ on the cake such that $v_i(x,y) = \alpha$ (or responds that there is no such $x$). For additive valuations such queries can be easily simulated using standard cut queries.}

In the monotone setting, with our assumption of Lipschitz-continuity of the valuations and in the context of looking for approximate fairness (such as $\eps$-envy-freeness), the two models are equivalent up to $\poly(\log(1/\eps))$ factors, since a cut query can be simulated by $O(\log(1/\eps))$ value queries~\cite{BranzeiN22-cake-query}. In particular, a lower bound for the (value) query complexity also implies a (qualitatively) similar lower bound for the Robertson-Webb query model.

Since the Robertson-Webb model is usually studied for additive valuations and \emph{without} the bounded Lipschitz-continuity assumption, we also present a version of our algorithm that applies to this setting, namely additive valuations that are non necessarily $L$-Lipschitz-continuous. See \cref{sec:RW-algo} for the details and the definition of the Robertson-Webb model.

\paragraph*{\textbf{Communication complexity.}}
In this model, each agent corresponds to one party in a communication setting. Each agent/party only knows its own valuation function. The parties can communicate by sending messages and the complexity of such a communication protocol is measured in terms of the number of bits sent between parties. A protocol is efficient if it uses at most $\poly(\log(1/\eps))$ communication.

A query algorithm yields a communication protocol of (qualitatively) equivalent complexity \cite{BranzeiN19-cake-communication}. As a result, a lower bound in the communication complexity setting yields a lower bound in the query complexity setting.

\paragraph*{\textbf{Computational complexity.}}
In the ``white box'' model, we assume that the valuations are given to us in the input, and an efficient algorithm is one that runs in polynomial time in the size of the representation of the valuations and in $\poly(\log(1/\eps))$. For example, the valuations can be given as well-behaved arithmetic circuits \cite{FearnleyGHS22-gradient} or as Turing machines together with a polynomial upper bound on their running time. In this model, the problem is a total \np/ search problem, i.e., it lies in the class \tfnp/. Furthermore, it is known to lie in the subclass \ppad/ of \tfnp/ \cite{DQS12}.

\section{A New Proof of Existence of Connected Envy-free Allocations for Four Monotone Agents}\label{sec:proof-existence}

In this section we present a new proof for the existence of connected envy-free allocations for four agents with monotone valuations. This new proof is the main insight that allows us to obtain efficient algorithms for the problem in subsequent sections.

\begin{theorem}\label{thm:existence}
For four agents with monotone valuations, there always exists a connected envy-free allocation.
\end{theorem}

Existing proofs \cite{Stromquist80-existence,Woodall80-existence,Su99-rental-harmony} apply to more general valuations, but rely on strong topological tools such as Brouwer's fixed point theorem, or equivalent formulations. Here we show that monotonicity allows for a simpler proof that bypasses these tools.

\paragraph{\textbf{Hungry valuations.}}
The first step of the proof is to restrict ourselves to hungry valuations. This is without loss of generality, since the existence of envy-free allocations for hungry agents is sufficient to obtain existence for the non-hungry case too. Indeed, given (possibly non-hungry) valuations $v_i$, for any $\eps \in (0,1)$ we can replace them by hungry valuations $v_i'$ defined as $v_i'(a,b) := (1-\eps)v_i(a,b) + \eps(b-a)$. An envy-free allocation for the $v_i'$ will then be approximately envy-free for the $v_i$, where the approximation can be made arbitrarily small by picking $\eps > 0$ sufficiently small. The following standard compactness argument then shows that a connected envy-free allocation must also exist for the original valuations $v_i$. Thus, in the remainder of this section we assume that the valuations are hungry.

\begin{lemma}
If for some valuations $v_1, \dots, v_n$ there exist $\eps$-envy-free allocations for all $\eps > 0$, then there exists an (exact) envy-free allocation.
\end{lemma}

\begin{proof}
Consider a sequence $(A^k)_k$ of $1/k$-envy-free divisions. Since the domain of all divisions is compact, there exists a converging subsequence of $(A^k)_k$. Furthermore, there exists such a converging subsequence such that the assignment of pieces to agents (i.e., Agent 1 gets piece $3$, etc) that makes the division $1/k$-envy-free remains the same throughout the subsequence. Since the valuations are continuous, the division that is the limit of the constructed subsequence together with the fixed assignment of pieces to agents must be envy-free.
\end{proof}

\paragraph{\textbf{Equipartitions.}}
The next step is to prove that an envy-free allocation always exists when all four agents have the same valuation function $v_i = v$. In this case an envy-free allocation corresponds to an \emph{equipartition} according to $v$, i.e., a division into four pieces which all have the same value according to $v$.  The monotonicity of the valuation function allows for a simple proof of existence of equipartitions.

\begin{lemma}\label{lem:equipartition}
For any monotone hungry valuation function $v$, and any $n \geq 1$, there exists a (unique) equipartition of the cake into $n$ equal parts according to $v$.
\end{lemma}

\begin{proof}
We only present the proof here for $n=4$, but it straightforwardly generalizes to any $n \geq 1$. For any $\alpha \in [0,1]$, we let $\ell(\alpha)$ denote the unique cut position satisfying $v(0,\ell(\alpha)) = \alpha$ if it exists, and set $\ell(\alpha) = 1$ otherwise (i.e., when $v(0,1) < \alpha$). Since $v$ is hungry and continuous, it is easy to see that $\ell(\cdot)$ is well-defined and continuous. Similarly, we define $m(\cdot)$ such that $v(\ell(\alpha),m(\alpha)) = \alpha$ (or $m(\alpha) = 1$ if this is not possible), and $r(\cdot)$ such that $v(m(\alpha),r(\alpha)) = \alpha$ (or $r(\alpha) = 1$ if this is not possible). Finally, since $v(r(0),1) - 0 = v(0,1) > 0$ and $v(r(1),1) - 1 = -1 < 0$, by continuity there exists $\alpha^* \in (0,1)$ such that $v(r(\alpha^*),1) = \alpha^*$. Then the division $(\ell(\alpha^*),m(\alpha^*),r(\alpha^*))$ is an equipartition into four parts. Furthermore, it is easy to see that the hungriness of $v$ also implies that the equipartition is unique.
\end{proof}

\paragraph{\textbf{Continuous path: Intuition.}}
We are now ready to prove the existence of envy-free allocations. At a high level, the proof proceeds as follows. Starting from the equipartition into four parts according to Agent 1's valuation, we show that we can move the cuts in a continuous manner such that the following two properties hold:
\begin{enumerate}
\item As long as we have not reached an envy-free allocation, there is a way to continue moving the cuts.
\item When we move the cuts, we make sure that the value of Agent 1 for its favorite piece always strictly increases.
\end{enumerate}
It then follows that this continuous path in the space of divisions has to terminate, and thus an envy-free division must exist.

In more detail, we move the cuts continuously while maintaining the following invariant: the division is \emph{critical}. We say that a given division into four pieces is \emph{critical} if it satisfies Condition A or Condition B (or both):
\begin{description}
\item[Condition A:] Agent 1 is indifferent between its three favorite pieces, and the remaining piece is (weakly) preferred by (at least) two of the three other agents.
\item[Condition B:] Agent 1 is indifferent between its two favorite pieces, and the two remaining pieces are each (weakly) preferred by (at least) two of the three other agents.
\end{description}
It is easy to see that Condition A necessarily holds at the starting point, i.e., for the equipartition into four parts according to Agent 1, unless this division is already envy-free. As long as Condition A holds, we shrink the remaining piece, while also making sure that Agent 1 remains indifferent between the other three pieces. When we reach a point where Condition A would no longer hold if we continued shrinking the piece, we can show that either the division is envy-free, or Condition B is satisfied. As long as Condition B holds, we increase the two pieces preferred by Agent 1, while making sure that Agent 1 remains indifferent between them, and also that one of the other agents remains indifferent between the remaining two pieces. When we reach a point where Condition B would no longer hold if we continued, we can show that either the division is envy-free, or Condition A is satisfied. This continuous movement of cuts thus satisfies point 1 above, namely that the path can always be extended as long we have not found an envy-free division. Furthermore, the value that Agent 1 has for its favorite pieces also always strictly increases, so point 2 is also satisfied.

\paragraph{\textbf{Continuous path: Formal argument.}}
We begin with the formal definitions of the two conditions. A division into four pieces $(P_1,P_2,P_3,P_4)$ is \emph{critical} if at least one of the following two conditions holds:
\begin{description}
\item[Condition A:] There exists a piece $k \in \{1,2,3,4\}$ such that (i) for all other pieces $t,t' \in \{1,2,3,4\} \setminus \{k\}$, $v_1(P_t) = v_1(P_{t'}) \geq v_1(P_k)$, and (ii) there exist two (distinct) agents $i,i' \in \{2,3,4\}$ such that $v_i(P_k) \geq \max_t v_i(P_t)$ and $v_{i'}(P_k) \geq \max_t v_{i'}(P_t)$.
\item[Condition B:] There exist two (distinct) pieces $k,k' \in \{1,2,3,4\}$ such that (i) for the other two pieces $t,t' \in \{1,2,3,4\} \setminus \{k,k'\}$, $v_1(P_t) = v_1(P_{t'}) \geq \max\{v_1(P_k),v_1(P_{k'})\}$, and (ii) for each piece $t \in \{k,k'\}$, there exist two agents $i,i' \in \{2,3,4\}$ such that $v_i(P_t) \geq \max_{t'} v_i(P_{t'})$ and $v_{i'}(P_t) \geq \max_{t'} v_{i'}(P_{t'})$.
\end{description}
Note that Condition A and Condition B can possibly hold at the same time.

By \cref{lem:equipartition} there exists an equipartition of the cake into four parts according to Agent 1. If this division is not envy-free, then it satisfies Condition A. The following lemma then allows us to conclude that an envy-free division exists and thus provides a new proof of \cref{thm:existence}. The lemma will also be used in the next sections when we develop algorithms for the problem.

\begin{lemma}\label{lem:continuous-path}
Let $(\ell_0,m_0,r_0)$ be a critical division and let $\alpha_0$ denote the value that Agent 1 has for its favorite pieces in that division. Then there exists $\valstar \in [\alpha_0,1)$ and a continuous path $T: [\alpha_0,\valstar]\to \{(\ell,m,r): 0 \leq \ell \leq m \leq r \leq 1\}$ such that:
\begin{enumerate}
    \item The path begins at $T(\alpha_0) = (\ell_0,m_0,r_0)$.
    \item For all $\val \in [\alpha_0, \valstar]$, the division $T(\val)$ is critical and is such that Agent 1 has value $\val$ for its favorite pieces.
    \item The division $T(\valstar)$ at the end of the path is envy-free.
\end{enumerate}
\end{lemma}

\noindent The following two claims will play a crucial role in the proof of \cref{lem:continuous-path}.

\begin{claim}\label{clm:condB}
If a division $(P_1, P_2, P_3, P_4)$ satisfies Condition B with pieces $k, k'$, but is not envy-free, then each of the Agents 2, 3, and 4 strictly prefers one of the pieces $k$ or $k'$ (or both) to any other piece, i.e., $\max\{v_i(P_k), v_i(P_{k'})\} > \max_{t \neq k,k'} v_i(P_t)$ for all $i \in \{2,3,4\}$.
\end{claim}

\begin{proof}
According to Condition B, the (distinct) pieces $k,k'$ satisfy (i) Agent 1 (weakly) prefers any of the two pieces in $\{1,2,3,4\} \setminus \{k,k'\}$, and (ii) for each of the pieces $t \in \{k,k'\}$, there are at least two distinct agents in $\{2,3,4\}$ who (weakly) prefer piece $t$. To be more concrete, and without loss of generality, say that Agent 2 (weakly) prefers piece $k$, Agent 3 (weakly) prefers piece $k'$, and Agent 4 (weakly) prefers both piece $k$ and $k'$. Now, we claim that no agent in $\{2,3,4\}$ can (weakly) prefer any piece other than $k$ or $k'$. Indeed, assume that in addition to piece $k$, Agent 2 also likes some piece $t \notin \{k,k'\}$. Then the division is envy-free: assign piece $t$ to Agent 2, piece $k$ to Agent 4, piece $k'$ to Agent 3, and the remaining piece to Agent 1. Similarly, if we instead assume that Agent 3 or Agent 4 likes some piece $t \notin \{k,k'\}$, we again can assign a desired piece to every agent, and thus the division is envy-free. Since by assumption the division is not envy-free, it must be that each of the Agents 2, 3, and 4 strictly prefers one of the pieces $k$ or $k'$ (or both) to any other piece.
\end{proof}

\begin{claim}\label{clm:condAnotB}
If a division $(P_1, P_2, P_3, P_4)$ satisfies Condition A with piece $k$, but not Condition B, and is not envy-free, then there exist two distinct agents $i,i' \in \{2,3,4\}$ that strictly prefer piece $k$ to any other piece, i.e., $v_i(P_k) > \max_{t \neq k} v_i(P_t)$ and $v_{i'}(P_k) > \max_{t \neq k} v_{i'}(P_t)$.
\end{claim}

\begin{proof}
According to Condition A, the piece $k$ satisfies (i) Agent 1 (weakly) prefers any of the three pieces in $\{1,2,3,4\} \setminus \{k\}$, and (ii) there exist two (distinct) agents $i,i' \in \{2,3,4\}$ who (weakly) prefer piece $k$. Now, because the division is not envy-free and does not satisfy Condition B, we claim that we can strengthen point (ii) to say that there exist two distinct agents $i,i' \in \{2,3,4\}$ who \emph{strictly} prefer piece $k$ to any of the other pieces. Indeed, assume that this is not the case. Then, Agents 2, 3, and 4 have to satisfy the following pattern. One agent, say Agent 2, (weakly) prefers piece $k$. Another agent, say Agent 3, (weakly) prefers both piece $k$ and another piece $t$. The remaining agent, Agent 4, (weakly) prefers a piece $t'$, that is different from $k$. (These agents can of course also weakly prefer some additional pieces.) Now, if $t \neq t'$, then the division is envy-free: assign piece $k$ to Agent 2, piece $t$ to Agent 3, piece $t'$ to Agent 4, and the remaining piece to Agent 1. If, on the other hand, $t = t'$, then the division satisfies Condition B with pieces $k$ and $t$. Thus, we obtain a contradiction in both cases, and it must be that there exist two distinct agents $i,i' \in \{2,3,4\}$ who \emph{strictly} prefer piece $k$ to any of the other pieces.
\end{proof}

We are now ready to proceed with the proof of the lemma.

\begin{proof}[Proof of \cref{lem:continuous-path}]
We begin by introducing some notation that will be used in this proof, but also later in the sections presenting the algorithms.
We let $\equitwo \in (0,1)$ be the value of the pieces in the equipartition into two parts according to Agent 1. Similarly, we let $\equithree \in (0,\equitwo)$ and $\equifour \in (0,\equithree)$ be the value of the pieces in the equipartition into three or four parts according to Agent 1, respectively. Note that Condition A can only (possibly) hold at values $\alpha \in [\equifour, \equithree)$, and Condition B can only (possibly) hold at values $\alpha \in [\equifour, \equitwo)$.

\paragraph{\textbf{Trails.}}
Next, we define ``trails'' that our continuous path will follow. A trail is itself a continuous path that satisfies some properties, and our final continuous path will be following one of these trails at any given time, but will sometimes choose to switch from one trail to another, when they intersect. The trails, which we define in detail in the next paragraph, are a succinct way of formalizing the intuition provided above for how the cuts should move in order to try to maintain the invariant, namely to ensure that the division remains critical.

For any choice of two distinct pieces $k, k' \in \{1,2,3,4\}$ with $k < k'$ (i.e., piece $k$ lies on the left of piece $k'$), we define the trail $\gamma_{k,k'}$ to be the function that maps any $\alpha \in [\equifour,\equitwo]$ to the division $(P_1,P_2,P_3,P_4)$ that is obtained as follows:
\begin{enumerate}
\item Require that $v_1(P_t) = \alpha$ for all $t \in \{1,2,3,4\} \setminus \{k,k'\}$. Given that requirement, the division is uniquely determined by the position $x$ of the cut at the right end of piece $P_k$. This position $x$ is then chosen as follows.
\item For each agent $i \in \{2,3,4\}$, let $x_i$ be the (unique) position of that cut that ensures $v_i(P_k) = v_i(P_{k'})$. Let $x_M$ be the median of the three positions $x_2, x_3,$ and $x_4$.
\item If $\alpha \leq \equithree$, let $x_L$ be the (unique) position of the cut that ensures $v_1(P_{k'}) = \alpha$, and $x_R$ be the (unique) position that ensures $v_1(P_k) = \alpha$.
\item If $\alpha > \equithree$, then simply let $x_L := - \infty$ and $x_R := +\infty$.
\item Finally, in both cases, set $x$ to be the median of $x_L, x_M,$ and $x_R$.
\end{enumerate}
The division thus obtained is always well-defined, i.e., it exists and is unique. Furthermore, the trail $\gamma_{k,k'}$ is continuous, namely the division above depends continuously on $\alpha$. Indeed, it is easy to see that the positions $x_2(\alpha), x_3(\alpha), x_4(\alpha)$, and thus $x_M(\alpha)$, are continuous in $\alpha$. In addition, $x_L(\alpha)$ and $x_R(\alpha)$ are continuous over $[\equifour,\equithree)$ and over $(\equithree,\equitwo]$. Finally, the continuity of $x(\alpha)$ follows from the fact that at $\alpha = \equithree$, we necessarily have $x_M \in [x_L,x_R]$ and thus $x = x_M$.

An important observation is that any division that is critical must lie on one of these trails. In more detail, we have the following equivalent formulations of the two conditions:
\begin{description}
\item[Condition A:] Division $(P_1,P_2,P_3,P_4)$ satisfies Condition A if there exist $\alpha \in [\equifour,\equitwo]$ and pieces $k < k'$ such that $(P_1,P_2,P_3,P_4) = \gamma_{k,k'}(\alpha)$ and additionally there exist two distinct agents $i,i' \in \{2,3,4\}$ such that either (a) $x_M(\alpha) \leq x_L(\alpha)$, $v_i(P_k) = \max_t v_i(P_t)$ and $v_{i'}(P_k) = \max_t v_{i'}(P_t)$, or (b) $x_M(\alpha) \geq x_R(\alpha)$, $v_i(P_{k'}) = \max_t v_i(P_t)$ and $v_{i'}(P_{k'}) = \max_t v_{i'}(P_t)$.
\item[Condition B:] Division $(P_1,P_2,P_3,P_4)$ satisfies Condition B if there exist $\alpha \in [\equifour,\equitwo]$ and pieces $k < k'$ such that $(P_1,P_2,P_3,P_4) = \gamma_{k,k'}(\alpha)$ and additionally $x_M(\alpha) \in [x_L(\alpha), x_R(\alpha)]$ and $\max\{v_i(P_k), v_i(P_{k'})\} = \max_t v_i(P_t)$ for all agents $i \in \{2,3,4\}$.
\end{description}
Here the positions $x_M(\alpha), x_L(\alpha), x_R(\alpha)$ refer to the corresponding positions on the trails $\gamma_{k,k'}$ that are under consideration. Using these formulations of the two conditions, we obtain the following characterization of criticality.

\begin{observation}\label{obs:invariant}
A division $(P_1,P_2,P_3,P_4)$ is critical if and only if there exist $\alpha \in [\equifour,\equitwo]$ and pieces $k < k'$ such that $(P_1,P_2,P_3,P_4) = \gamma_{k,k'}(\alpha)$ and $\max\{v_i(P_k), v_i(P_{k'})\} = \max_t v_i(P_t)$ for all agents $i \in \{2,3,4\}$.
\end{observation}

The ``if'' direction can be shown by considering the three cases: (i) $x_M(\alpha) \in [x_L(\alpha), x_R(\alpha)]$, (ii) $x_M(\alpha) \leq x_L(\alpha)$, and (iii) $x_M(\alpha) \geq x_R(\alpha)$. In case (i), it is immediate that Condition B holds. In case (ii), $x_M(\alpha) \leq x_L(\alpha)$ implies that there exist two distinct agents $i,i' \in \{2,3,4\}$ such that $v_i(P_k) \geq v_i(P_{k'})$ and $v_{i'}(P_k) \geq v_{i'}(P_{k'})$. As a result, Condition A holds. A very similar argument shows that case (iii) also satisfies Condition A. The proof of the ``only if'' direction is omitted; we will not use this direction below.

\paragraph{\textbf{Construction of the path.}}
The path $T$ starts at the division $T(\alpha_0) = (\ell_0,m_0,r_0)$ and follows a trail $\gamma_{k,k'}$ as long as the invariant holds on the trail, i.e., as long as the division remains critical. When it reaches a point such that going any further would make the division not be critical anymore, we show below that either the division is envy-free, or the path can follow a different trail on which the division remains critical. In more detail, below we prove the following crucial property:
\begin{description}
\item[Path Extension Property:] If a division is critical but not envy-free, then there exists a trail $\gamma_{k,k'}$ passing through the division (at some value $\alpha$) that we can follow, i.e., there exists $\delta > 0$ such that for all $\alpha' \in [\alpha, \alpha + \delta]$ the division $\gamma_{k,k'}(\alpha')$ remains critical.
\end{description}
This means that as long as no envy-free division is encountered, the continuous path $T$ can always be extended while remaining on critical divisions. Since no trail ends in a critical division, the path has to stop before that. Let $\valstar$ denote the highest value attained by the path. In theory, the path could approach $\valstar$ without ever attaining it. Namely, the divisions could be critical on the path at values $[\alpha_0,\valstar)$, but not at $\valstar$. However, this is in fact impossible, since the set of critical divisions is closed. This is because it is a finite union of sets defined by non-strict inequalities on the values of various pieces, and the valuation functions are continuous. Thus, if the divisions on the path at values $[\alpha_0,\valstar)$ are critical, then so is the division at $\valstar$. As a result, the end of the path is closed, and $T(\valstar)$ must be envy-free.

\paragraph{\textbf{Proof of the Path Extension Property.}}
Let $(P_1,P_2,P_3,P_4)$ be a division that is critical, but not envy-free. Consider first the case where the division satisfies Condition B with pieces $k<k'$. In particular, the division lies on the trail $\gamma_{k,k'}$ at some value $\alpha$. By \cref{clm:condB} we have that Agents 2, 3, and 4 strictly prefer one of the pieces $k$ or $k'$ (or both) to any other piece. By continuity of the trail and the valuations, it follows that there exists $\delta > 0$ such that this remains the case on the trail $\gamma_{k,k'}(\alpha')$ for all $\alpha' \in [\alpha, \alpha + \delta]$. Thus, by \cref{obs:invariant} these divisions are also critical.

Now consider the case where the division satisfies Condition A with piece $k$, but not Condition B. Since we assumed that the division is not envy-free, \cref{clm:condAnotB} implies that there exist two distinct agents $i,i' \in \{2,3,4\}$ that strictly prefer piece $k$ to any other piece. Assuming for now that $k$ is not the rightmost piece, the division lies on a trail $\gamma_{k,k'}$ at some value $\alpha$, where $k' > k$ is an arbitrary piece lying to the right of piece $k$. By continuity of the trail and the valuations, there exists $\delta > 0$ such that $i$ and $i'$ strictly prefer piece $k$ in the division $\gamma_{k,k'}(\alpha')$ for all $\alpha' \in [\alpha, \alpha + \delta]$. In particular, agents $i,i'$ strictly prefer piece $k$ over piece $k'$, which implies that $x_M(\alpha') < x(\alpha')$ and thus $x_M(\alpha') < x_L(\alpha')$ on the trail $\gamma_{k,k'}(\alpha')$ for all $\alpha' \in [\alpha, \alpha + \delta]$. As a result, $\gamma_{k,k'}(\alpha')$ satisfies Condition A for all $\alpha' \in [\alpha, \alpha + \delta]$. In the case where $k$ is the rightmost piece, we pick an arbitrary piece $k' < k$ and apply the same arguments to $\gamma_{k',k}$. In that case, we obtain $x_M(\alpha') > x_R(\alpha')$ instead of $x_M(\alpha') < x_L(\alpha')$, and the same conclusion follows.

With the Path Extension Property now proved, the Lemma follows.
\end{proof}

\begin{remark}\label{rem:old-trails}
In a previous version of this paper, the proof of \cref{lem:continuous-path} used a different set of trails, which were closer to the intuition provided earlier in this section, but which also made the proof longer and more involved.\footnote{We thank an anonymous reviewer for suggesting the trails that are now used in the proof.} Nevertheless, these trails had the advantage of being very simple and so we will sometimes use them in subsequent sections. We include their definition below.

For any piece $k \in \{1,2,3,4\}$, we define the trail $\gamma_k^A$ that maps any $\alpha \in [\equifour,\equithree]$ to the unique division $(P_1,P_2,P_3,P_4)$ that satisfies $v_1(P_t) = \alpha$ for all $t \in \{1,2,3,4\} \setminus \{k\}$.

For any agent $i \in \{2,3,4\}$ and any two distinct pieces $k,k' \in \{1,2,3,4\}$, we define the trail $\gamma_{i,k,k'}^B$ that maps any $\alpha \in [\equifour,\equitwo]$ to the unique division $(P_1,P_2,P_3,P_4)$ that satisfies $v_i(P_k) = v_i(P_{k'})$, and $v_1(P_t) = \alpha$ for all $t \in \{1,2,3,4\} \setminus \{k,k'\}$.

It is easy to check that the trails used in the proof are a refinement of these trails, namely they always follow one of these trails and will sometimes switch from following one trail to another.
\end{remark}

\section{An Efficient Algorithm for Four Monotone Agents}\label{sec:algo}

In this section we prove the following result.

\begin{theorem}[Efficient algorithm for four monotone agents]
For four agents with monotone 1-Lipschitz valuations, we can compute an $\eps$-envy-free connected allocation using $O(\log^3(1/\eps))$ value queries.
\end{theorem}

The crucial observation that allows us to obtain an efficient algorithm is that we can perform binary search on the continuous path whose existence is guaranteed by \cref{lem:continuous-path}. It is straightforward to check that the algorithm also yields a polynomial-time algorithm in the standard Turing machine model.

\subsection{Preprocessing}\label{sec:algo-prepro}

In this section we show that, without loss of generality, we can assume that:
\begin{itemize}
    \item $\eps$ is the inverse of an integer, i.e., $\eps = 1/m$ for some integer $m \geq 2$.
    \item The valuations $v_i$ are $\eps$-strongly-hungry, meaning that if we extend an interval by some length $t$, then its value increases by at least $t\eps$.
    \item The valuations $v_i$ are piecewise linear on the $\eps$-grid $\{0,\eps,2\eps,\dots,1-\eps,1\}$, as defined below.
\end{itemize}
Thus, in the next section it will suffice to provide an algorithm that finds an $\eps$-envy-free allocation for $1$-Lipschitz valuations that are $\eps$-strongly-hungry and piecewise linear on the $\eps$-grid, whenever $\eps$ is the inverse of an integer.

\begin{definition}
Let $m \in \mathbb{N}$ and $x_i := i \cdot \delta$ for all $i \in \{0, 1, \dots, m\}$. A valuation function $v$ is \emph{piecewise linear on the $1/m$-grid} if it is continuous over its entire domain and (affine) linear in each of the triangles of the following triangulation of $\{(a,b) \in [0,1]^2: a \leq b\}$:
\begin{itemize}
    \item $T_{i,j}^\leq := \{(a,b) : x_{i-1} \leq a \leq x_i, x_{j-1} \leq b \leq x_j, a - x_{i-1} \leq b - x_{j-1}\}$ $\forall i,j \in [m]$ with $i \leq j$,
    \item $T_{i,j}^\geq := \{(a,b) : x_{i-1} \leq a \leq x_i, x_{j-1} \leq b \leq x_j, a - x_{i-1} \geq b - x_{j-1}\}$ $\forall i,j \in [m]$ with $i < j$.
\end{itemize}
\end{definition}

The piecewise linear structure on the $\eps$-grid allows us to answer various types of queries exactly. For example, we can simulate exact cut queries: given a position $a$ on the cake, and a desired value $\val$, we can find the unique position $b$ such that $v(a,b) = \val$ (or output that it does not exist). To do this, we perform binary search on the $\eps$-grid for $O(\log(1/\eps))$ steps to determine the interval on the $\eps$-grid in which the cut $b$ must lie. Then, using a constant number of queries we obtain full information about $v(a,b')$ for any $b'$ in that interval, and can thus determine the exact value of $b$. Similarly, given two cuts $a,c$, we can determine the exact position of cut $b \in [a,c]$ such that $v(a,b) = v(b,c)$. More generally, if we are looking for multiple cuts satisfying some property, it suffices to determine for each cut the interval on the $\eps$-grid in which it must lie, and then with a constant number of queries we can determine the exact positions of the cuts. We will make extensive use of this in the algorithm, which is described in the next section.

The next lemma is the main tool used to show that the restrictions are without loss of generality. A simpler version of this construction was used by Br{\^a}nzei and Nisan~\cite{BranzeiN19-cake-communication} for additive valuation functions.

\begin{lemma}\label{lem:linear-interpolation}
Let $m \in \mathbb{N}$ and let $v$ be a monotone $1$-Lipschitz valuation function. Then there exists a monotone $1$-Lipschitz valuation function $\widetilde{v}$ that satisfies:
\begin{enumerate}
    \item $\widetilde{v}$ is a $3/m$-approximation of $v$, i.e., $|v(a,b) - \widetilde{v}(a,b)| \leq 3/m$ for all $0 \leq a \leq b \leq 1$.
    \item $\widetilde{v}$ is $1/m$-strongly-hungry, i.e., $\widetilde{v}(a',b') \geq \widetilde{v}(a,b) + (b'-b)/m + (a-a')/m$ for all $a,b,a',b' \in [0,1]$ with $a'\leq a \leq b \leq b'$.
    \item Any value query to $\widetilde{v}$ can be answered by making a constant number of value queries to $v$.
    \item $v$ is piecewise linear on the $1/m$-grid.
\end{enumerate}
\end{lemma}

\begin{proof}
First, define $v'$ by letting $v'(a,b) = v(a,b)\cdot (m-1)/m + |b-a|/m \in [0,1]$. Note that $v'$ is $1$-Lipschitz-continuous and $1/m$-strongly-hungry. Furthermore, $v'$ is a $1/m$-approximation of $v$ and any value query to $v'$ can be answered by making a single value query to $v$.

Next, define $\widetilde{v}$ to be the piecewise linear interpolation of $v'$ on the $1/m$-grid. In other words, the valuation $\widetilde{v}$ is piecewise linear on the $1/m$-grid and agrees with $v'$ at the grid points, i.e., $\widetilde{v}(i/m,j/m) = v'(i/m,j/m)$ for all $i,j \in \{0, 1, \dots, m\}$ with $i \leq j$.

More formally, for any $a,b \in [0,1]$ with $a \leq b$, let $\underline{a},\overline{a}$ be consecutive multiples of $1/m$ such that $\underline{a} \leq a \leq \overline{a}$, and let $\underline{b},\overline{b}$ be consecutive multiples of $1/m$ such that $\underline{b} \leq b \leq \overline{b}$. Then, when $b - \underline{b} \geq a - \underline{a}$, we let
$$\widetilde{v}(a,b) = \frac{a-\underline{a}}{1/m} \cdot v'(\overline{a},\overline{b}) + \frac{\overline{b}-b}{1/m} \cdot v'(\underline{a},\underline{b}) + \frac{(b-\underline{b})-(a-\underline{a})}{1/m} \cdot v'(\underline{a},\overline{b})$$
and, when $b - \underline{b} \leq a - \underline{a}$, we let
$$\widetilde{v}(a,b) = \frac{\overline{a} - a}{1/m} \cdot v'(\underline{a},\underline{b}) + \frac{b - \underline{b}}{1/m} \cdot v'(\overline{a},\overline{b}) + \frac{(a - \underline{a}) - (b - \underline{b})}{1/m} \cdot v'(\overline{a},\underline{b}).$$
Clearly, a value query to $\widetilde{v}$ can be answered by making three value queries to $v'$, since it suffices to query the three vertices of the triangle containing the queried point $(a,b)$. Furthermore, since $v'$ is $1$-Lipschitz-continuous, so is $\widetilde{v}$. In particular, it follows that $\widetilde{v}$ is a $2/m$-approximation of $v'$, and thus a $3/m$-approximation of $v$.

It remains to argue that $\widetilde{v}$ is $1/m$-strongly-hungry. In other words, we want to show that for any $0 \leq a' \leq a \leq b \leq b' \leq 1$, we have $\widetilde{v}(a',b') \geq \widetilde{v}(a,b) + (b'-b)/m + (a-a')/m$. It suffices to prove that this is the case whenever $(a,b)$ and $(a',b')$ lie in the same triangle of the triangulation. Indeed, if the condition holds within each triangle, then it also holds globally. This is easy to see by considering the straight path connecting $(a,b)$ to $(a',b')$ and using the fact that the statement holds for each portion of the path that lies within a triangle.

If $(a,b)$ and $(a',b')$ lie in the same triangle, then there exist consecutive multiples $\underline{a},\overline{a}$ of $1/m$ such that $a,a' \in [\underline{a}, \overline{a}]$. Similarly, there exist consecutive multiples $\underline{b},\overline{b}$ of $1/m$ such that $b,b' \in [\underline{b}, \overline{b}]$. Consider first the case where $b - \underline{b} \geq a - \underline{a}$. Since $(a,b)$ and $(a',b')$ lie in the same triangle, we must also have $b' - \underline{b} \geq a' - \underline{a}$. We can thus write
\begin{equation*}
\begin{split}
\widetilde{v}(a',b') - \widetilde{v}(a,b) &= \frac{a'-a}{1/m} \cdot v'(\overline{a},\overline{b}) + \frac{b-b'}{1/m} \cdot v'(\underline{a},\underline{b}) + \frac{(b'-b) - (a'-a)}{1/m} \cdot v'(\underline{a},\overline{b})\\
&= (b'-b) \cdot m \cdot (v'(\underline{a},\overline{b}) - v'(\underline{a},\underline{b})) + (a-a') \cdot m \cdot (v'(\underline{a},\overline{b}) - v'(\overline{a},\overline{b}))\\
&\geq (b'-b)/m + (a-a')/m
\end{split}
\end{equation*}
where we used the fact that $v'$ is $1/m$-strongly-hungry, and thus $v'(\underline{a},\overline{b}) - v'(\underline{a},\underline{b}) \geq (\overline{b} - \underline{b})/m = 1/m^2$ and $v'(\underline{a},\overline{b}) - v'(\overline{a},\overline{b}) \geq (\overline{a} -\underline{a})/m = 1/m^2$. The case where $b - \underline{b} \leq a - \underline{a}$ (and thus $b' - \underline{b} \leq a' - \underline{a}$) is handled analogously.
\end{proof}

Now assume that we have an algorithm that solves the problem using $O(\log^3(1/\eps))$ queries under the restrictions mentioned above. We show how to obtain an algorithm for the problem without these restrictions. Let $\eps \in (0,1)$ be given and let $m \in \mathbb{N}$ be the smallest integer such that $\eps' := 1/m$ satisfies $\eps' \leq \eps/7$. Note that $\eps' \geq \eps/8$. By \cref{lem:linear-interpolation}, we can run the algorithm on the modified valuations $\widetilde{v}_i$ and obtain an allocation that is $\eps'$-envy-free with respect to the modified valuations. Since $\widetilde{v}_i$ is a $3\eps'$-approximation of $v_i$, the allocation is $7\eps'$-envy-free with respect to the original valuations, i.e., $\eps$-envy-free. The algorithm made $O(\log^3(1/\eps'))$ queries to the modified valuations, which corresponds to $O(\log^3(8/\eps))$ queries to the original valuations by \cref{lem:linear-interpolation}.

\subsection{The Algorithm}

Consider four agents with valuation functions $v_1,v_2,v_3,v_4$ that are monotone and $1$-Lipschitz continuous. In this section we present an algorithm that computes an $\eps$-envy-free allocation using $O(\log^3(1/\eps))$ value queries. Using the preprocessing construction presented in the previous section, we also assume that the valuations are $\eps$-strongly-hungry and piecewise linear on the $\eps$-grid.

\paragraph*{\textbf{Critical divisions.}} We briefly recall the following definition. A division of the cake into four pieces is \emph{critical} if at least one of the following two conditions holds:
\begin{description}
\item[Condition A:] Agent 1 is indifferent between its three favorite pieces, and the remaining piece is (weakly) preferred by (at least) two of the three other agents.
\item[Condition B:] Agent 1 is indifferent between its two favorite pieces, and the two remaining pieces are each (weakly) preferred by (at least) two of the three other agents.
\end{description}
We say that ``Condition A (resp.~B) holds at value $\val$'' if there exists a division of the cake for which Condition A (resp.~B) holds, and in which Agent 1 has value $\val$ for its favorite pieces. Similarly, we say that ``there exists a critical division at value $\alpha$'', if there exists a critical division where Agent 1 has value $\alpha$ for its favorite pieces. For a more formal definition of the two conditions, see the previous section.

\paragraph*{\textbf{The Algorithm.}}

\begin{enumerate}
    \item Compute the (unique) equipartition of the cake into four equal parts according to Agent 1. If this equipartition yields an envy-free division, then return it. Otherwise, let $\equifour \in (0,1)$ be the value that Agent 1 has for each of the equal parts. Set $\vallow := \equifour$ and $\valhigh := v_1(0,1) \leq 1$.
    \item Repeat until $|\valhigh - \vallow| \leq \eps^4/12$:
    \begin{enumerate}
        \item Let $\val := (\vallow+\valhigh)/2$.
        \item If there exists a critical division at value $\val$, then set $\vallow := \val$. Otherwise, set $\valhigh := \val$.
    \end{enumerate}
    \item Return a critical division at value $\vallow$.
\end{enumerate}

We begin by explaining how each step can be implemented efficiently using value queries. Then, we argue about the correctness of the algorithm.

\subsection{Implementation using value queries}

\begin{lemma}\label{lem:algo-queries}
The algorithm can be implemented using $O(\log^3(1/\eps))$ value queries.
\end{lemma}

\begin{proof}
The loop in Step 2 runs for $\log(12/\eps^4) = O(\log(1/\eps))$ iterations. Thus, in order to prove the lemma, it suffices to prove that:
\begin{enumerate}
    \item[I.] We can find the (unique) equipartition of the cake into four equal parts according to Agent 1 using $O(\log^2(1/\eps))$ value queries.
    \item[II.] Given $\val \in [0,1]$, we can check whether there exists a critical division at value $\val$ using $O(\log^2(1/\eps))$ value queries, and, if so, output such a critical division.
\end{enumerate}
Our proof crucially uses the fact that cuts satisfying various properties are unique, because the valuation functions are hungry, and the fact that they can be found by binary search due to the monotonicity of the valuations and the preprocessing.

\paragraph*{\textbf{Proof of I.}}
For any position of the leftmost cut $\ell$, we can determine using $O(\log(1/\eps))$ value queries the exact position of the cut $m$, and then $r$, such that the three leftmost pieces have identical value for Agent 1. With one additional value query we can then check if the rightmost piece has a higher or a lower value for Agent 1 than the other three pieces. By the monotonicity of the valuation function, this comparison is monotone with respect to the position of cut $\ell$. As a result, after at most $\log(1/\eps)$ steps of binary search on the $\eps$-grid (where each step requires $O(\log(1/\eps))$ value queries), we can determine in which interval of the $\eps$-grid the cut $\ell$ must lie in the equipartition. A symmetric approach also allows us to determine the interval of the $\eps$-grid in which the cut $r$ must lie in the equipartition.

For any position $m$ of the middle cut, we can determine using $O(\log(1/\eps))$ value queries the exact position of the cut $\ell$ such that the two leftmost pieces have identical value for Agent 1. Similarly, we can also find the exact position of the cut $r$ such that the two rightmost pieces have identical value of Agent 1. We can then check if the rightmost piece has a higher or lower value for Agent 1 than the leftmost piece. By the monotonicity of the valuation function, this comparison is monotone with respect to the position of cut $m$. Thus, we can perform binary search to locate the interval on the $\eps$-grid in which the cut $m$ must lie in the equipartition.

Once we have determined the interval on the $\eps$-grid in which each cut must lie, we can determine the exact positions of the cuts using a constant number of value queries.

\paragraph*{\textbf{Proof of II.}}

Given a value $\val \in [0,1]$, we can check whether Condition A or Condition B holds at value $\val$ as follows.

\textbf{Condition A.}
For each of the four possible choices for the identity of the piece $k$ that is not (necessarily) favored by Agent 1, we can determine the exact (and unique) positions of the cuts that ensure that Agent 1 has value exactly $\val$ for all pieces except (possibly) $k$. For example, if piece $k$ is the rightmost piece, then we first determine the position of $\ell$ such that $v_1(0,\ell) = \val$, then the position of $m$ such that $v_1(\ell,m) = \val$, and finally the position of $r$ such that $v_1(m,r) = \val$. Depending on the identity of piece $k$, we might have to compute the cuts in some other order, but the cuts will always be unique and can be found using $O(\log(1/\eps))$ value queries.

Once we have determined the positions of the cuts, it then suffices to check that Agent 1 has value at most $\val$ for piece $k$, and that at least two of the other agents weakly prefer piece $k$. If this check succeeds for at least one choice of piece $k$, then Condition A holds at value $\val$ and we can output a division that satisfies it. If the check fails for all four possible choices of piece $k$, then Condition A does not hold at value $\val$.

\textbf{Condition B.}
We proceed similarly to the procedure for Condition A above. We consider all possible choices for the identity of the two pieces $k,k'$ not (necessarily) favored by Agent 1, and all possible choices for the identity of the agent $i \in \{2,3,4\}$ who is indifferent between pieces $k$ and $k'$. Indeed, observe that if Condition B holds, then one of the three players 2,3,4 must be indifferent between pieces $k$ and $k'$.

Once we have fixed the identities of pieces $k,k'$ and agent $i$, the positions of the cuts are uniquely determined and, as we show below, can be found using at most $O(\log^2(1/\eps))$ value queries. We can then check if the resulting division satisfies Condition B or not.

The unique positions of the cuts can be found as follows. There are three cases:
\begin{itemize}
    \item The pieces $k$ and $k'$ are adjacent. Consider the representative example where pieces $k$ and $k'$ are the two middle pieces. The other instantiations of this case are handled analogously. Using $O(\log(1/\eps))$ value queries, we can determine the positions of cuts $\ell$ and $r$ such that $v_1(0,\ell) = \val$ and $v_1(r,1) = \val$. Then, using $O(\log(1/\eps))$ value queries, we can find the cut $m$ that ensures that $v_i(\ell,m) = v_i(m,r)$.
    \item There is one piece between pieces $k$ and $k'$. Consider the representative example where piece $k$ is the leftmost piece $[0,\ell]$ and piece $k'$ is the third piece from the left $[m,r]$. Using $O(\log(1/\eps))$ value queries we first determine the exact position of cut $r$ such that $v_1(r,1) = \val$. For each possible position of cut $\ell$, let $m(\ell)$ denote the unique position of the middle cut that satisfies $v_1(\ell,m(\ell)) = \val$. Note that, given $\ell$, we can find $m(\ell)$ using $O(\log(1/\eps))$ value queries. The crucial observation is that as we move $\ell$ to the right, $m(\ell)$ must also move to the right, due to the monotonicity of Agent 1. As a result, when we move $\ell$ to the right, piece $[0,\ell]$ becomes larger (i.e., a superset of its previous state), and piece $[m(\ell),r]$ becomes smaller (i.e., a subset of its previous state). But this means that we can use binary search on the $\eps$-grid to find the interval on the $\eps$-grid which contains the cut $\ell$ that satisfies $v_i(0,\ell)=v_i(m(\ell),r)$. This uses $O(\log^2(1/\eps))$ queries, since we perform binary search for $O(\log(1/\eps))$ steps and each step requires the computation of $m(\ell)$ given $\ell$, which takes $O(\log(1/\eps))$ value queries. By an analogous argument, namely by letting $\ell(m)$ be a function of $m$ such that $v_1(\ell(m),m) = \val$, we can also determine the interval on the $\eps$-grid which must contain cut $m$, again using $O(\log^2(1/\eps))$ value queries. Finally, we now have all the information needed to find the exact positions of cuts $\ell$ and $m$.
    \item Piece $k$ is the leftmost piece $[0,\ell]$ and piece $k'$ is the rightmost piece $[r,1]$. This case is handled analogously to the previous one. Given cut $\ell$, using $O(\log(1/\eps))$ value queries, we can determine cuts $m(\ell)$ and $r(\ell)$ such that $v_1(\ell,m(\ell)) = \val$ and $v_1(m(\ell),r(\ell)) = \val$. As a result, by monotonicity we can use $O(\log^2(1/\eps))$ value queries to determine the interval on the $\eps$-grid in which cut $\ell$ must lie to satisfy $v_i(0,\ell) = v_i(r(\ell),1)$. Similarly, we can determine the interval in which $r$ must lie. Finally, given a position for cut $m$, we can find cuts $\ell(m)$ and $r(m)$ such that $v_1(\ell(m),m) = \val$ and $v_1(m,r(m)) = \val$ using $O(\log(1/\eps))$ value queries. Again, using $O(\log^2(1/\eps))$ value queries, we can perform binary search to find the interval on the $\eps$-grid that must contain cut $m$. Then, we have all the information that is needed to find the exact positions of the cuts.
\end{itemize}
This completes the proof of the lemma.
\end{proof}

\subsection{Proof of Correctness}

In this section, we argue about the correctness of the algorithm. First of all, note that if the equipartition computed in Step 1 of the algorithm is an envy-free division, then the algorithm's output is correct. If it is not an envy-free division, then Condition A or B must hold, and thus it is a critical division at value $\equifour$. Furthermore, it is easy to see that there is no critical division at value $v_1(0,1)$, since the valuations are hungry. As a result, the binary search procedure performed in Step 2 must return $\vallow$ and $\valhigh$ such that: (i) there exists a critical division at value $\vallow$, and (ii) there is no critical division at value $\valhigh$. Furthermore, we have $\vallow < \valhigh$ and $|\valhigh-\vallow| \leq \eps^4/12$.

The algorithm returns a critical division $(\ell_0,m_0,r_0)$ where Agent 1 has value $\vallow$ for its favorite pieces. Since this division is critical, we can apply \cref{lem:continuous-path} to obtain a continuous path $T$ that starts at the division and ends at an envy-free division $T(\valstar)$. Using the fact that all divisions on the path are critical, and that there is no critical division at value $\valhigh$, it must be that $\valstar < \valhigh$. As a result, $|\valstar-\vallow| \leq \eps^4/12$. Finally, as we show below, the path $T$ is $6/\eps^3$-Lipschitz-continuous, and thus the division $(\ell_0,m_0,r_0)$ is $\eps/2$-close to the division $T(\valstar)$, which is envy-free. As a result, given that the valuations are $1$-Lipschitz-continuous, the division $(\ell_0,m_0,r_0)$ must be $\eps$-envy-free.

The Lipschitz-continuity of the continuous path $T$ follows from the following lemma. Recall that the path follows trails of type $\gamma_k^A$ and $\gamma_{i,k,k'}^B$ (see \cref{rem:old-trails}).

\begin{lemma}
The trails $\gamma_k^A$ and $\gamma_{i,k,k'}^B$ are Lipschitz-continuous (w.r.t.~the $\ell_1$-norm) with Lipschitz constant $6/\eps^3$.
\end{lemma}

\begin{proof}
We consider the two types of trails separately. In each case, we show that for any $\alpha$ and infinitesimal $\dval$, the cuts in division $\gamma(\alpha+\dval)$ are $6 \dval/\eps^3$-close to the cuts in division $\gamma(\alpha)$. Since the valuations are piecewise-linear, this is sufficient to deduce the statement of the lemma.

\paragraph*{\textbf{Trail of type A.}}
Consider the division $(\ell,m,r) := \gamma_k^A(\alpha)$, i.e., Agent 1 has value $\alpha$ for all pieces except piece $k$. In going from $\gamma_k^A(\alpha)$ to $\gamma_k^A(\alpha+\dval)$ the cuts move continuously as follows.
\begin{itemize}
\item If piece $k$ is the rightmost (i.e., $[r,1]$), the leftmost cut $\ell$ moves to the right by $d\ell$ such that $v_1(0,\ell+d\ell)=\val+\dval$. Simultaneously, the other two cuts move to the right, namely $m$ by $dm$ and $r$ by $dr$, such that $v_1(\ell+d\ell,m+dm)=\val+\dval$ and $v_1(m+dm,r+dr)=\val+\dval$. Since $v_1$ is $\eps$-strongly-hungry, we have that $\dval = v_1(0,\ell+d\ell) - v_1(0,\ell) \geq \eps \cdot d\ell$, which implies $d\ell \leq \dval/\eps$. Then, since $v_1$ is also $1$-Lipschitz-continuous, we have $\dval = v_1(\ell+d\ell,m+dm) - v_1(\ell,m) \geq \eps \cdot dm - 1 \cdot d\ell$, and thus $dm \leq \dval/\eps + d\ell/\eps \leq 2\dval/\eps^2$. Similarly, we also obtain that $dr \leq 3\dval/\eps^3$. Thus, the cuts collectively move by at most $6 \dval/\eps^3$. The case where piece $k$ is $[0,\ell]$ is symmetric.
\item If piece $k$ is the second from the right (i.e., $[m,r]$), the cuts $\ell$ and $m$ move as before such that $v_1(0,\ell+d\ell)=\val+\dval$ and $v_1(\ell+d\ell,m+dm)=\val+\dval$. Simultaneously, cut $r$ moves to the left such that $v_1(r-dr,1)=\val+\dval$. Using the same arguments as above, the total movement of the cuts can be bounded by $4\dval/\eps^2$. The case where piece $k$ is $[\ell,m]$ is symmetric.
\end{itemize}

\paragraph*{\textbf{Trail of type B.}}
Consider the division $(\ell,m,r) := \gamma_{i,k,k'}^B(\alpha)$, i.e., Agent 1 has value $\alpha$ for the pieces $p$ and $q$ ($\{p,q\} = \{1,2,3,4\} \setminus \{k,k'\}$), and Agent $i$ is indifferent between pieces $k$ and $k'$. In going from $\gamma_{i,k,k'}^B(\alpha)$ to $\gamma_{i,k,k'}^B(\alpha+\dval)$ the cuts move continuously as follows.
\begin{itemize}
\item If pieces $p$ and $q$ are the two leftmost pieces ($[0,\ell]$ and $[\ell,m]$), $\ell$ and $m$ move to the right, such that $v_1(0,\ell+d\ell) = \val + \dval$ and $v_1(\ell+d\ell,m+dm) = \val + \dval$. Simultaneously, $r$ moves to the right such that $v_i(m+dm,r+dr) = v_i(r+dr,1)$. Note that $r$ indeed needs to be shifted to the right (and not to the left) because of the monotonicity of $v_i$. As above, we have $d\ell \leq \dval/\eps$ and $dm \leq 2\dval/\eps^2$. Furthermore, since the third piece's value decreased by at most $1 \cdot dm$ for Agent $i$ (before moving cut $r$), it follows that cut $r$ must move by at most $dm/\eps$, i.e., $dr \leq dm/\eps \leq 2\dval/\eps^3$. Thus, the cuts collectively move by at most $5\dval/\eps^3$. The case where pieces $p$ and $q$ are the two rightmost pieces ($[m,r]$ and $[r,1]$) is symmetric.
\item If pieces $p$ and $q$ are $[0,\ell]$ and $[r,1]$, cut $\ell$ moves to the right, and cut $r$ to the left, while maintaining Agent 1's indifference, i.e., such that $v_1(0,\ell+d\ell) = \val + \dval$ and $v_1(r-dr,1) = \val+\dval$. Cut $m$ moves to the unique position that makes Agent $i$ indifferent between the two middle pieces, i.e., such that $v_i(\ell+d\ell,m \pm dm) = v_i(m \pm dm,r+dr)$. Note that, depending on the situation, $m$ might have to move to the left or to the right to satisfy the indifference condition. For this reason we denote the shifted cut by $m \pm dm$. Using similar arguments to above, we obtain that $d\ell \leq \dval/\eps$, $dr \leq \dval/\eps$, and $|dm| \leq \dval/\eps^2$, which yields that the total movement of cuts is at most $4\dval/\eps^2$.
\item If pieces $p$ and $q$ are $[0,\ell]$ and $[m,r]$, cut $\ell$ moves to the right such that $v_1(0,\ell+d\ell) = \val + \dval$. Then, cuts $m$ and $r$ simultaneously move such that $v_1(m \pm dm, r + dr) = \val + \dval$ and $v_i(\ell+d\ell,m \pm dm) = v_i(m \pm dm, r + dr)$. Note that, depending on the situation, $m$ might need to move left or right, whereas $r$ will have to move to the right by monotonicity of the valuations. Using similar arguments to above, we can show that $d\ell \leq \dval/\eps$, $dm \geq -\dval/\eps$, $dm \leq \dval/\eps^2$, and $dr \leq 2\dval/\eps^3$. The total movement of the cuts is thus at most $4\dval/\eps^3$. The case where pieces $p$ and $q$ are $[\ell,m]$ and $[r,1]$ is symmetric.
\item If pieces $p$ and $q$ are $[\ell,m]$ and $[m,r]$, then $\ell$ moves to the left and $r$ moves to the right such that Agent $i$ remains indifferent between the leftmost and rightmost piece, i.e., such that $v_i(0,\ell-d\ell) = v_i(r+dr,1)$. This specifies $dr$ as a function of $d\ell$. Next, the middle cut $m$ is shifted such that Agent 1 remains indifferent between the two middle pieces, i.e., such that $v_1(\ell-d\ell,m \pm dm) = v_1(m \pm dm,r+dr)$. This also specifies $dm$ as a function of $d\ell$. Finally, note that the value of Agent 1 for the two middle pieces strictly increases as $\ell$ is moved left (and the other cuts follow according to the above). We then also require that $v_1(\ell-d\ell,m \pm dm) = \val + \dval$, and this specifies the shifts of all cuts as functions of $\dval$. Using similar arguments to above, we can bound $|dm| \leq \dval/\eps$, and then $d\ell \leq 2\dval/\eps^2$ and $dr \leq 2\dval/\eps^2$. Thus, the total movement of the cuts is at most $5\dval/\eps^2$.
\end{itemize}
This completes the proof of the lemma.
\end{proof}

\section{An Efficient Algorithm in the Robertson-Webb Model}\label{sec:RW-algo}

In this section we show how the algorithm can also be implemented in the Robertson-Webb cake-cutting model.

\paragraph*{\textbf{Robertson-Webb query model.}}
In the Robertson-Webb cake-cutting model~\cite{WoegingerS07-cake}, the valuations are assumed to be additive, but not necessarily Lipschitz-continuous (and even if they are $L$-Lipschitz-continuous, we are not given any bound on $L$). Formally, for each valuation function $v_i: [0,1]^2 \to [0,1]$ there exists a non-atomic integrable density function $f_i: [0,1] \to [0,+\infty)$ such that for all $0 \leq a \leq b \leq 1$ we have
$$v_i(a,b) = \int_a^b f_i(x) dx.$$
Note that $v_i$ is indeed continuous because $f_i$ is non-atomic. It is also monotone. Furthermore, we assume that $v_i(0,1) = 1$.
It is easy to see that the lack of a bound on the Lipschitzness of $v_i$ (or equivalently the lack of a bound on the density $f_i$) implies that we cannot hope to solve the problem (even for two agents) using only value queries. For this reason, the Robertson-Webb model also introduces a second type of query: \emph{cut} queries. On cut query $(x,\alpha)$, where $x$ is a position on the cake and $\alpha$ is a value, the agent responds by returning a position $y$ on the cake such that $v_i(x,y) = \alpha$ (or responds that there is no such $y$). Due to the additivity, it is sufficient to only allow cut queries with $x=0$. It is also easy to see that we can simulate reverse cut queries, i.e., given $(y,\alpha)$, return $x$ such that $v_i(x,y) = \alpha$ (or respond that there is no such $x$).

In this section we prove the following result.

\begin{theorem}
For four agents with additive valuations, we can compute an $\eps$-envy-free connected allocation using $O(\log^2(1/\eps))$ value and cut queries in the Robertson-Webb model.
\end{theorem}

\subsection{Preprocessing}

We begin with some useful ``preprocessing'' which will allow us to assume that the valuations have some additional structure.

\begin{lemma}\label{lem:RW-prepro}
Let $\nquant \in \mathbb{N}$ and let $v$ be an additive continuous valuation function. Then there exists an additive continuous valuation function $\widetilde{v}$ that satisfies:
\begin{enumerate}
\item $\widetilde{v}$ is a $2/\nquant$-approximation of $v$, i.e., $|v(a,b) - \widetilde{v}(a,b)| \leq 2/\nquant$ for all $0 \leq a \leq b \leq 1$.
\item $\widetilde{v}$ is $1/\nquant$-strongly-hungry, i.e., $\widetilde{v}(a,b) \geq (b-a)/\nquant$ for all $0 \leq a \leq b \leq 1$.
\item Any value or cut query to $\widetilde{v}$ can be answered by making a constant number of value and cut queries to $v$.
\item $\widetilde{v}$ is linear between its $\nquant$-quantiles, i.e., letting $0 = x_0 < x_1< \dots < x_\nquant = 1$ be the (unique) positions such that for all $j \in [\nquant]$
$$\widetilde{v}(x_{j-1},x_j) = 1/\nquant$$
we have that for all $j \in [\nquant]$ and all $t \in [0,1]$
$$\widetilde{v}(x_{j-1},x_{j-1} + t(x_j-x_{j-1})) = t/\nquant.$$
\end{enumerate}
\end{lemma}

\begin{proof}
The $\nquant$-quantiles of $v$, i.e., positions $0 = x_0 < x_1 < \dots < x_\nquant = 1$ such that for all $j \in [\nquant]$
$$v(x_{j-1},x_j) = 1/\nquant$$
are in general not unique, because $v$ might not be hungry. For any choice $x_0, x_1, \dots, x_\nquant$ of $\nquant$-quantiles of $v$, we can define a corresponding valuation function $\widetilde{v}$ as follows.

The valuation function $\widetilde{v}$ is constructed by evening out the mass between the $\nquant$-quantiles $x_0, x_1, \dots, x_\nquant$ of $v$. Formally, $\widetilde{v}$ is given by its density function $\widetilde{f}$, which is defined as
$$\widetilde{f}(x) = \frac{1}{\nquant} \sum_{j=1}^\nquant \chi_{[x_{j-1},x_j]}(x) \cdot \frac{1}{(x_j-x_{j-1})}$$
for all $x \in [0,1]$. Here $\chi_I: [0,1] \to \{0,1\}$ denotes the characteristic function of interval $I$. Clearly, $\widetilde{v}$ is additive and continuous. Furthermore, it immediately follows from the construction that $\widetilde{v}(0,1) = 1$, that $x_0, x_1, \dots, x_\nquant$ are the (unique) $\nquant$-quantiles of $\widetilde{v}$, and that $\widetilde{v}$ is linear between its $\nquant$-quantiles. Finally, it is easy to check that $\widetilde{v}$ is $1/\nquant$-strongly-hungry and that it is a $2/\nquant$-approximation of $v$.

It remains to argue that we can simulate queries to $\widetilde{v}$ by using queries to $v$. To do this, we will fix the choice of $\nquant$-quantiles $x_0, x_1, \dots, x_\nquant$ of $v$ as follows: $x_j$ is defined as the cut position returned when we perform a cut query\footnote{A somewhat subtle detail is that performing the same cut query $(0,j/\nquant)$ to $v$ multiple times could potentially output different results (because the quantiles might not be unique). However, this is not an issue: it suffices for the simulation to memorize the answer the first time it performs the query and then to just reuse this answer for subsequent queries. This will ensure that the $x_j$'s are fixed throughout the simulation.} $(0,j/\nquant)$ to $v$. If the simulation never queries $(0,j/\nquant)$ for some $j$, then all simulated queries to $\widetilde{v}$ will be consistent with any of the possible choices for $x_j$.

First, we observe that for any $j \in [\nquant] \cup \{0\}$ we can obtain $x_j$ by using at most one cut query to $v$. Indeed, if $j = 0$ or $j=\nquant$, then we just use $x_0=0$ or $x_\nquant=1$ respectively. Otherwise, when $j \in [\nquant-1]$, we can use one cut query $(0,j/\nquant)$ to $v$ to obtain $x_j$. We can now simulate queries to $\widetilde{v}$ as follows:
\begin{itemize}
\item cut query $(0,\alpha)$: We first find $j \in [\nquant]$ such that $(j-1)/\nquant \leq \alpha \leq j/\nquant$. Then we compute $x_{j-1}$ and $x_j$ as described above and output $x_{j-1} + (x_j - x_{j-1}) (\nquant \alpha-(j-1))$.
\item value query $(0,b)$: We first perform one value query to $v$ to obtain $v(0,b) =: \alpha$. If $\alpha = 0$, then we compute $x_1$ as described above and output $b/(\nquant x_1)$. If $\alpha = 1$, then we compute $x_{\nquant-1}$ as described above and output $1 - (1-b)/(\nquant(1-x_{\nquant-1}))$. If $\alpha \in (0,1)$, then we find $j \in [\nquant-1]$ such that $(j-1)/\nquant < \alpha < (j+1)/\nquant$ and compute $x_j$. If $b \geq x_j$, then we compute $x_{j+1}$ and output $j/\nquant + (b-x_j)/(\nquant(x_{j+1}-x_j))$. If $b < x_j$, then we compute $x_{j-1}$ and output $j/\nquant - (x_j-b)/(\nquant(x_j-x_{j-1}))$.
\end{itemize}
Note that it is sufficient to only consider queries with endpoint $0$, since the valuations are additive.
\end{proof}

In the rest of this section, we will assume that the valuations $v_i$ have been obtained from the original valuations by applying \cref{lem:RW-prepro} with $m := \lceil 4/\eps \rceil$, i.e., we have replaced $v_i$ by $\widetilde{v}_i$. Thus, finding an envy-free allocation with respect to these valuations will yield an $\eps$-envy-free allocation with respect to the original valuations.

We let $x_0^i, \dots, x_\nquant^i$ denote the (unique) $\nquant$-quantiles of valuation $v_i$. The following will be used repeatedly in our algorithm.

\begin{lemma}\label{lem:binary-search}
Let $c: [t_1,t_2] \to [0,1]$ be continuous, strictly monotone, and let $c(t_1)$ and $c(t_2)$ be given. Furthermore, assume that given any $a \in \{c(t): t \in [t_1,t_2]\}$ we can find the unique $t$ such that $c(t) = a$ using at most $Q(\nquant)$ queries to the agents' valuations. Finally, let $s: [t_1,t_2] \to \mathbb{R}$ be an arbitrary function that can be evaluated using at most $Q(\nquant)$ queries and that satisfies $s(t_1) \leq 0$ and $s(t_2) \geq 0$. Then, using at most $O(Q(\nquant) \log \nquant)$ queries, we can compute $t_1^*, t_2^* \in [t_1,t_2]$ with $t_1^* < t_2^*$, $s(t_1^*) \leq 0$, $s(t_2^*) \geq 0$, and such that for each agent $i \in \{1,2,3,4\}$, there exist some (possibly different) $j \in [\nquant]$ such that $c(t_1^*), c(t_2^*) \in [x_{j-1}^i,x_j^i]$.
\end{lemma}

\begin{proof}
We consider the case where $c$ is strictly increasing. The case where it is strictly decreasing can be handled analogously. We initialize $t_1^* := t_1$ and $t_2^* := t_2$. Using two value queries to Agent 1 we can determine $j_1 \leq j_2$ such that $x_{j_1-1}^1 \leq c(t_1^*) \leq x_{j_1}^1$ and $x_{j_2-1}^1 \leq c(t_2^*) \leq x_{j_2}^1$. If $j_1 = j_2$ then we are done with Agent 1. Otherwise, we perform binary search as follows:
\begin{enumerate}
\item[] As long as $j_1 < j_2$:
\item Let $j := \lfloor (j_1+j_2)/2 \rfloor$.
\item Compute $t$ such that $c(t) = x_j^1$.
\item If $s(t) \geq 0$, then set $t_2^* := t$ and $j_2 := j$.
Otherwise, set $t_1^* := t$ and $j_1 := j+1$.
\end{enumerate}
It is easy to see that this runs for at most $\log \nquant$ iterations, and each iteration requires $O(Q(\nquant))$ queries. When it terminates, we have that $[c(t_1^*), c(t_2^*)] \in [x_{j-1}^1,x_j^1]$. Repeating this in sequence for the other three agents yields the desired outcome. In particular, note that each binary search procedure can only make $[c(t_1^*), c(t_2^*)]$ smaller (w.r.t.~inclusion order), so the property will hold simultaneously for all agents at the end.
\end{proof}

\subsection{Algorithm}

We recall some notation that was introduced in the proof of \cref{lem:continuous-path} and in \cref{rem:old-trails}. We let
$\equitwo, \equithree, \equifour$ denote the value of the pieces in the equipartition according to Agent 1 into two, three, and four parts, respectively. Since we consider the additive setting here, we know that $\equitwo = 1/2$, $\equithree = 1/3$, and $\equifour = 1/4$. As mentioned in \cref{rem:old-trails}, the continuous path guaranteed by \cref{lem:continuous-path} always uses the following two types of trails.

For any piece $k \in \{1,2,3,4\}$, we define the trail $\gamma_k^A$ that maps any $\alpha \in [\equifour,\equithree]$ to the unique division $(P_1,P_2,P_3,P_4)$ that satisfies $v_1(P_t) = \alpha$ for all $t \in \{1,2,3,4\} \setminus \{k\}$.

For any agent $i \in \{2,3,4\}$ and any two distinct pieces $k,k' \in \{1,2,3,4\}$, we define the trail $\gamma_{i,k,k'}^B$ that maps any $\alpha \in [\equifour,\equitwo]$ to the unique division $(P_1,P_2,P_3,P_4)$ that satisfies $v_i(P_k) = v_i(P_{k'})$, and $v_1(P_t) = \alpha$ for all $t \in \{1,2,3,4\} \setminus \{k,k'\}$.

Note that any critical division must lie on one of these trails. If it satisfies Condition A then it must lie on a trail of type $\gamma_k^A$ for some $k$. If it satisfies Condition B then it must lie on a trail of type $\gamma_{i,k,k'}^B$ for some $i,k,k'$. We can equivalently restate the definitions of the conditions as follows.
\begin{description}
\item[Condition A:] We say that division $(P_1,P_2,P_3,P_4)$ satisfies Condition A if there exist $\alpha \in [\equifour,\equithree]$ and $k \in \{1,2,3,4\}$ such that $(P_1,P_2,P_3,P_4) = \gamma_k^A(\alpha)$ and additionally there exist two distinct agents $i,i' \in \{2,3,4\}$ such that $v_i(P_k) = \max_t v_i(P_t)$ and $v_{i'}(P_k) = \max_t v_{i'}(P_t)$. (Note that $v_1(P_k) \leq \alpha$ automatically holds since $\alpha \geq \equifour$.)
\item[Condition B:] We say that division $(P_1,P_2,P_3,P_4)$ satisfies Condition B if there exist $\alpha \in [\equifour,\equitwo]$, an agent $i \in \{2,3,4\}$, and two distinct pieces $k,k' \in \{1,2,3,4\}$ such that $(P_1,P_2,P_3,P_4) = \gamma_{i,k,k'}^B(\alpha)$ and additionally:
\begin{enumerate}
\item[i.] $v_1(P_k) \leq \alpha$ and $v_1(P_{k'}) \leq \alpha$, and
\item[ii.] $(v_i(P_{k'}) =)\,\, v_i(P_k) = \max_t v_i(P_t)$, and
\item[iii.] there exists $i' \in \{2,3,4\} \setminus \{i\}$ such that $v_{i'}(P_k) = \max_t v_{i'}(P_t)$, and
\item[iv.] there exists $i' \in \{2,3,4\} \setminus \{i\}$ with $v_{i'}(P_{k'}) = \max_t v_{i'}(P_t)$.
\end{enumerate}
\end{description}
We say that there exists a critical division at value $\val$, if there exists a critical division where Agent 1 has value $\val$ for its favorite pieces.

Let $\gamma$ be a trail, and $[\underline{\alpha}, \overline{\alpha}]$ some interval of values. Let $(\ell(\alpha),m(\alpha),r(\alpha)) := \gamma(\alpha)$ whenever this is well-defined. We say that for all $\alpha \in [\underline{\alpha}, \overline{\alpha}]$, the cuts in $\gamma(\alpha)$ remain in the same $\nquant$-quantile for all agents if for each cut $c \in \{\ell,m,r\}$, and each agent $i \in \{1,2,3,4\}$ there exists $j \in [\nquant]$ such that for all $\alpha \in [\underline{\alpha}, \overline{\alpha}]$ we have $x_{j-1}^i \leq c(\alpha) \leq x_j^i$ (whenever $c(\alpha)$ is well-defined).

We are now ready to state the algorithm.
\begin{enumerate}
\item Compute the (unique) equipartition of the cake into four equal parts according to Agent 1. If this equipartition yields an envy-free division, then stop and return this division.
\item Let $\underline{\alpha} = 1/4$ and $\overline{\alpha} = 1/2$.
\item Using binary search, find $\underline{\alpha}$ and $\overline{\alpha}$ such that
\begin{itemize}
\item There exists a critical division at value $\underline{\alpha}$, but not at value $\overline{\alpha}$.
\item For any trail $\gamma$ and for all $\alpha \in [\underline{\alpha}, \overline{\alpha}]$, the cuts in $\gamma(\alpha)$ remain in the same $\nquant$-quantile for all agents.
\end{itemize}
\item Using the information gathered, output an envy-free allocation.
\end{enumerate}

Before explaining how the steps of the algorithm are implemented, we briefly argue about the correctness. First of all, note that if the equipartition computed in Step 1 of the algorithm is an envy-free division, then the algorithm's output is correct. If it is not an envy-free division, then Condition A or B must hold at value $\equifour = 1/4$. Furthermore, it is easy to see that neither Condition A nor Condition B can hold at value $\equitwo = 1/2$. After the binary search procedure, \cref{lem:continuous-path} guarantees that there is an envy-free division on one of the trails in the interval $[\underline{\alpha}, \overline{\alpha}]$. Since all cuts remain in the same $\nquant$-quantile for all agents in all trails, we have full information about all the trails in the interval $[\underline{\alpha}, \overline{\alpha}]$. We can thus locate the envy-free division and return it.

\subsection{Implementation using value and cut queries}

In this section we argue that all steps of the algorithm can be implemented efficiently using value and cut queries. First of all, computing the equipartition in the first step of the algorithm is straightforward using a constant number of cut queries. Next, we show that we can check whether there exists a critical division at some value $\alpha$.

\begin{lemma}\label{lem:RW-check-invariant}
Given $\alpha \in [0,1]$, we can check whether there exists a critical division at value $\alpha$ using $O(\log (1/\eps))$ value and cut queries.
\end{lemma}

\begin{proof}
We show that we can check whether Condition A or B holds at a given value $\val$.
For Condition A it suffices to show that for any $k \in \{1,2,3,4\}$ we can compute the division $\gamma_k^A(\alpha)$, i.e., the division where Agent 1 has value $\alpha$ for all pieces except (possibly) piece $k$. Indeed, given that division, we can then easily verify the condition by querying the value of each piece for each agent (and this only uses a constant number of value queries). The division $\gamma_k^A(\alpha)$ can be computed using three cut queries. For example, if piece $k=3$, then we first use one cut query to determine the position $\ell$ such that $v_1(0,\ell) = \alpha$, then another cut query to determine $m$ such that $v_1(\ell,m) = \alpha$, and one final cut query to find $r$ such that $v_1(r,1) = \alpha$.

Similarly, for Condition B it also suffices to show that for any agent $i \in \{2,3,4\}$ and any distinct pieces $k,k' \in \{1,2,3,4\}$ we can compute the division $\gamma_{i,k,k'}^B(\alpha)$, i.e., the division where Agent $i$ has identical value for pieces $k$ and $k'$, and where Agent 1 has value $\alpha$ for each of the remaining two pieces. We distinguish between the following three cases:

\paragraph*{\textbf{Case 1: The pieces $k$ and $k'$ are adjacent.}}
We consider the representative example where $k$ and $k'$ are the two middle pieces. The other cases are handled analogously. Using two cut queries we can determine $\ell$ and $r$ such that $v_1(0,\ell) = \alpha$ and $v_1(r,1) = \alpha$. Then, we use one value query to obtain $\beta := v_i(\ell,r)$, and one cut query to obtain $m$ that satisfies $v_i(\ell,m) = \beta/2$. The division $(\ell,m,r)$ corresponds to $\gamma_{i,k,k'}^B(\alpha)$.

\paragraph*{\textbf{Case 2: There is one piece between pieces $k$ and $k'$.}}
Consider the representative example where piece $k$ is the piece $[\ell,m]$, i.e., the second piece from the left, and piece $k'$ is the piece $[r,1]$, i.e., the rightmost piece. Using one cut query we can determine $\ell$ such that $v_1(0,\ell) = \alpha$, and then using one value query we obtain $\beta := v_i(\ell,1)$. For any $t \in [0,\beta/2]$, let $m(t)$ and $r(t)$ denote the unique cuts that satisfy $v_i(\ell,m(t)) = t$ and $v_i(r(t),1) = t$. Our goal is to find the unique $t^*$ that satisfies $v_1(m(t^*),r(t^*)) = \alpha$. For this we can use the binary search approach of \cref{lem:binary-search}. Indeed, $t \mapsto m(t)$ is continuous and strictly increasing, and both $m(\cdot)$ and its inverse $m^{-1}(\cdot)$ can be computed using a single query. Similarly, $t \mapsto r(t)$ is continuous and strictly decreasing, and both $r(\cdot)$ and its inverse $r^{-1}(\cdot)$ can be computed using one query. Let $s: t \mapsto \alpha - v_1(m(t),r(t))$, which can be evaluated using three queries, and note that $s(0) = \alpha - v_1(\ell,1) \leq 0$ and $s(\beta/2) = \alpha \geq 0$. Applying \cref{lem:binary-search} twice, first on $m(\cdot)$ and then subsequently on $r(\cdot)$, we use $O(\log(1/\eps))$ queries to obtain $[t_1^*,t_2^*] \subseteq [0, \beta/2]$ such that $s(t_1^*) \leq 0$, $s(t_2^*) \geq 0$, and $[m(t_1^*),m(t_2^*)]$ and $[r(t_2^*),r(t_1^*)]$ lie within $\nquant$-quantiles for all agents. As a result, with an additional constant number of queries we now have full information about the functions $m$, $r$, and thus $s$ on $[t_1^*,t_2^*]$, and we can find $t^*$ such that $s(t^*) = 0$, i.e., $v_1(m(t^*),r(t^*)) = \alpha$. The division $(\ell, m(t^*), r(t^*))$ then corresponds to $\gamma_{i,k,k'}^B(\alpha)$.

\paragraph*{\textbf{Case 3: Piece $k$ is the leftmost piece and piece $k'$ is the rightmost piece.}}
This case is handled using the same approach as in the previous case. For any $t \in [0,1/2]$, we let $\ell(t)$ and $r(t)$ denote the unique cuts that satisfy $v_i(0,\ell(t)) = t$ and $v_i(r(t),1) = t$. Similar arguments show that we can find $t^*$ such that $v_1(\ell(t^*),r(t^*)) = 2\alpha$ using $O(\log(1/\eps))$ queries. The desired division is then $(\ell(t^*),m,r(t^*))$ where $m$ is picked so as to satisfy $v_1(\ell(t^*),m) = \alpha$.

This completes the proof of the lemma.
\end{proof}

Finally, we have to explain how the third step of the algorithm is implemented, namely the binary search. Starting with $[\underline{\alpha}, \overline{\alpha}] = [1/4,1/2]$, we first consider the trail $\gamma_1^A$ and using binary search we shrink the search interval $[\underline{\alpha}, \overline{\alpha}]$ until the cuts in $\gamma_1^A(\alpha)$ remain in the same $\nquant$-quantile for all agents (we show below how exactly to do this). Next, we consider the trail $\gamma_2^A$ and further shrink the search interval $[\underline{\alpha}, \overline{\alpha}]$ until the desired property holds for $\gamma_2^A$ as well. Importantly, note that the property that we ensured for $\gamma_1^A$ continues to hold, because we have only made the search interval smaller. We continue like this until all four trails of type A and all 18 trails of type B have been handled. The final search interval thus obtained now satisfies the desired property for all trails simultaneously.

In order to complete the proof it remains to show that for any given trail we can efficiently shrink the interval until the desired property holds. This is proved in the following two lemmas.

\begin{lemma}\label{lem:RW-search-A}
Let $\underline{\alpha} < \overline{\alpha}$ be such that there is a critical division at $\underline{\alpha}$ but not at $\overline{\alpha}$. Let $k \in \{1,2,3,4\}$. Then using $O(\log^2(1/\eps))$ queries we can find $\underline{\alpha}^* < \overline{\alpha}^*$ such that $[\underline{\alpha}^*, \overline{\alpha}^*] \subseteq [\underline{\alpha}, \overline{\alpha}]$, there is a critical division at $\underline{\alpha}^*$ but not at $\overline{\alpha}^*$, and such that for all $\alpha \in [\underline{\alpha}^*, \overline{\alpha}^*]$ the cuts $(\ell(\alpha), m(\alpha), r(\alpha)) = \gamma_k^A(\alpha)$ remain in the same $\nquant$-quantile for all agents.
\end{lemma}

\begin{proof}
For concreteness let us consider the case where $k=3$. The other cases are handled analogously. For $\alpha \in [\underline{\alpha}, \overline{\alpha}]$ let $(\ell(\alpha), m(\alpha), r(\alpha)) := \gamma_k^A(\alpha)$. Recall that these are the unique cuts satisfying $v_1(0,\ell(\alpha)) = v_1(\ell(\alpha),m(\alpha)) = v_1(r(\alpha),1) = \alpha$. It is easy to check that $\ell(\cdot)$, $m(\cdot)$, and $r(\cdot)$ are continuous, strictly monotone, and that their inverses can be computed using a single query. We define the function $s: [\underline{\alpha}, \overline{\alpha}] \to \{-1,+1\}$ as follows
\begin{equation}\label{eq:RW-invariant-function}
s(\alpha) = \left\{\begin{tabular}{cl}
$-1$ & if there exists a critical division at value $\alpha$\\
$+1$ & if there is no critical division at value $\alpha$
\end{tabular} \right.
\end{equation}
Note that $s(\underline{\alpha}) = -1$ and $s(\overline{\alpha}) = +1$, and that by \cref{lem:RW-check-invariant} we can evaluate $s$ using $O(\log(1/\eps))$ queries.

It follows that we can apply the binary search approach of \cref{lem:binary-search} to find $\underline{\alpha}^*$ and $\overline{\alpha}^*$ that satisfy the desired properties using $O(\log^2(1/\eps))$ queries. Namely, we first apply \cref{lem:binary-search} to ensure that $\ell(\cdot)$ remains in the same $\nquant$-quantile for all agents, then we apply the lemma again to further restrict the interval of $\alpha$ so that the same holds for $m(\cdot)$, and then a third time for $r(\cdot)$.
\end{proof}

\begin{lemma}
Let $\underline{\alpha} < \overline{\alpha}$ be such that there is a critical division at $\underline{\alpha}$ but not at $\overline{\alpha}$. Let $i \in \{2,3,4\}$ and $k, k' \in \{1,2,3,4\}$ with $k \neq k'$. Then using $O(\log^2(1/\eps))$ queries we can find $\underline{\alpha}^* < \overline{\alpha}^*$ such that $[\underline{\alpha}^*, \overline{\alpha}^*] \subseteq [\underline{\alpha}, \overline{\alpha}]$, there is a critical division at $\underline{\alpha}^*$ but not at $\overline{\alpha}^*$, and such that for all $\alpha \in [\underline{\alpha}^*, \overline{\alpha}^*]$ the cuts $(\ell(\alpha), m(\alpha), r(\alpha)) = \gamma_{i,k,k'}^B(\alpha)$ remain in the same $\nquant$-quantile for all agents.
\end{lemma}

\begin{proof}
Our approach will be similar to the proof of \cref{lem:RW-search-A}. Namely, we will seek to apply the binary search procedure of \cref{lem:binary-search} repeatedly to ensure that all cuts remain in the same $\nquant$-quantile for all agents. In particular, we will use the same function $s$ from \eqref{eq:RW-invariant-function}. However, some cuts will no longer be monotone functions of $\alpha$, and as a result we will have to argue more carefully. For $\alpha \in [\underline{\alpha}, \overline{\alpha}]$, let $(\ell(\alpha), m(\alpha), r(\alpha)) := \gamma_{i,k,k'}^B(\alpha)$. We consider three different cases.

\paragraph*{\textbf{Case 1: The pieces $k$ and $k'$ are adjacent.}}
We focus on the most challenging subcase, namely when $k$ and $k'$ are the two middle pieces. Recall that the cuts satisfy $v_1(0,\ell(\alpha)) = v_1(r(\alpha),1) = \alpha$ and $v_i(\ell(\alpha),m(\alpha)) = v_i(m(\alpha),r(\alpha))$ in this case. As a result, it is easy to see that $\ell(\cdot)$ is strictly increasing and $r(\cdot)$ is strictly decreasing, but $m(\cdot)$ is not necessarily monotone. Nevertheless, we can apply the binary search approach from \cref{lem:binary-search} first to $\ell(\cdot)$, and then to $r(\cdot)$ to ensure that those cuts remain in the same $\nquant$-quantile for all agents. Recall that this requires $O(\log^2(1/\eps))$ queries, because evaluating $s$ requires $O(\log(1/\eps))$ queries. In particular, we now have full information about $\ell(\cdot)$ and $r(\cdot)$ as functions of $\alpha$.

The crucial observation is that $m(\cdot)$ is now a monotone function in $\alpha$. To be more precise, $m(\cdot)$ is either constant or strictly monotone. The reason for this is that as $\alpha$ increases, the cuts $\ell(\alpha)$ and $r(\alpha)$ traverse the value of Agent $i$ at some constant rate (because these two cuts do not cross any $\nquant$-quantile for any agent). More formally, we have that $v_i(\ell(\alpha),\ell(\alpha+\dval)) = c \cdot \dval$, and $v_i(r(\alpha+\dval),r(\alpha)) = c' \cdot \dval$, where $c$ and $c'$ do not depend on $\alpha$. If the two rates are the same, i.e., $c=c'$, then $m(\alpha)$ will not move. Otherwise, if the two rates are different, then $m(\alpha)$ will move left or right depending on which rate is larger.

Now, if $m(\cdot)$ is constant, then it trivially remains in the same $\nquant$-quantile for all agents. If it is strictly monotone, then its inverse is well-defined and we can evaluate it using a constant number of queries. Indeed, for any fixed $m^*$, using a constant number of queries we can obtain full information about the values of all pieces to all agents, and thus determine the value of $\alpha$ such that $m^* = m(\alpha)$. As a result we can now apply \cref{lem:binary-search} to ensure that $m(\cdot)$ also remains in the same $\nquant$-quantile for all agents.

\paragraph*{\textbf{Case 2: There is one piece between pieces $k$ and $k'$.}}
We consider the representative example where piece $k$ is the piece $[\ell,m]$, i.e., the second piece from the left, and piece $k'$ is the piece $[r,1]$, i.e., the rightmost piece. Recall that in this case the cuts satisfy $v_1(0,\ell(\alpha)) = v_1(m(\alpha),r(\alpha)) = \alpha$ and $v_i(\ell(\alpha),m(\alpha)) = v_i(r(\alpha),1)$. Clearly, $\ell(\cdot)$ is strictly increasing in $\alpha$ and we can thus use the binary search approach of \cref{lem:binary-search} to restrict the interval of $\alpha$ so that $\ell(\cdot)$ remains in the same $\nquant$-quantile for all agents.

Next, we observe that $r(\cdot)$ is also strictly increasing in $\alpha$. Indeed, if there exist $\alpha_1 < \alpha_2$ such that $r(\alpha_1) \geq r(\alpha_2)$, then it follows that $m(\alpha_1) \geq m(\alpha_2)$, since $m(\alpha_1) < m(\alpha_2)$ would imply the contradiction that $v_1(m(\alpha_1),r(\alpha_1)) > v_1(m(\alpha_2),r(\alpha_2)) = \alpha$ because $v_1$ is hungry. But since $\ell(\alpha_1) < \ell(\alpha_2)$, we then obtain that $v_i(\ell(\alpha_1),m(\alpha_1)) > v_i(\ell(\alpha_2),m(\alpha_2)) = v_i(r(\alpha_2),1) \geq v_i(r(\alpha_1),1)$, again a contradiction.

Furthermore, we can compute the inverse of $r(\cdot)$ using $O(\log(1/\eps))$ queries. Indeed, given any $r^*$, we can query $\beta := v_i(r^*,1)$, and then use $O(\log(1/\eps))$ queries to find the cuts $\ell^*, m^*$ such that $v_i(\ell^*,m^*) = \beta$ and $v_1(0,\ell^*) = v_1(m^*,r^*)$ (use the same approach as in the proof of \cref{lem:RW-check-invariant}, but swap the roles of Agents 1 and $i$, and replace $\alpha$ by $\beta$). Letting $\alpha := v_1(0,\ell^*)$ we then must have $r^* = r(\alpha)$. As a result, we can use $O(\log^2(1/\eps))$ queries to perform the binary search approach of \cref{lem:binary-search} to restrict the search interval of $\alpha$ so that $r(\cdot)$ remains in the same $\nquant$-quantile for all agents.

In general, $m(\cdot)$ is not a monotone function of $\alpha$. However, now that $\ell(\cdot)$ and $r(\cdot)$ remain in the same $\nquant$-quantile for all agents, $m(\cdot)$ must be a monotone function. Namely, it will be either constant or strictly monotone. To see why this is the case, let us assume that $m(\cdot)$ is not constant and show that $m(\cdot)$ is then necessarily strictly monotone. Let $m^*$ be such that there exists $\alpha^*$ in the current search interval with $m(\alpha^*) = m^*$. We will show that $\alpha^*$ is unique (in the current search interval) and thus $m(\cdot)$ has to be strictly monotone. Since we have fixed $m^*$, using a constant number of queries we now have full information about the position of cut $r'(\alpha)$ as a function of $\alpha$, where $r'(\cdot)$ is defined such that $v_1(m^*,r'(\alpha)) = \alpha$.\footnote{Here we also restrict $\alpha$ to be such that $r'(\alpha)$ lies in the same $\nquant$-quantile as $r(\cdot)$ for all agents. Indeed, if $r'(\alpha)$ does not satisfy this, then it immediately follows that $\alpha$ is such that $m^* \neq m(\alpha)$.} Furthermore, since the cuts $\ell(\cdot)$ and $r'(\cdot)$ do not cross any $\nquant$-quantiles, we know that $v_1(\ell(\alpha),\ell(\alpha+\dval)) = c \cdot v_i(\ell(\alpha),\ell(\alpha+\dval))$ and $v_1(r'(\alpha),r'(\alpha+\dval)) = c' \cdot v_i(r'(\alpha),r'(\alpha+\dval))$. If $c \neq c'$, then there exists at most one $\alpha^*$ with $v_i(\ell(\alpha^*),m^*) = v_i(r'(\alpha^*),1)$, and thus at most one $\alpha^*$ with $m^* = m(\alpha^*)$. If $c=c'$, then since there exists $\alpha^*$ with $m^*=m(\alpha^*)$, it must be that all $\alpha$ in the current interval satisfy $m^*=m(\alpha)$, i.e., $m(\cdot)$ is constant, a contradiction.

If $m(\cdot)$ is constant, then it trivially remains in the same $\nquant$-quantile for all agents. If it is strictly monotone, then its inverse is well-defined and it can be evaluated using a constant number of queries (by using the approach for inverting in the previous paragraph). It follows that we can use the binary search approach of \cref{lem:binary-search} to ensure that $m(\cdot)$ also remains in the same $\nquant$-quantile for all agents.

\paragraph*{\textbf{Case 3: Piece $k$ is the leftmost piece and piece $k'$ is the rightmost piece.}}
Recall that in this case the cuts satisfy $v_i(0,\ell(\alpha)) = v_i(r(\alpha),1)$ and $v_1(\ell(\alpha),m(\alpha)) = v_1(m(\alpha),r(\alpha)) = \alpha$. It is easy to see that $\ell(\cdot)$ is strictly decreasing and $r(\cdot)$ is strictly increasing. Furthermore, they can both be inverted using a constant number of queries. As a result, we can use \cref{lem:binary-search} to ensure that $\ell(\cdot)$ and $r(\cdot)$ remain in the same $\nquant$-quantile for all agents.

In general, $m(\cdot)$ is not a monotone function, but now that $\ell(\cdot)$ and $r(\cdot)$ remain in the same $\nquant$-quantile for all agents, $m(\cdot)$ will be constant or strictly monotone. Furthermore, when it is strictly monotone, it can also be inverted using a constant number of queries, because we can obtain full information about the values of all pieces for all agents when the middle cut is fixed (again using the fact that $\ell(\cdot)$ and $r(\cdot)$ remain in the same $\nquant$-quantile for all agents). As a result, by applying \cref{lem:binary-search} we can again ensure that $m(\cdot)$ remains in the same $\nquant$-quantile for all agents.

This completes the proof of the lemma.
\end{proof}

\section{The Intersection End-of-Line Problem}\label{sec:IEoL}

We begin with the definition of the problem.

\begin{definition}[{\sc Intersection End-of-Line}] \hfill

\begin{description}
\item[Input] 
Each party receives a superset of the edges; we say that an edge is {\em active} in the {\sc End-of-Line} instance $G = (V,E)$ if it is in the intersection of all the supersets. Edges that only appear in the supersets of some parties are called {\em inactive}.

\item[Output] A solution to the {\sc End-of-Line} instance defined by the active edges.
\end{description}

\end{definition}

In this section we prove the following result.

\begin{lemma}\label{lem:CC-IEoL}

{\sc Intersection End-of-Line} with $k \ge 3$ parties requires $\poly(|G|)$ communication complexity,
even in the special case where the instances satisfy the following promises: 
\begin{enumerate}
\item[0.] Every node has at most one active incoming edge, and at most one active outgoing edge.
\item Every inactive edge is only included in at most one party's set of edges.
\item No vertex has both an $i$-inactive and a $j$-inactive incident edge, for some $i \neq j$. (An edge is said to be {\em $i$-inactive}, if it is inactive, but included in Party $i$'s superset.)
\item Every vertex is assigned to one party (the vertex-party assignment is publicly known): Only that party's superset  may have more than one edge coming into that vertex. Similarly, only that party's superset may have more than one edge going out of that vertex.
\end{enumerate}
\end{lemma}

\subsection{Number-on-Forehead End-of-Line}
The starting point for the proof of \cref{lem:CC-IEoL} is the number-on-forehead communication complexity of lifted {\sc End-of-Line}.

\begin{definition}[Number-on-Forehead (NoF) End-of-Line] \hfill

\noindent We consider a (common knowledge) host graph  $H$ and $k$-party {\em lifting gadgets} $g : \Sigma^k \rightarrow \{0,1\}$.\\
{\bf Inputs:} $k$ parties receive number-on-forehead inputs to gadgets $g$ for each edge of $H$.
The inputs induce a subgraph $G'$ of the host graph $H$, i.e.~the edges with $g(\cdot) = 1$.\\
{\bf Output:} The goal is to find a solution of the {\sc End-of-Line} problem on $G'$.
\end{definition}

The following lemma follows as a corollary of results from~\cite{Sherstov14,GP18,GR18}.
\begin{lemma}[Complexity of {\sc NoF End-of-Line}]
For every constant number of parties $k \ge 2$, there exist gadgets $g_k : \Sigma^k \rightarrow \{0,1\}$ for $|\Sigma| = 2^{k^{o(1)}}$, and a constant-degree host graph $H$ over $n$ vertices such that {\sc NoF End-of-Line} requires bounded-error randomized communication of $\tilde{\Omega}(\sqrt{n})$.
\end{lemma}

\begin{proof}
By~\cite{GP18}, there exist gadgets as in the theorem statement such that the communication complexity of {\sc NoF End-of-Line} is lower bounded by the $k$-party NoF communication complexity of unique set disjointness of size $cbs(\textsc{EoL})$, where $cbs(\cdot)$ denotes the critical block sensitivity (we will not define unique set disjointness of critical block sensitivity here --- refer, e.g., to~\cite{GP18}).

By~\cite{GR18}, the critical block sensitivity of {\sc End-of-Line} is $cbs(\textsc{EoL})=\tilde{\Omega}(n)$.
Finally, by~\cite{Sherstov14}, the $k$-party NoF communication complexity of the unique set disjointness of size $cbs(\textsc{EoL})$ is at least $\Omega(\frac{\sqrt{cbs(\textsc{EoL})}}{2^k k}) = \tilde{\Omega}(\sqrt{n})$.
\end{proof}

\subsection{Embedding the lifting gadgets in an End-of-Line graph}\label{sub:hard-intersect-EoL}

\begin{proof}[Proof of \cref{lem:CC-IEoL}]
We reduce {\sc NoF End-of-Line} to {\sc Intersection End-of-Line}. We present the reduction for $k=4$, since this is the version we will later use, but it easily generalizes to any $k \geq 3$.

\subsubsection{Constructing the vertices}
Given host graph $H = (V_H,E_H)$ for the {\sc NoF End-of-Line} problem, we construct the vertex set $V$ of the {\sc Intersection End-of-Line} as follows.

For each vertex $v \in V_H$, we construct the following vertices:
\begin{enumerate}
\item  We construct a vertex $v^1 \in V$, and assign it to Party 1. It will later be convenient to refer to $v^1$ using the labels $v^{\outgoing,1}$ and $v^{\incoming,1}$ (both refer to the same vertex). 
\item
\begin{enumerate} 
\item For each possible vector of inputs $\sigma^{\outgoing,1} \in \Sigma^{\odeg{v}}$ on Party 1's forehead in the lifting gadgets corresponding to edges going out of $v$, we construct a vertex $v^{\outgoing,2}_{\sigma^{\outgoing,1}} \in V$, and assign it to Party 2.
\item Similarly, for each possible vector of inputs $\sigma^{\incoming,1} \in \Sigma^{\ideg{v}}$ on Party 1's forehead in the lifting gadgets corresponding to edges coming into $v$, we construct a vertex $v^{\incoming,2}_{\sigma^{\incoming,1}} \in V$, and assign it to Party 2.
\end{enumerate} 
\item
\begin{enumerate} 
\item For each possible vector of outgoing gadgets inputs $(\sigma^{\outgoing,1}, \sigma^{\outgoing,2}) \in \Sigma^{\odeg{v}} \times \Sigma^{\odeg{v}}$ on Party 1's and Party 2's foreheads, we construct a vertex $v^{\outgoing,3}_{\sigma^{\outgoing,1}, \sigma^{\outgoing,2}} \in V$, and assign it to Party 3.
\item Similarly, for each possible vector of incoming gadgets inputs $(\sigma^{\incoming,1}, \sigma^{\incoming,2}) \in \Sigma^{\odeg{v}} \times \Sigma^{\odeg{v}}$ on Party 1's and Party 2's foreheads, we construct a vertex $v^{\incoming,3}_{\sigma^{\incoming,1}, \sigma^{\incoming,2}} \in V$, and assign it to Party 3.
\end{enumerate} 
\item Analogously, we construct vertices $v^{\outgoing,4}_{\sigma^{\outgoing,1}, \sigma^{\outgoing,2}, \sigma^{\outgoing,3}}, v^{\incoming,4}_{\sigma^{\incoming,1}, \sigma^{\incoming,2}, \sigma^{\incoming,3}} \in V$ and assign them to Party 4. 
\end{enumerate}

We're now ready to construct vertices for the full vectors $\sigma^{\outgoing} = (\sigma^{\outgoing,1}, \sigma^{\outgoing,2}, \sigma^{\outgoing,3}, \sigma^{\outgoing,4})$ (and similarly, $\sigma^{\incoming} = (\sigma^{\incoming,1}, \sigma^{\incoming,2}, \sigma^{\incoming,3}, \sigma^{\incoming,4})$) that correspond to all the inputs for the edges potentially going out of vertex $v$. The vector $\sigma^{\outgoing}$ determines all outgoing neighbors of $v$ in the original graph. We let $w$ denote an arbitrary such outgoing neighbor. If there are no outgoing neighbors, we set $w := v$. Similarly, we let $u$ denote an arbitrary incoming neighbor, and set $u := v$ if there is no such neighbor.
We use $\rho$ and $\tau$ to denote the parties' inputs corresponding to vertices $u$ and $w$ (resp.).

For each $v \in V$ and corresponding lifting gadget inputs $\sigma^{\outgoing}$ and $\sigma^{\incoming}$, yielding neighborhood $u \rightarrow v \rightarrow w$, we construct the following vertices%
\footnote{We remark that Layers 5-8 are useful for guaranteeing the desideratum about vertices assigned to parties; for hardness of {\sc Intersection End-of-Line} in the general case they are unnecessary.}:
\begin{enumerate}
  \setcounter{enumi}{4}
\item Two vertices $(v^{\outgoing,5}_{\sigma^{\outgoing}},w^{\incoming})$ and $(v^{\incoming,5}_{\sigma^{\incoming}},u^{\outgoing})$, assigned to Party 1.
\item Another layer of vertices, depending on Party 1's inputs and assigned to Party 2:
\begin{enumerate}
\item For each input $\tau^{\incoming,1}$ on Party 1's forehead for $w$'s incoming edges, construct vertex $(v^{\outgoing,6}_{\sigma^{\outgoing}},w^{\incoming}_{\tau^{\incoming,1}})$.
\item For each input $\rho^{\outgoing,1}$ on Party 1's forehead for $u$'s outgoing edges, construct vertex $(v^{\incoming,6}_{\sigma^{\incoming}},u^{\outgoing}_{\rho^{\outgoing,1}})$.
\end{enumerate}
\item Vertices $(v^{\outgoing,7}_{\sigma^{\outgoing}},w^{\incoming}_{\tau^{\incoming,1},\tau^{\incoming,2}}), (v^{\incoming,7}_{\sigma^{\incoming}},u^{\outgoing}_{\rho^{\outgoing,1},\rho^{\outgoing,2}})$ assigned to Party 3.
\item Vertices $(v^{\outgoing,8}_{\sigma^{\outgoing}},w^{\incoming}_{\tau^{\incoming,1},\tau^{\incoming,2},\tau^{\incoming,3}}), (v^{\incoming,8}_{\sigma^{\incoming}},u^{\outgoing}_{\rho^{\outgoing,1},\rho^{\outgoing,2},\rho^{\outgoing,3}})$ assigned to Party 4.
\end{enumerate}

Finally we're ready to construct the last layer of vertices, which is not assigned to any party\footnote{These vertices can be assigned to an arbitrary party for the purpose of ensuring desideratum 3.
}:
\begin{enumerate}
  \setcounter{enumi}{8}
  \item We construct two vertices, each with two symmetric labels:
\begin{align*}
(v^{\outgoing,9}_{\sigma^{\outgoing}},w^{\incoming}_{\tau^{\incoming}}) &= (w^{\incoming,9}_{\tau^{\incoming}},v^{\outgoing}_{\sigma^{\outgoing}})\\
(v^{\incoming,9}_{\sigma^{\incoming}},u^{\outgoing}_{\rho^{\outgoing}}) &=(u^{\outgoing,9}_{\rho^{\outgoing}},v^{\incoming}_{\sigma^{\incoming}}).
\end{align*}
  \end{enumerate}

\subsubsection{Constructing the edges}

We describe the edges between the ``$\outgoing$'' vertices; the edge construction for ``$\incoming$'' vertices is analogous. 
\begin{itemize}
\item For $i \in \{1,2,3,4\}$, Party $i$'s superset includes the edge from $v^{\outgoing,i}$-vertex to a $v^{\outgoing,i+1}$-vertex whenever:
\begin{itemize}
\item $v^{\outgoing,i}$'s subscript is a prefix of $v^{\outgoing,i+1}$'s subscript; and
\item all the lifting gadget inputs in the subscripts (in particular $v^{\outgoing,i+1}$'s subscript) are consistent with Party $i$'s information, i.e., all the inputs not on $i$'s forehead are correct.
\end{itemize}
\item Similarly, for $i \in \{5,6,7,8\}$, Party $(i-4)$'s superset includes the edge from $(v^{\outgoing,i},w^{\outgoing})$-vertex to a $(v^{\outgoing,i+1},w^{\outgoing})$-vertex whenever the former subscript is a prefix of the latter, and they are consistent with Party $i$'s input. 
\end{itemize}

Notice that an edge $(v \rightarrow w)$ in the {\sc NoF End-of-Line} instance is embedded as the following path in the intersection of the parties' supersets:

\begin{gather*}\Big(v^1 \rightarrow \dots \rightarrow (v^{\outgoing,5}_{\sigma^{\outgoing}},w) \rightarrow \dots \rightarrow (v^{\outgoing,9}_{\sigma^{\outgoing}},w^{\incoming}_{\tau^{\incoming}})= (w^{\incoming,9}_{\tau^{\incoming}},v^{\outgoing}_{\sigma^{\outgoing}}) \rightarrow \dots \rightarrow w^1\Big).
\end{gather*}

\subsubsection{\textsc{NoF End-of-Line} to \textsc{Intersection End-of-Line}: Analysis}

First, notice that by construction every vertex has at most one active incoming edge, and at most one active outgoing edge. This is due to the fact that whenever there are multiple possible successors in the tree-like graph that we construct, at least 3 parties have all the information needed to exclude all but one of those outgoing edges. Furthermore, the only vertices that can have odd degree are the vertices on the last layer, namely the vertices of the form $(v^{\outgoing,9}_{\sigma^{\outgoing}},w^{\incoming}_{\tau^{\incoming}})$ and $(v^{\incoming,9}_{\sigma^{\incoming}},u^{\outgoing}_{\rho^{\outgoing}})$. This is because our construction ensures that there is a unique active path from $v^1$ to one of its leaves $(v^{\outgoing,9}_{\sigma^{\outgoing}},w^{\incoming}_{\tau^{\incoming}})$, and a unique active path from one of the leaves $(v^{\incoming,9}_{\sigma^{\incoming}},u^{\outgoing}_{\rho^{\outgoing}})$ to $v^1$. Every other edge in any of the two trees around $v^1$ will be inactive, because at least one of the endpoints of such an edge must have a subscript that is inconsistent with the correct input information, and thus at least one party (in fact, all but one) will see this and not include the edge. Now, if a vertex of the form $(v^{\outgoing,9}_{\sigma^{\outgoing}},w^{\incoming}_{\tau^{\incoming}})$ has odd degree, then it must have out-degree 0, since it has in-degree 1. But if it has out-degree 0, that means that it is not included in the tree constructed around $w^1$. This can only happen if: (i) $w$ has multiple predecessors in the original graph (and picked a different one in the tree branch corresponding to those inputs), or (ii) $w=v$ (i.e., $v$ has no outgoing edge) and $v$ has at least one incoming edge. In both cases we obtain a solution to {\sc NoF End-of-Line}.

We now verify the remaining particular desiderata from the statement of \cref{lem:CC-IEoL}.
\begin{enumerate}
\item \underline{Every inactive edge is only included in at most one party's set of edges.}\\
Every inactive edge has at least one endpoint with at least one label that is different from the corresponding lifting gadget input on some party's forehead. 
Hence it will not be included in any of the other parties' supersets (all other parties know this label is wrong). 
\item \underline{No vertex has both an $i$-inactive and a $j$-inactive incident edge.}\\
If an $i$-inactive edge is between two vertices that both have (purported) information from Party $i$'s input in their subscript, then that information is the same on both vertices and does not correspond to what lies on Party $i$'s forehead. As a result, no other Party will have any incident edges on these two vertices. If only one of the two endpoints contains (purported) information from Party $i$'s input, then that information is incorrect (so that endpoint does not have any $j$-inactive edge) and all other inactive edges incident on the other endpoint must have an endpoint with incorrect $i$-information as well.

\item \underline{Only the assigned party's superset may have $>$1 outgoing or $>1$ incoming edges on a vertex.}\\
Party $i$'s superset can have more than one edge coming into the same vertex if and only if it does not know the next label because it is written on Party $i$'s forehead. By construction, this only happens for vertices assigned to Party $i$. An analogous argument holds for outgoing edges.
\end{enumerate}
This completes the proof.
\end{proof}

\section{Communication Lower Bound for Envy-free Cake-cutting with Four Non-monotone Agents}\label{sec:cake-hardness}

In this section we prove the following result.

\begin{theorem}
For four agents with non-monotone valuations, finding an $\eps$-envy-free connected allocation requires $\poly(1/\eps)$ communication.
\end{theorem}

The rest of this section is devoted to the proof of this theorem. We reduce from the $4$-party {\sc Intersection End-of-Line} problem. In~\cref{sub:special-embedding} we show how each party can embed its superset of edges into a two-dimensional variant of the {\sc Sperner} problem. In~\cref{sub:valuations} we then use each party's embedding to construct a corresponding valuation function for a cake cutting instance. Finally, in~\cref{sub:hardness-analysis} we analyze this reduction; specifically we show that any $\eps$-envy free partition of the cake maps to a region in the Sperner embedding that embeds a solution of the original {\sc Intersection End-of-Line} problem.

Our reduction can also be used to obtain the following two results. We provide more details about this in \cref{sec:hardness-PPAD-query}.

\begin{theorem}
For four agents with \emph{identical} non-monotone valuations, finding an $\eps$-envy-free connected allocation requires $\poly(1/\eps)$ queries.
\end{theorem}

\begin{theorem}
For four agents with \emph{identical} non-monotone valuations, finding an $\eps$-envy-free connected allocation is \ppad/-complete.
\end{theorem}

\subsection{Special Sperner Embedding of End-of-Line}\label{sub:special-embedding}

In this section, we present the first step of the reduction, which is a special embedding of an \textsc{End-of-Line} instance into a two-dimensional Sperner-type problem. While it is well-known how to do this in the query or white-box model \cite{CD09}, we have to very carefully construct a special embedding here since we are reducing from the communication problem {\sc Intersection End-of-Line}. In more detail, each of the four parties is going to separately embed the edges of {\sc Intersection End-of-Line} that it sees (i.e., that are included it its superset of the edges) in a very specific way.\footnote{An important point is that agents can see multiple incoming or outgoing edges on a vertex. Existing embeddings do not allow for this, because it cannot occur in a standard End-of-Line instance. Our embedding is carefully constructed in order to be able to handle this.} This will ensure that an important set of properties (proved in the claims below) will hold. In the next section, the embedding of party $i$ will then be used to construct the valuation function of agent $i$. The properties proved here will be crucial for the latter analysis of the reduction. In this section we use ``agent'' and ``party'' interchangeably.

Consider an {\sc Intersection End-of-Line} instance with four parties/agents that satisfies the promises of \cref{lem:CC-IEoL}. Let $n$ denote the number of nodes of the instance. Let vertex $1$ be the trivial source.\footnote{We sometimes use $0$ to denote the trivial source in other parts of the paper.}

\paragraph*{\bf Grid and labeling.}
Consider the two-dimensional grid $S = \{0, 1, 2, \dots, N\}^2$ where $N := 300 n^4$. For each agent $i \in \{1,2,3,4\}$ we construct a labeling $\lambda_i: S \to \{+1,-1\}^2$. In other words, every agent assigns one label to each grid point, and the possible labels are $(+1,+1)$, $(+1,-1)$, $(-1,+1)$, $(-1,-1)$. We say that $(+1,+1)$ and $(-1,-1)$ are opposite labels, and similarly $(+1,-1)$ and $(-1,+1)$ are opposite labels too. A square of the grid is a set of points $\{j, j+1\} \times \{k, k+1\}$ for some $j,k \in \{0, 1, 2, \dots, N-1\}$. We say that an agent sees a Sperner solution in a square of the grid, if the four corners contain two opposite labels, i.e., if both labels $(+1,+1)$ and $(-1,-1)$ appear, or if both labels $(+1,-1)$ and $(-1,+1)$ appear. In particular, if a square is not a Sperner solution, then the labels appearing at its four corners must be all be the same or belong to one of the following four sets: $\{(+1,+1), (+1,-1)\}$, $\{(+1,+1), (-1,+1)\}$, $\{(-1,-1), (-1,+1)\}$, or $\{(-1,-1), (+1,-1)\}$.

\paragraph*{\bf Vertex regions and lanes.}
For each of the $n$ vertices of End-of-Line, we define a corresponding region on the diagonal of grid $S$. Namely, for vertex $j \in [n]$, its vertex region is
$$V(j) := \{300 n^3 (j-1) + 1, \dots, 300 n^3 j\}^2.$$
The vertex region has a certain number of lanes incident to it. They are shown in \cref{fig:vertex-region} and defined formally below.

We subdivide the set of grid points
$$\{1, \dots, 300 n^3 (j-1) + 1\} \times \{300 n^3 (j-1) + 100 n^3 + 1, \dots, 300 n^3 (j-1) + 200 n^3\}$$
into $n$ horizontal lanes of width $100 n^2$ each. For any $k < j$, we let $L_j^{left}(k)$ denote the $k$th such lane. Formally, $L_j^{left}(k) := \{1, \dots, 300 n^3 (j-1) + 1\} \times \{300 n^3 (j-1) + 100 n^3 + 100 n^2 (k-1) + 1, \dots, 300 n^3 (j-1) + 100 n^3 + 100 n^2 k\}$.

We subdivide the set of grid points
$$\{300 n^3 j, \dots, N\} \times \{300 n^3 (j-1) + 100 n^3 + 1, \dots, 300 n^3 (j-1) + 200 n^3\}$$
into $n$ horizontal lanes of width $100 n^2$ each. For any $k > j$, we let $L_j^{right}(k)$ denote the $k$th such lane. Formally, $L_j^{right}(k) := \{300 n^3 j, \dots, N\} \times \{300 n^3 (j-1) + 100 n^3 + 100 n^2 (k-1) + 1, \dots, 300 n^3 (j-1) + 100 n^3 + 100 n^2 k\}$.

We subdivide the set of grid points
$$\{300 n^3 (j-1) + 200 n^3 + 1, \dots, 300 n^3 j\} \times \{300 n^3 j, \dots, N\}$$
into $n$ vertical lanes of width $100 n^2$ each. For any $k > j$, we let $L_j^{top}(k)$ denote the $k$th such lane. Formally, we have $L_j^{top}(k) := \{300 n^3 (j-1) + 200 n^3 + 100 n^2 (k-1) + 1, \dots, 300 n^3 (j-1) + 200 n^3 + 100 n^2 k\} \times \{300 n^3 j, \dots, N\}$.

Similarly, we also subdivide the set of grid points
$$\{300 n^3 (j-1) + 1, \dots, 300 n^3 (j-1) + 100 n^3\} \times \{1, \dots, 300 n^3 (j-1)+1\}$$
into $n$ vertical lanes of width $100 n^2$ each. For any $k < j$, we let $L_j^{bottom}(k)$ denote the $k$th such lane. Formally, we have $L_j^{bottom}(k) := \{300 n^3 (j-1) + 100 n^2 (k-1) + 1, \dots, 300 n^3 (j-1) + 100 n^2 k\} \times \{1, \dots, 300 n^3 (j-1)+1\}$.

\begin{figure}
	\centering
	\includegraphics[scale=0.75]{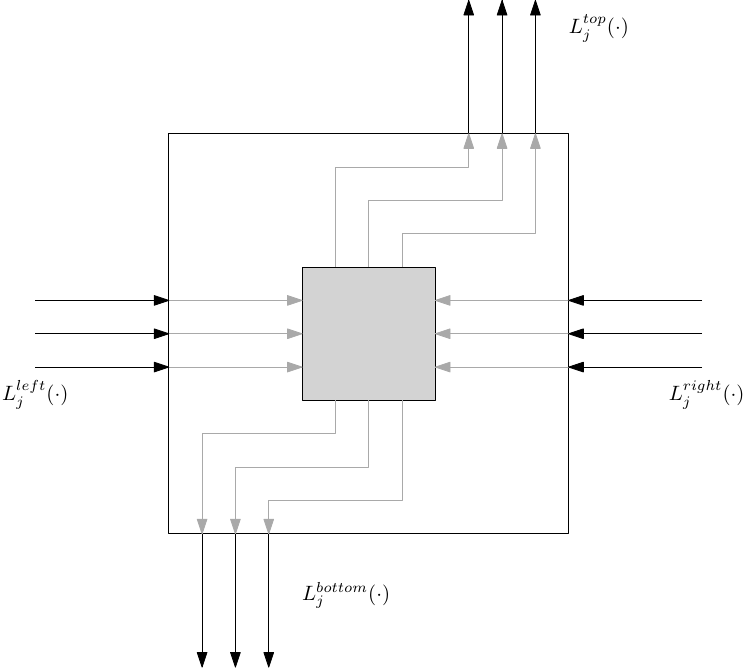}
	\caption{Illustration of the vertex region $V(j)$. The black edges represent lanes (where we have only represented three of each type here). The gray square in the center of the vertex region represents the connector region. The gray edges inside the vertex region show how (potential) paths can be routed to the connector region without intersecting.}
	\label{fig:vertex-region}
\end{figure}

\paragraph*{\bf Embedding of edges.} If agent $i$ sees an edge from vertex $j$ to vertex $k$ in the {\sc Intersection End-of-Line} instance, then it will implement a path from vertex region $V(j)$ to vertex region $V(k)$ in its labeling $\lambda_i$. In more detail:
\begin{itemize}
    \item If $j < k$, then the path leaves $V(j)$ using vertical lane $L_j^{top}(k)$ and moves up until it reaches horizontal lane $L_k^{left}(j)$. The path then turns towards the right and follows the horizontal lane $L_k^{left}(j)$ until it reaches $V(k)$.
    \item If $j > k$, then the path leaves $V(j)$ using vertical lane $L_j^{bottom}(k)$ and moves down until it reaches horizontal lane $L_k^{right}(j)$. The path then turns towards the left and follows the horizontal lane $L_k^{right}(j)$ until it reaches $V(k)$.
\end{itemize}

\paragraph*{\bf Environment label and paths.} The environment label is $(+1,-1)$. Unless otherwise specified, the label of a grid point is the environment label $(+1,-1)$. A path consists of a central part with label $(-1,+1)$ that is surrounded by one protective layer on each side. The protective layer on the left (when moving in the forward direction) has label $(+1,+1)$. The protective layer on the right has label $(-1,-1)$. Note that a path does not introduce a Sperner solution, because opposite labels are isolated from each other. The path can also turn without introducing Sperner solutions. Solutions do however appear at the beginning or end of a path, because the central part with label $(-1,+1)$ is in contact with the environment label $(+1,-1)$. For our purposes it will be convenient to use paths where the central part and protective layers have width two. See \cref{fig:path} for an illustration of the implementation of a path.

\begin{figure}
	\centering
	\includegraphics[scale=0.8]{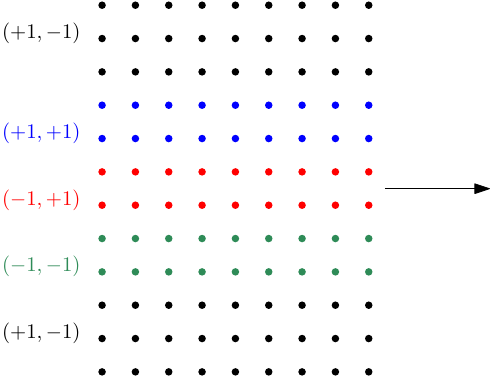}
	\caption{Implementation of a path using the Sperner labels. The arrow indicates that the path is moving from left to right. The color of each grid point indicates its label. The label corresponding to each color is shown on the left.}
	\label{fig:path}
\end{figure}

\paragraph*{\bf Boundary.} On the boundary of the grid $S$, we fix the labels so that they satisfy a Sperner-type boundary condition. In more detail
\begin{itemize}
    \item Points close to the top boundary of the grid, i.e., in $\{0, \dots, N\} \times \{N-2, N-1, N\}$, have label $(+1,-1)$.
    \item Points close to the bottom boundary of the grid, i.e., in $\{0, \dots, N\} \times \{0,1,2\}$, have label $(-1,+1)$.
    \item Points close to the left boundary of the grid, i.e., in $\{0,1,2\} \times \{0, \dots, N\}$, have label $(+1,+1)$.
    \item Points close to the right boundary of the grid, i.e., in $\{N-2, N-1, N\} \times \{0, \dots, N\}$, have label $(-1,-1)$.
\end{itemize}
For points close to the corners of the grid, where the boundary conditions above specify two different labels, we pick one of those two labels arbitrarily but consistently across agents.

Using a simple and standard construction \cite{CD09,GoldbergH21-hairy-ball}, it is possible to use these boundary conditions to create a single starting path in the vertex region $V(1)$. Since all agents agree that the trivial source $1$ has no incoming edges and exactly one outgoing edge, this starting path is used to implement this outgoing edge. Note that the agents also agree on the successor of vertex $1$. Thus, all four agents also agree on the labeling in region $V(1)$.

\paragraph*{\bf Crossing regions.} Outside of vertex regions, the only regions where two paths can cross are the regions where a horizontal lane crosses a vertical lane. Previous work \cite{CD09} has shown how such crossings can be handled locally in order to ensure that no Sperner solutions occur there. We will also make use of such crossing gadgets, although we will have to design them very carefully in our setting. But before we even get to the crossing gadgets, our arguments later in the reduction will require us to ensure that all crossing gadgets occur on different $x$-coordinates on the grid $S$.

Recall that we have specified that edges are implemented by paths that follow specific lanes, depending on which edge is implemented. We have however left some freedom in how exactly the paths travel inside their respective lanes. Furthermore, we have constructed the lanes so that they each have width $100 n^2$. Consider a path moving upwards in vertical lane $L_j^{top}(k)$. At any given point, the path intersects at most one horizontal lane, and overall it is going to intersect at most $n^2$ different horizontal lanes (because every vertex yields at most $n$ horizontal lanes). Now, we subdivide lane $L_j^{top}(k)$ into $n^2$ sublanes of width $100$ each. We can enforce that the path going up in $L_j^{top}(k)$ always uses the sublane corresponding to the horizontal lane it is currently crossing. Namely, if the path is currently crossing the $p$th horizontal lane of vertex $q$, then our path will use the $((q-1) n + p)$th sublane of $L_j^{top}(k)$. \cref{fig:zig-zag-lanes} shows how this can be implemented. For paths traveling downwards on a vertical lane $L_j^{bottom}(k)$ we also define sublanes and proceed in the same way. For paths traveling on horizontal lanes, we do not need to do any of this. Instead, we just let the path travel in a straight line in the center of the lane.

The intersection of a horizontal lane and the corresponding sublane of a vertical lane is a big rectangle of size $100 \times 100 n^2$. Note that the (potential) paths traveling in the horizontal lane and in the vertical sublane would meet in the center of that rectangle. We delimit a big square of dimensions $50 \times 50$ around the center of the rectangle and call it the \emph{crossing region} between the horizontal lane and the vertical lane. The crossing gadgets will be constructed inside the corresponding crossing region, when needed. The construction detailled above ensures that all crossing regions have separate $x$-coordinates. More formally, any two grid points that lie in two different crossing regions must have different $x$-coordinates.

\begin{figure}
	\centering
	\includegraphics[scale=0.8]{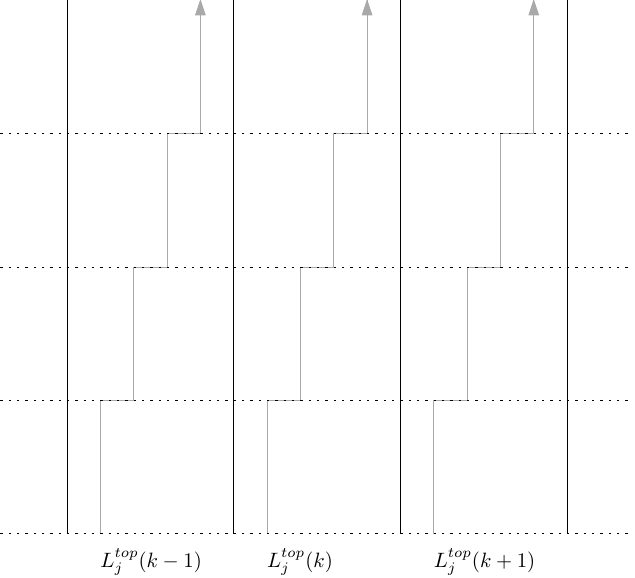}
	\caption{Illustration showing how we can ensure that paths use the sublane corresponding to the horizontal lane currently being crossed. The vertical black lines delimit the vertical lanes, and the horizontal dotted lines delimit the horizontal lanes. The gray arrow shows the trajectory that a (potential) path in the corresponding vertical lane would follow.}
	\label{fig:zig-zag-lanes}
\end{figure}

\begin{claim}[Embedding properties outside special regions]\label{clm:properties-outside}
The special embedding ensures that outside any crossing or vertex region, every square satisfies one of the following properties:
\begin{enumerate}
    \item All four agents agree on the labels of the four corners, and there is no Sperner solution.
    \item At least three agents see the environment label $(+1,-1)$ at all four corners.
\end{enumerate}
\end{claim}

\begin{proof}
Outside of crossing and vertex regions every potential path has its own lane, and these lanes do not intersect by construction (except in crossing regions, of course). Furthermore, if one agent sees a path in a lane, then by property 1 of \cref{lem:CC-IEoL}, either all agents see the path, or no other agent sees the path. Finally, all agents agree on the labeling close to the boundary, and, in particular, in the creation of the starting path for $V(1)$. As a result, every square falls into one of the two categories above.
\end{proof}

It remains to handle vertex regions and crossing regions, and, in particular, to explain how the labeling is defined there.

\paragraph*{\bf Labeling inside vertex regions.} Let $j > 1$ be any vertex. Let the connector region be defined as the central part of $V(j)$, namely $\{300 n^3 (j-1) + 100 n^3 + 1, \dots, 300 n^3 (j-1) + 200 n^3\}^2$. Note that by construction the connector region does not share any $x$-coordinates with any of the vertical lanes $L_j^{top}(k)$ or $L_j^{bottom}(k)$.

For any agent $i$, the labeling inside the vertex region $V(j)$ is constructed as follows.
First, we route each incoming and outgoing path from the boundary of $V(j)$ to its corresponding position on the boundary of the connector region. \cref{fig:vertex-region} shows a principled way of achieving this while ensuring that no paths intersect.

It remains to specify the labeling inside the connector region. If the agent sees more than one incoming path, or more than one outgoing path, then all points in the connector region are labeled $(-1,-1)$. If the agent sees one incoming path and no outgoing path, or one outgoing path and no incoming path, then all points in the connector region are labeled with the environment label $(+1,-1)$. Finally, if the agent sees exactly one incoming path and exactly one outgoing path, then the two paths are connected inside the connector region. We can make sure that this connection is performed in a consistent way, namely such that different agents who see the same incoming and outgoing path perform the connection in the exact same way.

A vertex region $V(j)$ is said to be an \emph{end-of-line vertex region}, if the corresponding vertex $j$ is an end-of-line solution of the {\sc Intersection End-of-Line} instance. We say that $V(j)$ is a \emph{non-end-of-line vertex region}, if it is not an end-of-line vertex region. Ultimately, our construction will ensure that any envy-free division yields a solution in an end-of-line vertex region.

\begin{claim}[Embedding properties in non-end-of-line vertex region]\label{clm:properties-vertex}
The special embedding ensures that in any non-end-of-line vertex region, one of the following cases occurs:
\begin{enumerate}
\item All four agents see the same single path passing through (and no other paths).
\item At least three agents see no path.
\item Three agents see the same single path passing through, and the remaining agent sees that path along with possibly others, but has label $(-1,-1)$ everywhere in the connector region.
\end{enumerate}
\end{claim}

\begin{proof}
Since the vertex region $V(j)$ is a non-end-of-line vertex region, it follows that in the original graph, and in terms of the active edges, $j$ is either an isolated vertex or has exactly one incoming and exactly one outgoing edge. If $j$ is isolated, then by property 2 in \cref{lem:CC-IEoL}, at most one agent sees incident edges on $j$ (which are all inactive), and the other three agents see no incident edges on $j$. As a result, at least three agents see no path at all in $V(j)$.

If $j$ has exactly one incoming and exactly one outgoing edge, then all four agents must see these two edges (by definition of what it means for an edge to be active). If they do not see any further incident edges on $j$, then all four agents agree and see that single path passing through $V(j)$. If, on the other hand, some agents see additional (inactive) edges, then by property 3 in \cref{lem:CC-IEoL}, only a single agent, say agent $i$, can see additional edges, while the other three agents only see the active incoming and outgoing edge. As a result, the three agents see the same single path passing through $V(j)$, and agent $i$ sees that path along with others. Furthermore, by construction of the labeling in $V(j)$, the connector region is labeled $(-1,-1)$ for agent $i$, as claimed.
\end{proof}

\begin{claim}\label{clm:properties-vertex-details}
In any non-end-of-line vertex region, every square satisfies one of the following properties:
\begin{enumerate}
\item All four agents agree on the labels at the four corners, and there is no Sperner solution.
\item At least three agents agree and they see the same identical label at all four corners.
\item At least three agents see labels only in $\{(-1,+1), (+1,+1)\}$, or only in $\{(+1,-1), (+1,+1)\}$, and the square touches the connector region.
\item All four agents see labels only in $\{(+1,-1), (-1,-1)\}$, or only in $\{(-1,+1), (-1,-1)\}$.
\end{enumerate}
\end{claim}

\begin{proof}
By \cref{clm:properties-vertex} there are three possible cases. We show that in each of those three cases, every square in the vertex region falls into (at least) one of the four categories mentioned in the statement.

If all agents see the same single path passing through the vertex region and nothing else, then the construction ensures that all four agents agree on the labels of all points in the vertex region. In particular, in every square, all four agents agree on the labels at the four corners, and these labels do not correspond to a Sperner solution, by the construction of the paths. Thus, every square belongs to category 1.

If at least three agents see no path, then they label the whole vertex region with the environment label $(+1,-1)$. As a result, every square belongs to category 2.

The final case is when three agents see the same single path passing through, and the remaining agent, say agent $i$, sees that path along with possibly others, but has label $(-1,-1)$ everywhere in the connector region. First of all, note that any square that is not a part of the single path seen by the three agents, must be fully labeled with the environment label $(+1,-1)$ by the three agents, and thus falls into category 2. Next, for any square lying on the single path but not touching the connector region, agent $i$ agrees with the labels of the other three agents (since agent $i$ also sees that path). Thus, the square falls into category 1. Furthermore, any square on the single path where the three agents have identical label (i.e., the same label at all four corners) falls into category 2. It remains to handle squares that lie on the single path and touch the connector region, and where the three agents see two different labels. By construction of the paths, the only possibilities for these two labels are: (i) $\{(-1,+1), (+1,+1)\}$, (ii) $\{(+1,-1), (+1,+1)\}$, (iii) $\{(+1,-1), (-1,-1)\}$, and (iv) $\{(-1,+1), (-1,-1)\}$. Cases (i) and (ii) immediately fall into category~3.

It remains to argue about cases (iii) and (iv). We consider two subcases: a square that lies on the single path and touches the connector region, and where the three agents see labels (iii) or (iv) must necessarily occur in (a) well inside the connector region, or (b) on the boundary of the connector region. In case (a), agent $i$ has label $(-1,-1)$ and the square falls into category 4, because label $(-1,-1)$ is a subset of both (iii) and (iv). Finally, in case (b) the square lies on the boundary of the connector region, meaning that two corners lie inside the connector region, and two corners lie outside the connector region. Since agent $i$ agrees with the other three agents on the two corners outside the connector region, and sees label $(-1,-1)$ in the two corners inside the connector region, we again have that the square falls into category 4, because label $(-1,-1)$ is a subset of both (iii) and (iv).
\end{proof}

\paragraph*{\bf Labeling inside crossing regions.} Any crossing region is of one of the two following types: (a) a (potential) crossing between a path going up (increasing in the $y$-coordinate) and a path going from left to right (increasing in the $x$-coordinate), or (b) a (potential) crossing between a path going down (decreasing in the $y$-coordinate) and a path going from right to left (decreasing in the $x$-coordinate).

Let us first consider a crossing region of type (a). We now explain how the labeling is locally modified inside the crossing region. If the agent sees no paths in the crossing region, then the crossing is fully labeled by the environment label $(+1,-1)$. If the agent sees only the vertical path, then we just implement the vertical path normally (\cref{fig:crossing-vertical-path}). If the agent only sees the horizontal path, then we implement the horizontal path, but with a small modification, see \cref{fig:crossing-horizontal-path}. If the agent sees both a horizontal and vertical path, then we implement a crossing gadget which locally reroutes the paths to avoid the intersection. In our setting, we have to construct this gadget very carefully, see \cref{fig:crossing}.

In the latter parts of this reduction we will have to argue very carefully about this crossing gadget. It is therefore convenient to define some notation that will be useful later. Consider the path that a single vertical path follows inside the crossing region. We can think of the path as being decomposed into separate columns, such that in each column all of the squares are identical. We let $c_1^+$ denote the column that contains the squares that contain both label $(+1,-1)$ and label $(-1,-1)$. Note that this indeed identifies a unique column of squares on that vertical path. Similarly, we let $c_1^-$ denote the column that contains the squares that contain both label $(-1,+1)$ and label $(-1,-1)$. Note that this column is also unique. In the same manner, we also consider the (bent) horizontal path and note that it has a portion of path that moves down vertically. For this vertical portion, we define columns $c_2^+$ and $c_2^-$ analogously. Namely, $c_2^+$ is the column of squares with labels $(+1,-1)$ and $(-1,-1)$, and $c_2^-$ is the column of squares with labels $(-1,+1)$ and $(-1,-1)$. The red edge in \cref{fig:crossing-overlay-horizontal} and in \cref{fig:crossing-overlay-vertical} is the portion of path of interest here.

For crossing regions of type (b), we use the same construction for the crossing, except that all paths now run in the opposite direction, compared to type (a). We also define columns $c_1^+, c_1^-, c_2^+, c_2^-$ analogously, where we use the upward vertical portion of the horizontal (bent) path, instead of the downward portion. In particular, we make sure that $c_1^+$ and $c_2^+$ still correspond to columns with squares with labels $(+1,-1)$ and $(-1,-1)$, and $c_1^-$ and $c_2^-$ still correspond to columns with squares with labels $(-1,+1)$ and $(-1,-1)$.

\begin{figure}
	\centering
	\includegraphics[scale=0.5]{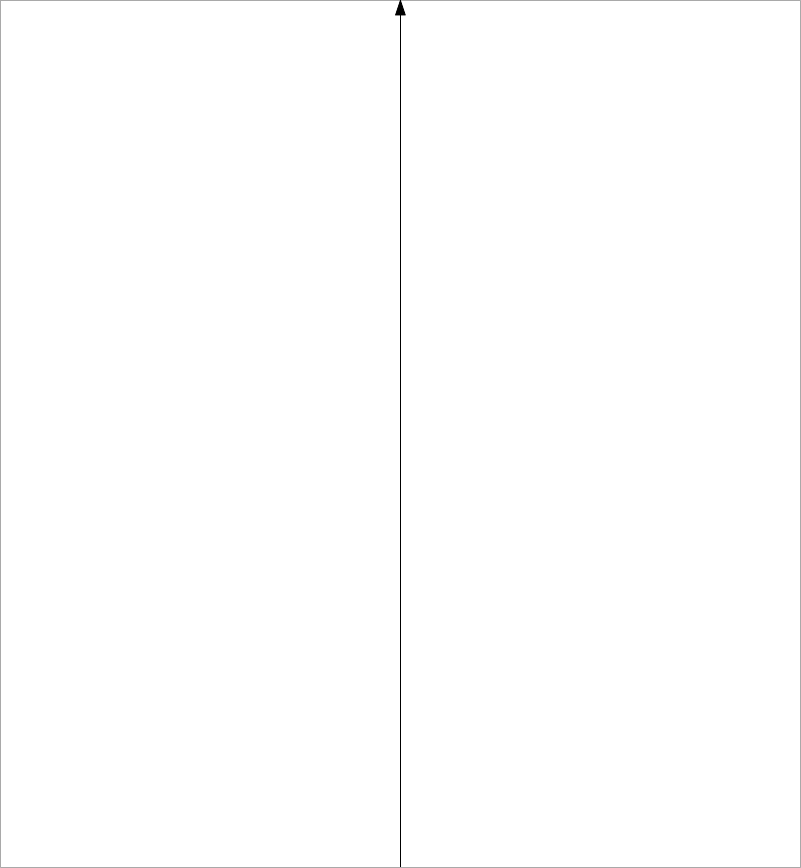}
	\caption{Implementation of the crossing gadget when the agent only sees a vertical path.}
	\label{fig:crossing-vertical-path}
\end{figure}

\begin{figure}
	\centering
	\includegraphics[scale=0.5]{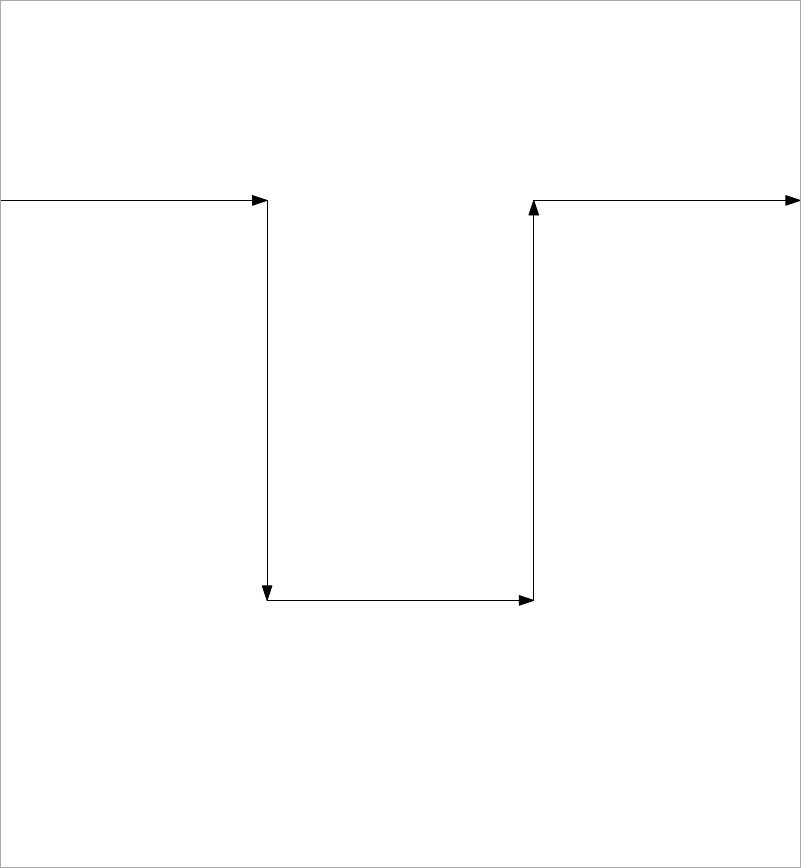}
	\caption{Implementation of the crossing gadget when the agent only sees a horizontal path.}
	\label{fig:crossing-horizontal-path}
\end{figure}

\begin{figure}
	\centering
	\includegraphics[scale=0.5]{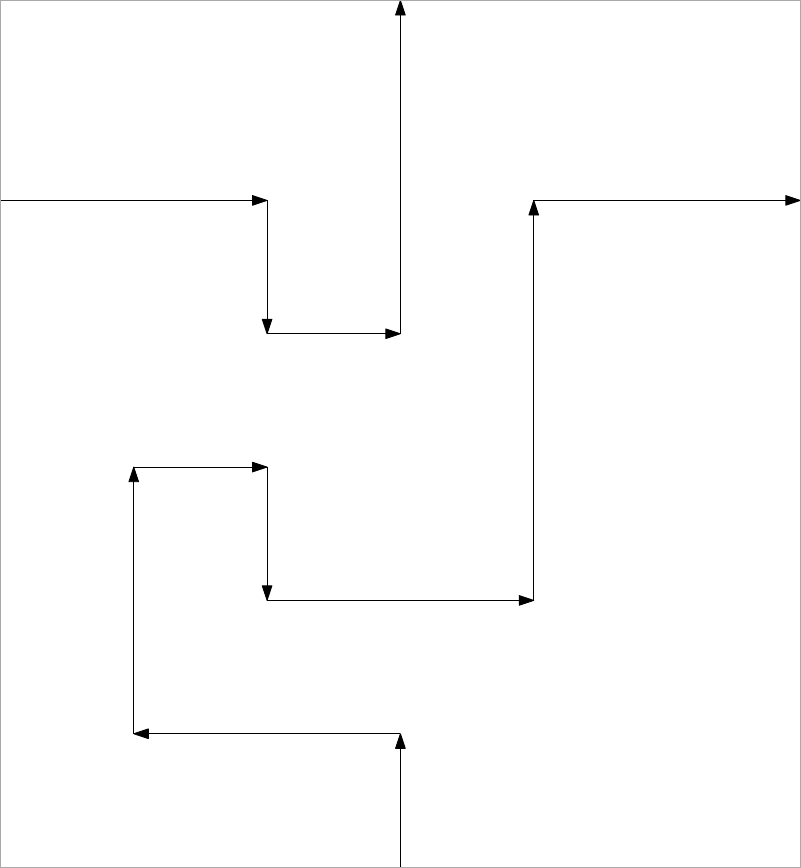}
	\caption{Implementation of the crossing gadget when the agent sees a horizontal path crossing a vertical path.}
	\label{fig:crossing}
\end{figure}

\begin{claim}[Embedding properties in crossing region]\label{clm:properties-crossing}
The special embedding ensures that in any crossing region, one of the following cases occurs:
\begin{enumerate}
    \item All four agents agree and see no Sperner solution.
    \item At least three agents see no path.
    \item Two agents see no path, one agent sees a horizontal path, and one other agent sees a vertical path.
    \item One agent sees a crossing, and the other three agents see the same single path (horizontal or vertical).
\end{enumerate}
Note that some cases are not mutually exclusive (e.g., case 1 and case 2).
\end{claim}

\begin{proof}
Note that in a crossing region every agent sees either (i) no path, (ii) a single horizontal path, (iii) a single vertical path, or (iv) both a horizontal and a vertical path, in which case the agent sees a crossing. Furthermore, by property 1 of \cref{lem:CC-IEoL}, if at least two agents see a specific path, then all four agents see that path.

Consider first the case where at least one agent, say agent $i$, sees a crossing. Then the other three agents must have an identical view, i.e., they all see either (i) no path, (ii) the same single path (horizontal or vertical), or (iii) a crossing. Indeed, since agent $i$ sees both paths, if a second agent also sees a path, then all other agents must also see it. If case (i) holds, then we are in case 2 of the claim. Case (ii) yields case 4 of the claim, and case (iii) yields case 1 of the claim.

Now consider the case where no agent sees a crossing. If all agents see no path, then clearly we are in case 1 of the claim. If an agent, say agent $i$, sees a single path (horizontal or vertical), then either (i) the other three agents see no path, (ii) the other three agents also see that single path, or (iii) one other agent sees the other path (vertical or horizontal, respectively) and the two remaining agents see no path. Note that at most one agent can see the other path, since if two agents see it, then all four agents must see it (which is not possible, since agent $i$ does not see a crossing). Case (i) yields case 2 of the claim, case (ii) yields case 1 of the claim, and, finally, case (iii) yields case 3 of the claim.
\end{proof}

\begin{claim}\label{clm:properties-crossing-detail}
In a crossing region where one agent, say agent $i$, sees a crossing, and the other three agents see the same single path (horizontal or vertical), every square satisfies (at least) one of the following properties:
\begin{enumerate}
\item All four agents agree on the labels at the four corners, and there is no Sperner solution.
\item At least three agents agree and they see the same identical label at all four corners.
\item The square lies in column $c_1^+$ or $c_2^+$, and the three agents (other than $i$) only see labels from $\{(+1,-1), (-1,-1)\}$ in the square.
\item The square lies in column $c_1^-$ or $c_2^-$, and the three agents (other than $i$) only see labels from $\{(-1,+1), (-1,-1)\}$ in the square.
\item The square is not directly adjacent to any of the columns $c_1^+, c_1^-, c_2^+, c_2^-$, and three agents see labels only from $\{(-1,+1), (+1,+1)\}$ or only from $\{(+1,-1), (+1,+1)\}$ in the square.
\end{enumerate}
\end{claim}

\begin{proof}
Since we have three agents seeing the same single path, it follows that any square outside that single path is fully labeled with the environment label $(+1,-1)$ by these three agents. As a result these squares fall into category 2. Next, any square on the single path where the three agents have identical label (i.e., the same label at all four corners) also falls into category 2. It remains to handle squares on the single path where the three agents see two different labels. By construction of the paths, the only possibilities for these two labels are: (i) $\{(-1,+1), (+1,+1)\}$, (ii) $\{(+1,-1), (+1,+1)\}$, (iii) $\{(+1,-1), (-1,-1)\}$, and (iv) $\{(-1,+1), (-1,-1)\}$.

The crucial observation is that the intricate construction of the crossing gadget ensures that for any such square on the single path, agent $i$ agrees with the other three agents on the labeling (in which case the square is in category 1), unless the square lies on:
\begin{itemize}
    \item the single vertical path, or
    \item the downward vertical portion of the horizontal (bent) path in a crossing of type (a)
    \item the upward vertical portion of the horizontal (bent) path in a crossing of type (b)
\end{itemize}
See \cref{fig:crossing-overlay-horizontal} and \cref{fig:crossing-overlay-vertical} for an illustration.
Now, if the square is of type (i) or (ii), then by construction we have that the square is not adjacent\footnote{Here we use the fact that the paths have internal width two.} to any of the columns $c_1^+, c_1^-, c_2^+, c_2^-$, and thus falls into category 5. If the square is of type (iii), then by construction it lies in column $c_1^+$ or $c_2^+$, and thus falls into category 3. Finally, if the square is of type (iv), then by construction it lies in column $c_1^-$ or $c_2^-$, and thus falls into category 4.
\end{proof}

\begin{figure}
	\centering
	\includegraphics[scale=0.5]{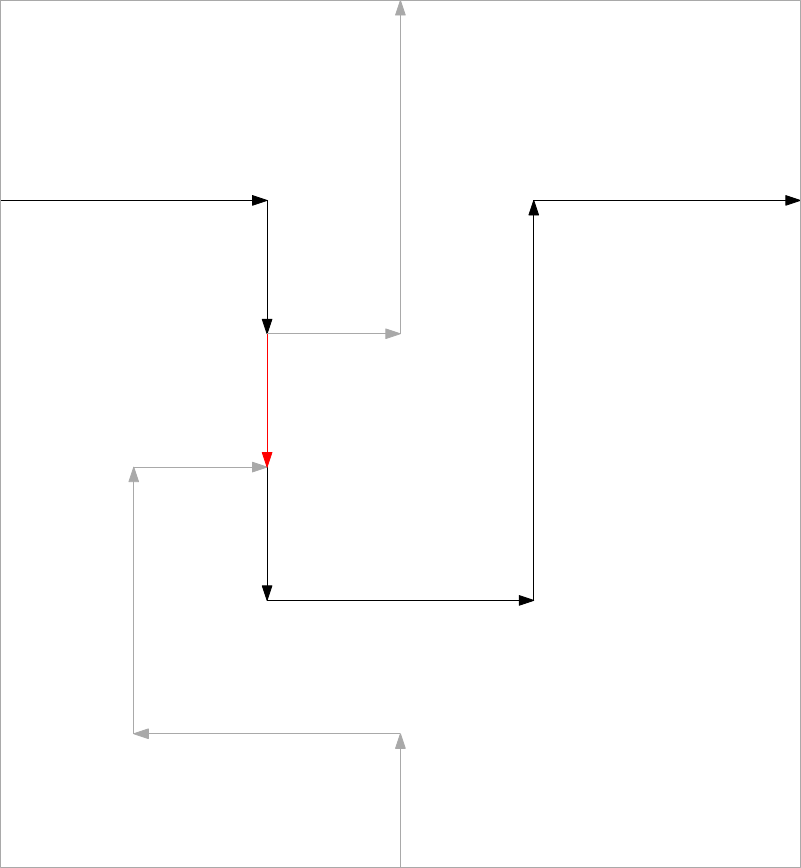}
	\caption{One agent sees a crossing, and the other three agents see a horizontal path. The black edges are the ones where all four agents agree that there is a path. The gray edges correspond to portions of paths only seen by the agent who sees a crossing. The red edge is the only portion of path that is seen by the three agents, but not by the agent who sees a crossing.}
	\label{fig:crossing-overlay-horizontal}
\end{figure}

\begin{figure}
	\centering
	\includegraphics[scale=0.5]{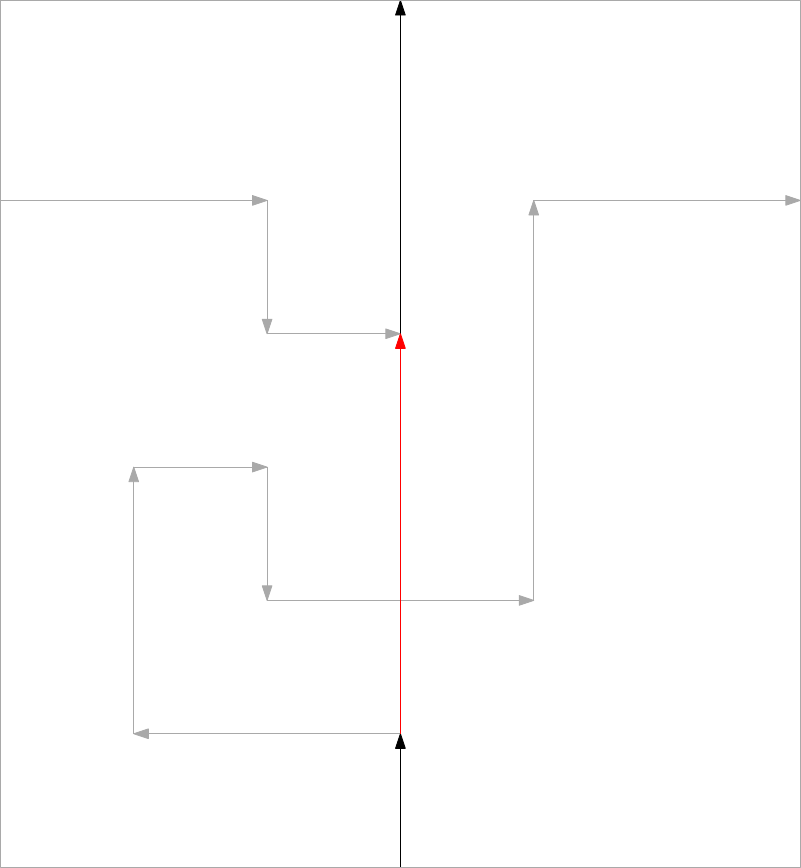}
	\caption{One agent sees a crossing, and the other three agents see a vertical path. The black edges are the ones where all four agents agree that there is a path. The gray edges correspond to portions of paths only seen by the agent who sees a crossing. The red edge is the only portion of path that is seen by the three agents, but not by the agent who sees a crossing.}
	\label{fig:crossing-overlay-vertical}
\end{figure}

\subsection{Construction of Valuations}\label{sub:valuations}

We now describe how the instance of envy-free cake-cutting is constructed. For each party $i$ of the {\sc Intersection End-of-Line} problem, we construct a corresponding agent $i$. The valuation function of agent $i$ over the cake will be constructed using only the information of party $i$.

We let $\delta := \frac{1}{10N} \leq 0.01$, where $N$ is the length of the grid $S$ in the previous section, and let $D$ denote the set $\{0, 1/10N, \dots, 1\}$, i.e., the $\delta$-discretization of $[0,1]$. For agent $i \in \{1,2,3,4\}$, its valuation function $v_i$ is constructed as follows. We first specify the values $v_i(a,b)$ for all $a,b$ in $D$. Then we use the piecewise linear interpolation described in \cref{sec:algo} to obtain a continuous valuation function.

At a high level, the valuation of an agent for an interval of the cake will depend on the Sperner labeling described in the previous section, when the endpoints of the interval are close to $0.25$ and $0.5$, or to $0.5$ and $0.75$. Otherwise, the valuation will essentially be the length of the interval.

We think of the Sperner grid from the previous section as being embedded in the grid $(D \cap [0.2,0.3]) \times (D \cap [0.45,0.55])$. Note that this grid has size exactly $(N+1) \times (N+1)$, so there is indeed a straightforward bijection with the $\{0,\dots,N\} \times \{0,\dots,N\}$ Sperner grid. The labeling $\lambda_i$ from the previous section is represented here by functions $\flabel_i$ and $\slabel_i$ that map a point in $(D \cap [0.2,0.3]) \times (D \cap [0.45,0.55])$ to the corresponding (single-coordinate) label in $\{+1,-1\}$. Namely, if $\lambda_i$ yields label $(+1,-1)$, then $\flabel_i$ outputs $+1$, and $\slabel_i$ outputs $-1$.

We let $\beta := \delta/8$, and $\gamma := \beta/8 = \delta/64$.
For $a,b \in D \subset [0,1]$, we let $v_i(a,b) := \text{length}([a,b]) = b - a$ for all $a,b \in D \subset [0,1]$, except in the following cases:
\begin{itemize}
\item When $b \in D \cap [0.2,0.3]$, we let
$$v_i(0,b) := \text{length}([0,b]) + \boost_i(b).$$
\item When $a \in D \cap [0.2,0.3]$ and $b \in D \cap [0.45,0.55]$, we let
$$v_i(a,b) := v_i(0,a) + \gamma \cdot \flabel_i(a,b).$$
\item When $a \in D \cap [0.45,0.55]$ and $b \in D \cap [0.7,0.8]$, we let
$$v_i(a,b) := v_i(b,1) + \gamma \cdot \text{second-label}_i(1-b,a).$$
\end{itemize}
In particular, note that we always have $v_i(a,1) = \text{length}([a,1]) = 1-a$ for all $a \in D$.

\paragraph*{\bf The boost function.}
The function $\text{boost}_i: D \cap [0.2,0.3] \to \{-\beta, 0, \beta\}$ is constructed as follows. The value $\boost_i(b)$ is
\begin{itemize}
    \item $+\beta$ if $b$, interpreted as an $x$-coordinate in the original Sperner grid $S$, lies in a column $c_1^+$ or $c_2^+$ of a crossing region in which agent $i$ sees a crossing.
    \item $-\beta$ if $b$, interpreted as an $x$-coordinate in the original Sperner grid $S$, lies in a column $c_1^-$ or $c_2^-$ of a crossing region in which agent $i$ sees a crossing.
    \item $0$ otherwise.
\end{itemize}
Note that $\boost_i$ is well-defined, since the construction in the previous section ensures that different crossing regions never overlap in terms of the $x$-coordinate.

\paragraph*{\bf Intuition.}
At a given division $(\ell,m,r)$ of the cake, with $\ell \in [0.2,0.3]$, $m \in [0.45,0.55]$, $r \in [0.7,0.8]$, the Sperner labels influence the agent's preferences as follows. The first label (i.e., $-1$ in label $(-1,+1)$) determines which piece is preferred between the first two pieces (i.e., the two leftmost pieces). If the first label is $-1$, then the first piece is preferred to the second piece. If the first label is $+1$, then the second piece is preferred to the first piece. Similarly, the second label determines which piece is preferred between the two last pieces (i.e., the two rightmost pieces). If the second label is $+1$, then the third piece is preferred to the fourth piece.

\paragraph*{\bf Further observations.}
Since the valuation function $v_i$ is constructed by piecewise linear interpolation, we can write for all $\ell \in [0.2,0.3]$ and $m \in [0.45,0.55]$
$$v_i(\ell,m) = v_i(0,\ell) + \gamma \cdot \flabel_i(\ell,m)$$
where we abuse notation to let $\flabel_i$ also denote the piecewise linear interpolation of $\flabel_i$. Similarly, we also have that for all $m \in [0.45,0.55]$ and $r \in [0.7,0.8]$
$$v_i(m,r) = v_i(r,1) + \gamma \cdot \slabel_i(1-r,m)$$
where again $\slabel_i$ also denotes the piecewise linear interpolation of $\slabel_i$. Note that both $\flabel_i$ and $\slabel_i$ are Lipschitz-continuous with Lipschitz constant $2/\delta$.

We can also write for all $\ell \in [0,1]$
$$v_i(0,\ell) = \ell + \boost_i(\ell)$$
where we have abused notation to let $\boost_i$ denote the piecewise linear interpolation of $\boost_i$, and where $\boost_i(\ell) = 0$ for $\ell \notin [0.2,0.3]$.

\subsection{Analysis}\label{sub:hardness-analysis}

Let $\eps := \gamma/8$. Note that $\delta > \beta > \gamma > \eps > 0$. Furthermore, note that $\eps = \Theta(1/n^4)$ and thus $\poly(n) = \poly(1/\eps)$.
In this section we show that any $\varepsilon$-envy-free allocation of the cake among the four agents with valuations $v_1,v_2,v_3,v_4$ yields a solution to the {\sc Intersection End-of-Line} problem. Let $(\ell,m,r)$ be an $\varepsilon$-envy-free division of the cake. We begin with the following claim which states that $\ell \approx 1-r$. We use the notation $x = y \pm z$ as a shorthand for $x \in [y - z, y + z]$.

\begin{claim}\label{clm:close-baseline}
It necessarily holds that $\ell = 1-r \pm 2\beta$. Furthermore, if $\boost_i(\ell) = 0$ for all agents $i \in \{1,2,3,4\}$, then we have $\ell = 1-r \pm \eps$.
\end{claim}

\begin{proof}
Note that by construction we have that for all agents $i \in \{1,2,3,4\}$, $v_i(r,1) = 1-r$, and $v_i(0,\ell) = \ell \pm \beta$. Since $(\ell,m,r)$ is an $\varepsilon$-envy-free division, there exists an agent $i$ that is allocated piece $[0, \ell]$ and thus satisfies $v_i(0,\ell) \geq v_i(r,1) - \eps$, which yields $\ell + \beta \geq 1-r-\eps$, i.e., $\ell \geq 1-r -\eps-\beta$. Similarly, there exists an agent $i$ that is allocated piece $[r, 1]$ and thus satisfies $v_i(r,1) \geq v_i(0,\ell) - \eps$, which yields $1-r \geq \ell - \beta - \eps$, i.e., $\ell \leq 1-r + \beta + \eps$. As a result, since $\eps \leq \beta$, we indeed obtain that $\ell = 1-r \pm 2\beta$.

If we additionally have that $\boost_i(\ell) = 0$ for all agents $i \in \{1,2,3,4\}$, then $v_i(0,\ell) = \ell$, and the same arguments as above yield that $\ell = 1-r \pm \eps$.
\end{proof}

The analysis now proceeds as follows. In the next section, we prove that $(\ell,m,r)$ cannot be $\eps$-envy-free if it lies outside the embedding region. Next, in \cref{sec:analysis-toolbox}, we prepare a toolbox for arguing that there are no unwanted solutions inside the embedding region, which we then proceed to prove in \cref{sec:analysis-inside}. Together, these three sections show that any $\varepsilon$-envy-free allocation must yield a solution to the {\sc Intersection End-of-Line} instance.

\subsubsection{Analysis: No solutions outside the embedding region}\label{sec:analysis-outside}

In this section we show that if $(\ell,m,r)$ is an $\varepsilon$-envy-free division of the cake, then it cannot lie outside the embedding region.

\begin{claim}\label{clm:no-solutions-outside-embedding}
It necessarily holds that $\ell \in [0.2,0.3]$, $m \in [0.45,0.55]$, and $r \in [0.7,0.8]$.
\end{claim}

\begin{proof}
We consider the following four cases, and show that in each case the division cannot be $\eps$-envy-free
\begin{itemize}
    \item Regime 1: $\ell \leq 0.2 + \delta$ and $m \geq 0.5$, or $\ell \leq 0.25$ and $m \geq 0.55$
    \item Regime 2: $\ell \leq 0.2 + \delta$ and $m \leq 0.5$, or $\ell \leq 0.25$ and $m \leq 0.45$
    \item Regime 3: $\ell \geq 0.3 - \delta$ and $m \geq 0.5$, or $\ell \geq 0.25$ and $m \geq 0.55$
    \item Regime 4: $\ell \geq 0.3 - \delta$ and $m \leq 0.5$, or $\ell \geq 0.25$ and $m \leq 0.45$
\end{itemize}
It thus follows that we must have $\ell \in [0.2 + \delta, 0.3 - \delta]$ and $m \in [0.45,0.55]$. By \cref{clm:close-baseline} we also necessarily have $1-r = \ell \pm 2\beta$, and since $2\beta \leq \delta$, it follows that we also have $r \in [0.7,0.8]$, as desired. In the remainder of this proof we now show that there is no $\eps$-envy-free division in each of the four regimes.

\textbf{Regime 1.} For any pair $(\ell,m)$ in Regime 1 that also lies on the $D \times D$ grid, we show below that $v_i(\ell,m) > v_i(0,\ell) + \eps$ for all agents $i \in \{1,2,3,4\}$. Since the valuations are constructed by piecewise linear interpolation, it then follows that $v_i(\ell,m) > v_i(0,\ell) + \eps$ also holds for points off the grid, i.e, for all $(\ell,m)$ in Regime 1. As a result, the division cannot be $\eps$-envy-free. It remains to prove that we indeed have $v_i(\ell,m) > v_i(0,\ell) + \eps$ for all grid points in Regime 1. If $(\ell,m)$ lies on a grid point in Regime 1 inside the embedding region, then by construction of the labels on the boundary of the embedding region we must have $\flabel_i(\ell,m) = +1$, because the top and left boundary of the Sperner instance have first label $+1$. As a result
$$v_i(\ell,m) = v_i(0,\ell) + \gamma \cdot \flabel_i(\ell,m) = v_i(0,\ell) + \gamma > v_i(0,\ell) + \eps$$
since $\gamma > \eps$. On the other hand, if $(\ell,m)$ lies on a grid point in Regime 1 outside the embedding region, then we consider two cases. If $\ell \leq 0.2 + \delta$ and $m \geq 0.5$, then
$$v_i(\ell,m) = m - \ell \geq 0.5 - 0.2 - \delta > 0.25 > 0.2 + \delta + \beta + \eps \geq \ell + \beta + \eps \geq v_i(0,\ell) + \eps$$
where we used $v_i(0,\ell) \leq \ell + \beta$, and $\delta + \beta + \eps < 0.05$. If $\ell \leq 0.25$ and $m \geq 0.55$, then
$$v_i(\ell,m) = m - \ell \geq 0.55 - 0.25 = 0.3 > 0.25 + \beta + \eps \geq \ell + \beta + \eps \geq v_i(0,\ell) + \eps$$
where we used $v_i(0,\ell) \leq \ell + \beta$, and $\beta + \eps < 0.05$.

\textbf{Regime 4.} We can argue that there is no solution in Regime 4 by using very similar arguments. For all grid points $(\ell,m)$ inside the embedding region, we have $\flabel_i(\ell,m) = -1$ for all agents $i \in \{1,2,3,4\}$ by construction of the labels on the boundary (because the bottom and right boundary of the Sperner instance have first label $-1$.). As a result,
$$v_i(\ell,m) = v_i(0,\ell) + \gamma \cdot \flabel_i(\ell,m) = v_i(0,\ell) - \gamma < v_i(0,\ell) - \eps$$
since $\gamma > \eps$. On the other hand, if $(\ell,m)$ is a grid point in Regime 4 outside the embedding region, we again consider two cases. If $\ell \geq 0.3 - \delta$ and $m \leq 0.5$, then
$$v_i(\ell,m) = m - \ell \leq 0.5 - 0.3 + \delta < 0.25 < 0.3 - \delta - \beta - \eps \leq \ell - \beta - \eps \leq v_i(0,\ell) - \eps$$
where we used $v_i(0, \ell) \geq \ell - \beta$, and $\delta + \beta + \eps < 0.05$. If $\ell \geq 0.25$ and $m \leq 0.45$, then
$$v_i(\ell,m) = m - \ell \leq 0.45 - 0.25 = 0.2 < 0.25 - \beta - \eps \leq \ell - \beta - \eps \leq v_i(0, \ell) - \eps$$
where we used $v_i(0, \ell) \geq \ell - \beta$, and $\beta + \eps < 0.05$.

\textbf{Regime 2.} Next, we consider Regime 2. By \cref{clm:close-baseline}, $\ell \leq 0.2 + \delta$ implies $r \geq 1 - \ell - 2\beta \geq 0.8 - \delta -2\beta \geq 0.8 -2\delta$, and similarly $\ell \leq 0.25$ implies $r \geq 0.75 -2\beta \geq 0.75 - \delta$, since $\delta \geq 2\beta$. Thus, it suffices to show that for any pair $(m,r)$ satisfying
\begin{itemize}
    \item Regime $2'$: $r \geq 0.8 - 2\delta$ and $m \leq 0.5$, or $r \geq 0.75 - \delta$ and $m \leq 0.45$
\end{itemize}
the division cannot be $\eps$-envy-free. As before, it will suffice to argue that $v_i(m,r) > v_i(r,1) + \eps$ holds for all agents $i \in \{1,2,3,4\}$ and for all grid points $(m,r) \in D \times D$ in Regime $2'$, since the valuations are constructed by piecewise linear interpolation. This will be enough to show that no agent is happy with piece $[r,1]$, and thus the division cannot be $\eps$-envy-free. Consider first any grid point $(m,r)$ in Regime $2'$ such that $(1-r,m)$ lies inside the embedding region (i.e., $1-r \in [0.2,0.3]$ and $m \in [0.45,0.55]$). By construction of the labels on the boundary we have that $\slabel_i(1-r,m) = +1$ (because the bottom and left boundary of the Sperner instance have second label $+1$.), and thus
$$v_i(m,r) = v_i(r,1) + \gamma \cdot \slabel_i(1-r,m) = v_i(r,1) + \gamma > v_i(r,1) + \eps$$
since $\gamma > \eps$. On the other hand, if $(m,r)$ is a grid point in Regime $2'$ such that $(1-r,m)$ lies outside the embedding region, then we have two cases. If $r \geq 0.8 - 2\delta$ and $m \leq 0.5$, then
$$v_i(m,r) = r - m \geq 0.8 - 2\delta - 0.5 > 0.25 > 0.2 + 2\delta + \eps \geq 1 - r + \eps = v_i(r,1) + \eps$$
since $2\delta + \eps < 0.05$. If $r \geq 0.75 - \delta$ and $m \leq 0.45$, then
$$v_i(m,r) = r - m \geq 0.75 - \delta - 0.45 = 0.3 - \delta > 0.25 + \delta + \eps \geq 1 - r + \eps = v_i(r,1) + \eps$$
since $2\delta + \eps < 0.05$.

\textbf{Regime 3.} We can argue that there is no solution in Regime 3 using very similar arguments. By \cref{clm:close-baseline} we have that if $(\ell,m)$ is in Regime 3, then $(m,r)$ must satisfy
\begin{itemize}
    \item Regime $3'$: $r \leq 0.7 + 2\delta$ and $m \geq 0.5$, or $r \leq 0.75 + \delta$ and $m \geq 0.55$
\end{itemize}
For any grid point $(m,r)$ in Regime $3'$ such that $(1-r,m)$ lies inside the embedding region, we have by construction of the labels on the boundary that $\slabel_i(1-r,m) = -1$ for all agents $i \in \{1,2,3,4\}$ (because the top and right boundary of the Sperner instance have second label $-1$.), and thus
$$v_i(m,r) = v_i(r,1) + \gamma \cdot \slabel_i(1-r,m) = v_i(r,1) - \gamma < v_i(r,1) - \eps$$
since $\gamma > \eps$.
On the other hand, for any grid point $(m,r)$ in Regime $3'$ such that $(1-r,m)$ lies outside the embedding region, we have two cases. If $r \leq 0.7 + 2\delta$ and $m \geq 0.5$, then
$$v_i(m,r) = r - m \leq 0.7 + 2\delta - 0.5 < 0.25 < 0.3 - 2\delta - \eps \leq 1 - r - \eps = v_i(r,1) - \eps$$
since $2\delta + \eps < 0.05$. If $r \leq 0.75 + \delta$ and $m \geq 0.55$, then
$$v_i(m,r) = r - m \leq 0.75 + \delta - 0.55 = 0.2 + \delta < 0.25 - \delta - \eps \leq 1 - r - \eps = v_i(r,1) - \eps$$
since $2\delta + \eps < 0.05$.
\end{proof}

\subsubsection{Analysis: Toolbox}\label{sec:analysis-toolbox}

By the previous section, we can now assume that $(\ell,m,r)$ is an $\varepsilon$-envy-free division with $\ell \in [0.2,0.3]$, $m \in [0.45,0.55]$, and $r \in [0.7,0.8]$. In particular, $(\ell,m)$ is a point in the Sperner embedding and it lies inside a square of the Sperner discretization (or in multiple squares if it lies on the boundary between different squares). The following claims provide a toolbox for arguing that, depending on the labels at the corners of the square, the division cannot be $\varepsilon$-envy-free in various cases.

We begin with the following claim which will be heavily used in this section.

\begin{claim}\label{clm:close-label}
It holds that
$$|\slabel_i(1-r,m) - \slabel_i(\ell,m)| \leq 1/2.$$
\end{claim}

\begin{proof}
Given that $\slabel_i$ is Lipschitz-continuous with Lipschitz constant $2/\delta$, by \cref{clm:close-baseline} we have that
$$|\slabel_i(1-r,m) - \slabel_i(\ell,m)| \leq \frac{2}{\delta} |1-r-\ell| \leq \frac{4\beta}{\delta} \leq 1/2$$
since $\beta \leq \delta/8$.
\end{proof}

\begin{claim}\label{clm:all-agents-agree}
If $(\ell,m)$ lies in a square of the Sperner embedding such that all four agents agree on the labels of the four corners of the square, and these labels do not yield a Sperner solution, then $(\ell,m,r)$ cannot be $\varepsilon$-envy-free.
\end{claim}

This claim is a special case of the following stronger claim.

\begin{claim}\label{clm:all-agents-same-partial-label}
If $(\ell,m)$ lies in a square of the Sperner embedding such that the first or second label remains constant across all four agents and at all four corners (i.e., all the labels lie in one of the four following sets: $\{(+1,+1), (+1,-1)\}$, $\{(+1,+1), (-1,+1)\}$, $\{(-1,-1), (-1,+1)\}$, or $\{(-1,-1), (+1,-1)\}$), then $(\ell,m,r)$ cannot be $\varepsilon$-envy-free.
\end{claim}

\begin{proof}
Let us consider the case where for all agents the labels lie in $\{(+1,+1), (+1,-1)\}$ at all four corners of the square containing $(\ell,m)$. Then, it must be that $\flabel_i(\ell,m) = +1$ for all agents $i$, and, as a result, $v_i(\ell,m) = v_i(0,\ell) + \gamma > v_i(0,\ell) + \varepsilon$, since $\gamma > \eps$. But this means that no agent is happy with the leftmost piece and thus the division cannot be $\varepsilon$-envy-free. The case $\{(-1,-1), (-1,+1)\}$ is handled very similarly.

Next, we consider the case where for all agents the labels lie in $\{(+1,+1), (-1,+1)\}$ at all four corners of the square containing $(\ell,m)$. Then, we necessarily have $\slabel_i(\ell,m) = +1$ for all agents $i$. By \cref{clm:close-label}, we have $\slabel_i(1-r,m) \geq 1/2$, and thus $v_i(m,r) = v_i(r,1) + \gamma/2 > v_i(r,1) + \varepsilon$, since $\gamma > 2\eps$. This means that no agent is happy with the rightmost piece and thus the division cannot be $\varepsilon$-envy-free. The case $\{(-1,-1), (+1,-1)\}$ is handled very similarly.
\end{proof}

\begin{claim}\label{clm:3-agents-identical}
If $(\ell,m)$ lies in a square of the Sperner embedding such that for at least three agents all four corners of the square have the same identical label, then $(\ell,m,r)$ cannot be $\varepsilon$-envy-free.
\end{claim}

\begin{proof}
It is easy to see that, in this case, there are two pieces that the three agents do not like. For example, if agents 1, 2, and 3 see identical label $(+1,-1)$ at all four corners, then for all $i \in \{1,2,3\}$ we have $\flabel_i(\ell,m) = +1$, and $\slabel_i(1-r,m) \leq -1/2$ (since $\slabel_i(\ell,m) = -1$ and by using \cref{clm:close-label}). This implies that for all three agents $i \in \{1,2,3\}$, $v_i(\ell,m) = v_i(0,\ell) + \gamma > v_i(0,\ell) + \eps$, and $v_i(m,r) \leq v_i(r,1) -\gamma/2 < v_i(r,1) - \eps$, since $\gamma > 2\eps$. In other words, none of the three agents is happy with piece $[0,\ell]$ or piece $[m,r]$, and thus the division cannot be $\varepsilon$-envy-free. The arguments are very similar for the other three possible labels.
\end{proof}

\begin{claim}\label{clm:2-agents-environment-no-boost}
If $(\ell,m)$ lies in a square of the Sperner embedding such that for at least two agents all four corners of the square have the environment label $(+1,-1)$, and if the boost function is zero at all four corners for all four agents, then $(\ell,m,r)$ cannot be $\varepsilon$-envy-free.
\end{claim}

\begin{proof}
Without loss of generality, let agents 1 and 2 see the environment label $(+1,-1)$ at all four corners. In that case, we must have that for $i \in \{1,2\}$, $\flabel_i(\ell,m) = +1$ and $\slabel_i(1-r,m) \leq -1/2$ (since $\slabel_i(\ell,m) = -1$ and by using \cref{clm:close-label}). This implies that $v_i(\ell,m) = v_i(0,\ell) + \gamma > v_i(0,\ell) + \eps$, and $v_i(m,r) \leq v_i(r,1) -\gamma/2 < v_i(r,1) - \eps$, since $\gamma > 2\eps$. In particular, the two agents strictly prefer piece $[\ell,m]$ over piece $[0,\ell]$, and also strictly prefer piece $[r,1]$ over piece $[m,r]$. As a result, one of these two agents, say agent 1, must be allocated piece $[r,1]$, which implies that, for the division to be $\eps$-envy-free, we must have
$$1-r = v_1(r,1) \geq v_1(\ell,m) - \eps = v_1(0,\ell) + \gamma - \eps.$$
Now, given that agent 1 has boost function zero everywhere in the square, it follows that $v_1(0,\ell) = \ell$, and thus $1-r \geq \ell + \gamma - \eps > \ell + \eps$, since $\gamma > 2\eps$. But since all four agents have boost function equal to zero in the square, \cref{clm:close-baseline} yields that $\ell = 1-r \pm \eps$, a contradiction.
\end{proof}

\begin{claim}\label{clm:3-agents-positive-no-boost}
If $(\ell,m)$ lies in a square of the Sperner embedding such that three agents only see labels from $\{(-1,+1), (+1,+1)\}$ or only from $\{(+1,-1), (+1,+1)\}$, and if the boost function is zero at all four corners for all four agents, then $(\ell,m,r)$ cannot be $\varepsilon$-envy-free.
\end{claim}

\begin{proof}
Without loss of generality, let agents 1, 2, and 3 only see labels from $\{(-1,+1), (+1,+1)\}$ at all four corners. Then, it must be that for all agents $i \in \{1,2,3\}$, $\slabel_i(\ell,m) = +1$, and thus $\slabel_i(1-r,m) \geq 1/2$ by \cref{clm:close-label}. It follows that $v_i(m,r) \geq v_i(r,1) + \gamma/2 > v_i(r,1) + \eps$, and thus all three agents strictly prefer piece $[m,r]$ to piece $(r,1)$. As a result, for the division to be $\eps$-envy-free, one of the three agents, say agent 1, has to be allocated piece $[0,\ell]$, and satisfy
$$v_1(0,\ell) \geq v_1(m,r) - \eps \geq v_1(r,1) + \gamma/2 - \eps = 1-r + \gamma/2 - \eps.$$
Given that agent 1 has boost function zero everywhere in the square, we have $v_1(0,\ell) = \ell$, and thus $\ell \geq 1-r + \gamma/2 - \eps > 1-r + \eps$, since $\gamma > 4\eps$. Since all four agents have boost function equal zero in the square, \cref{clm:close-baseline} yields $\ell = 1-r \pm \eps$, a contradiction. The case where the three agents only see labels from $\{(+1,-1), (+1,+1)\}$ at all four corners is handled using the same arguments.
\end{proof}

\begin{claim}\label{clm:3-agents-negative-boost}
If $(\ell,m)$ lies in a square of the Sperner embedding such that one of the following two cases holds
\begin{enumerate}
\item The boost function is $+\beta$ at all four corners for some agent, and the other three agents have boost zero at all four corners and only labels in $\{(+1,-1), (-1,-1)\}$.
\item The boost function is $-\beta$ at all four corners for some agent, and the other three agents have boost zero at all four corners and only labels in $\{(-1,+1), (-1,-1)\}$.
\end{enumerate}
then $(\ell,m,r)$ cannot be $\varepsilon$-envy-free.
\end{claim}

\begin{proof}
Without loss of generality, let agents 1, 2, and 3 have boost zero in the square, and only see labels in $\{(+1,-1), (-1,-1)\}$, while agent 4 has boost $+\beta$ in the square. For $i \in \{1,2,3\}$, we have $\slabel_i(\ell,m) = -1$, and thus $\slabel_i(1-r,m) \leq -1/2$ by \cref{clm:close-label}. This means that $v_i(m,r) \leq v_i(r,1) - \gamma/2 < v_i(r,1) - \eps$, since $\gamma > 2\eps$. As a result, all three agents strictly prefer piece $[r,1]$ to piece $[m,r]$. In particular, for the division to be $\eps$-envy-free, agent 4 must necessarily be allocated piece $[m,r]$, and satisfy $v_4(m,r) \geq v_4(0,\ell) - \eps$.

Furthermore, one of the three agents, say agent 1, must be allocated piece $[0,\ell]$ and thus satisfy $v_1(0,\ell) \geq v_1(r,1) - \eps = 1-r - \eps$. Since agent 1 has boost zero in the square, we have $v_1(0,\ell) = \ell$, and thus deduce that $\ell \geq 1-r - \eps$.

Note that by construction of the valuation function we have $v_4(m,r) \leq v_4(r,1) + \gamma = 1-r + \gamma$. On the other hand, given that agent 4 has boost $+\beta$ in the square, we also have $v_4(0,\ell) = \ell + \beta$. Together, we obtain that
$$v_4(m,r) \leq 1-r + \gamma \leq \ell + \eps + \gamma = v_4(0,\ell) - \beta + \eps + \gamma < v_4(0,\ell) - \eps$$
where we used $\beta > \gamma + 2\eps$. This means that agent 4 is not happy with its allocated piece, a contradiction.

For the case where the three agents have labels in $\{(-1,+1), (-1,-1)\}$ only, using similar arguments we can deduce that all three agents strictly prefer piece $[0,\ell]$ to piece $[\ell,m]$. In particular, agent 4 must be allocated piece $[\ell,m]$ and satisfy $v_4(\ell,m) \geq v_4(1,r) - \eps$. Since one of the three agents must be allocated piece $[r,1]$, and that agent has boost zero in the square, we can deduce that $1-r \geq \ell - \eps$. Using the fact that agent 4 has boost $-\beta$ in the square, we can now write
$$v_4(\ell,m) \leq v_4(0,\ell) + \gamma = \ell - \beta + \gamma \leq 1-r + \eps - \beta + \gamma = v_4(r,1) + \eps - \beta + \gamma < v_4(r,1) - \eps$$
where we used $\beta > \gamma + 2\eps$. This means that agent 4 is not happy with its allocated piece, a contradiction.
\end{proof}

\subsubsection{Analysis: No unwanted solutions inside the embedding region}\label{sec:analysis-inside}

In this last section of the proof, we use the toolbox created in the previous section, together with various properties of the construction of the embedding, to argue that no unwanted solutions occur inside the embedding region.

\paragraph*{\bf No solutions outside crossing and vertex regions.}
If $(\ell,m)$ lies in a square outside any crossing or vertex region, then, by \cref{clm:properties-outside}, one of the following two cases occurs:
\begin{enumerate}
    \item All four agents agree on the labels of the four corners of the square, and these labels do not correspond to a Sperner solution. By \cref{clm:all-agents-agree}, $(\ell,m,r)$ cannot be $\varepsilon$-envy-free.
    \item At least three agents see the environment label at all four corners of the square. By \cref{clm:3-agents-identical}, $(\ell,m,r)$ cannot be $\varepsilon$-envy-free.
\end{enumerate}
Thus, we deduce that if $(\ell,m,r)$ is $\varepsilon$-envy-free, $(\ell,m)$ must lie in a crossing region or in a vertex region.

\paragraph*{\bf No solutions in crossing regions.}
If $(\ell,m)$ lies in a square inside a crossing region, then, by \cref{clm:properties-crossing}, one of the following four cases occurs:
\begin{enumerate}
    \item All four agents agree on the labels of the four corners of the square, and these labels do not correspond to a Sperner solution. By \cref{clm:all-agents-agree}, $(\ell,m,r)$ cannot be $\varepsilon$-envy-free.
    \item At least three agents see the environment label at all four corners of the square. By \cref{clm:3-agents-identical}, $(\ell,m,r)$ cannot be $\varepsilon$-envy-free.
    \item Two agents see no path in the crossing region, and the other two agents see different single paths. In particular, by construction, the boost function is zero for all four agents in the crossing region, since no agent sees a crossing. As a result, we have two agents who see the environment label at all four corners of the square, and the boost function is zero at all four corners for all four agents. By \cref{clm:2-agents-environment-no-boost}, $(\ell,m,r)$ cannot be $\varepsilon$-envy-free.
    \item One agent, say agent $i$, sees a crossing, and the other three agents see the same single path (horizontal or vertical). In that case, by \cref{clm:properties-crossing-detail}, one of the following cases occurs:
    \begin{enumerate}
\item All four agents agree on the labels at the four corners, and there is no Sperner solution. By \cref{clm:all-agents-agree}, $(\ell,m,r)$ cannot be $\varepsilon$-envy-free.
\item At least three agents agree and they see the same identical label at all four corners. By \cref{clm:3-agents-identical}, $(\ell,m,r)$ cannot be $\varepsilon$-envy-free.
\item The square lies in column $c_1^+$ or $c_2^+$, and the three agents (other than $i$) only see labels from $\{(+1,-1), (-1,-1)\}$ in the square. By construction, agent $i$ has boost $+\beta$ at all four corners, while the other agents have boost $0$ at all corners. Thus, by \cref{clm:3-agents-negative-boost}, $(\ell,m,r)$ cannot be $\varepsilon$-envy-free.
\item The square lies in column $c_1^-$ or $c_2^-$, and the three agents (other than $i$) only see labels from $\{(-1,+1), (-1,-1)\}$ in the square. By construction, agent $i$ has boost $-\beta$ at all four corners, while the other agents have boost $0$ at all corners. Thus, by \cref{clm:3-agents-negative-boost}, $(\ell,m,r)$ cannot be $\varepsilon$-envy-free.
\item The square is not incident on any of the columns $c_1^+, c_1^-, c_2^+, c_2^-$, and three agents see labels only from $\{(-1,+1), (+1,+1)\}$ or only from $\{(+1,-1), (+1,+1)\}$ in the square. By construction, all four agents have boost $0$ at all corners. Thus, by \cref{clm:3-agents-positive-no-boost}, $(\ell,m,r)$ cannot be $\varepsilon$-envy-free.
\end{enumerate}
\end{enumerate}

\paragraph*{\bf No solutions in non-end-of-line vertex regions.}
If $(\ell,m)$ lies in a square inside a vertex region, then by \cref{clm:properties-vertex-details}, one of the following cases occurs:
\begin{enumerate}
\item All four agents agree on the labels at the four corners, and there is no Sperner solution. By \cref{clm:all-agents-agree}, $(\ell,m,r)$ cannot be $\varepsilon$-envy-free.
\item At least three agents agree and they see the same identical label at all four corners. By \cref{clm:3-agents-identical}, $(\ell,m,r)$ cannot be $\varepsilon$-envy-free.
\item At least three agents see labels only in $\{(-1,+1), (+1,+1)\}$, or only in $\{(+1,-1), (+1,+1)\}$, and the square touches the connector region. Then, all four agents have boost $0$ at all corners, since by construction the connector region does not overlap with any crossing region in terms of the $x$-coordinates. Thus, by \cref{clm:3-agents-positive-no-boost}, $(\ell,m,r)$ cannot be $\varepsilon$-envy-free.
\item All four agents see labels only in $\{(+1,-1), (-1,-1)\}$, or only in $\{(-1,+1), (-1,-1)\}$. By \cref{clm:all-agents-same-partial-label}, $(\ell,m,r)$ cannot be $\varepsilon$-envy-free.
\end{enumerate}

As a result, if $(\ell,m,r)$ is $\eps$-envy-free, then $(\ell,m)$ must lie in an end-of-line vertex region. This completes the proof.

\subsection{Obtaining \ppad/-hardness and a Query Lower Bound for Identical Agents}\label{sec:hardness-PPAD-query}

The same construction can be used to show \ppad/-hardness and a query lower bound for the problem with identical agents. To do this, we reduce from the standard \textsc{End-of-Line} problem which is known to be \ppad/-complete and to require $\poly(n)$ queries. Every agent sees the same \textsc{End-of-Line} instance and thus has the same valuation function. Note that most of the arguments become much simpler, because all the agents agree on the labeling everywhere. Furthermore, there is no need to ever overwrite the connector region with label $(-1,-1)$, since all nodes have in-degree and out-degree at most one. It is easy to check that the construction can be implemented in a query-efficient way and in polynomial time, and thus yields a valid reduction in the query complexity and white-box settings.

\appendix

\section{Envy-free cake-cutting: Reduction from $n$ to $n+1$ agents}\label{app:sec:reduction-to-more-agents}

Below, we present a reduction from $n$ to $n+1$ agents for the envy-free cake-cutting problem. The reduction applies to the setting with general valuations, and also to the setting with monotone valuations. Whether there is a way to also achieve this for additive valuations is an interesting open question.

\begin{lemma}
The connected $\eps$-envy-free cake-cutting problem with $n$ agents with general valuations reduces to the same problem with $n+1$ agents. This reduction is efficient in all three models (query complexity, communication complexity, and computational complexity). Furthermore, the reduction maps an instance with $n$ monotone agents to an instance with $n+1$ monotone agents. Finally, for the query and computational complexity models, there is also a version of the reduction that reduces instances with $n$ identical agents to $n+1$ identical agents.
\end{lemma}

\begin{proof}
Let $v_1, \dots, v_n$ denote the $1$-Lipschitz continuous valuation functions of the $n$ agents. For $i \in [n]$ and $a,b \in [0,1]$, we let
$$v_i'(a,b) := \frac{1}{3} v_i\Big([2a,2b] \cap [0,1]\Big) + \frac{2}{3} \phi\Big(\big|[a,b] \cap [1/2,1]\big|\Big)$$
and
$$v_{n+1}'(a,b) := \frac{2}{3} \phi\Big(\big|[a,b] \cap [1/2,1]\big|\Big)$$
where $\phi: [0,1/2] \to [0,1]$ is defined as
\begin{equation*}
\phi(t) = \left\{\begin{tabular}{lc}
    $0$ & if $t \leq 1/3$ \\
    $6(t-1/3)$ & if $t \geq 1/3$
\end{tabular} \right.
\end{equation*}
Note that $v_1', \dots, v_{n+1}'$ are $5$-Lipschitz-continuous valuation functions (which can easily be normalized to be $1$-Lipschitz continuous). Furthermore, if $v_1, \dots, v_n$ are monotone, then so are $v_1', \dots, v_{n+1}'$. Finally, the new valuations can easily be constructed/simulated from the original ones in all three models of computation.

Consider any $\eps$-envy-free allocation $A_1', \dots, A_n', A_{n+1}'$ with respect to $v_1', \dots, v_{n+1}'$. We claim that, without loss of generality (and possibly by switching to a $2\eps$-envy-free allocation), the piece $A_{n+1}'$ obtained by agent $n+1$ is the rightmost piece, i.e., $A_{n+1}' = [a,1]$ for some $a \in [0,1]$. Indeed, if $A_{n+1}'$ is empty, then this trivially holds. Otherwise, $A_{n+1}' = [a,b]$ with $a < b < 1$. Let $i$ be an agent who gets the rightmost piece, i.e., $A_i = [c,1]$ for some $c \geq b$. We distinguish between the following two cases:
\begin{itemize}
    \item $\phi(|[c,1] \cap [1/2,1]|) > 3\eps/2$: In that case, it must be that $c < 2/3$, and thus $b < 2/3$, which implies $v_{n+1}'(a,b) = (2/3) \cdot \phi(|[a,b] \cap [1/2,1]|) = 0$. Since $v_{n+1}'(c,1) = (2/3) \cdot \phi(|[c,1] \cap [1/2,1]|) > \eps$, this contradicts the $\eps$-envy-freeness of the division.
    \item $\phi(|[c,1] \cap [1/2,1]|) \leq 3\eps/2$: In that case, $c \geq 1/2$ and thus $v_i'(c,1) = v_{n+1}'(c,1) = (2/3) \cdot \phi(|[c,1] \cap [1/2,1]|) \leq \eps$. By construction of $\phi$, at most one of $v_{n+1}'(a,b)$ and $v_{n+1}'(c,1)$ can be strictly positive. Thus, if $v_{n+1}'(a,b) > \eps$, then $v_{n+1}'(c,1) = 0$. It follows that $v_i'(c,1) = 0$ and $v_i'(a,b) > \eps$, a contradiction to $\eps$-envy-freeness. As a result, we must also have $v_{n+1}'(a,b) \leq \eps$. By $\eps$-envy-freeness we obtain that $v_{n+1}'(A_j) \leq 2\eps$ and $v_i'(A_j) \leq 2\eps$ for all $j \in [n+1]$. As a result, assigning $A_i$ to agent $n+1$, and $A_{n+1}$ to agent $i$ yields a $2\eps$-envy-free allocation in which agent $n+1$ has the rightmost piece.
\end{itemize}

We can thus assume without loss of generality that we have an $\eps$-envy-free allocation $A_1', \dots, A_n',$ $A_{n+1}'$ with respect to $v_1', \dots, v_{n+1}'$, where $A_{n+1}'$ is the rightmost piece. We define the allocation $A_1, \dots, A_n$ by letting $A_i := 2 \cdot (A_i' \cap [0,1/2])$. Note that this is indeed a partition of $[0,1]$, because $A_{n+1}' \subset [1/2,1]$. Otherwise, the rightmost piece $A_{n+1}'$ would have value at least $2/3$ to all agents, while any other piece would have value at most $1/3$.

Finally, we argue that the allocation $A_1, \dots, A_n$ must be $6\eps$-envy-free with respect to $v_1, \dots, v_n$. First of all, note that $v_{n+1}'(A_j') \leq \eps$ for all $j \in [n]$. Otherwise, if $v_{n+1}'(A_j') > \eps$ for some $j$, then $v_{n+1}'(A_{n+1}') = 0$, and agent $n+1$ would be envious. Thus, it follows that $v_i'(A_j') \in [(1/3) \cdot v_i(A_j), (1/3) \cdot v_i(A_j) + \eps]$ for all $i,j \in [n]$. As a result, for all $i,j \in [n]$ we have
$$v_i(A_i) \geq 3(v_i'(A_i') - \eps) \geq 3(v_i'(A_j') - 2\eps) \geq v_i(A_j) - 6\eps$$
as desired.

For the setting of identical valuations we proceed similarly. Given a valuation $v$ shared by $n$ agents, we define
$$v'(a,b) := \frac{1}{3} v\Big([2a,2b] \cap [0,1]\Big) + \frac{2}{3} \phi\Big(\big|[a,b] \cap [1/2,1]\big|\Big)$$
to be shared by $n+1$ agents. Consider an $\eps$-envy-free allocation $A_1', \dots, A_n', A_{n+1}'$, and rename the agents if needed such that $A_{n+1}'$ is the rightmost piece. As above, we can argue that $A_{n+1}' \subset [1/2,1]$, and define an allocation $A_1, \dots, A_n$ of $[0,1]$. Note that $v'(A_{n+1}') = (2/3) \cdot \phi(|A_{n+1}' \cap [1/2,1]|)$. If $v'(A_{n+1}') > 0$, then $v'(A_j') = (1/3) \cdot v(A_j)$ for all $j \in [n]$, and we argue as above. If $v'(A_{n+1}') = 0$, then $v'(A_j') \leq \eps$ for all $j \in [n]$, and thus $v(A_j) \leq 3\eps$, and the division is $3\eps$-envy-free.
\end{proof}

\section{Communication Protocols}\label{app:communication}

Brânzei and Nisan~\cite{BranzeiN19-cake-communication} proved that for three agents with additive valuations an $\eps$-envy-free allocation can be found using $O(\log(1/\eps))$ communication. This is achieved by a reduction to the \emph{monotone-crossing} problem, which they introduce and for which they prove a $O(\log(1/\eps))$ communication bound. Their approach can easily be extended to agents with monotone valuations.

\begin{lemma}
For three agents with monotone valuations, there exists a $O(\log(1/\eps))$ communication protocol that finds an $\eps$-envy-free connected allocation.
\end{lemma}

For completeness, we provide a proof sketch.

\begin{proof}[Proof sketch]
In the first step of the protocol, our goal is to find positions $a < b$ such that one agent is indifferent between the three pieces resulting from cutting at $a$ and $b$, while at least one of the other two agents prefers the middle piece $[a,b]$. In the second step, given such $a$ and $b$, we then use the monotone-crossing problem to obtain a solution.

For the first step, every agent $i$ communicates positions $\ell_i$ and $r_i$ such that $v_i(0,\ell_i) = v_i(\ell_i,r_i) = v_i(r_i,1)$. Using some simple preprocessing we can make sure that those positions are unique and that they can be described using $O(\log(1/\eps))$ bits. Now one of the following two cases must occur:
\begin{itemize}
    \item There exist two agents, say agents 1 and 2, such that $\ell_1 < \ell_2$, but $r_2 < r_1$. In that case, we let $a := \ell_1$ and $b := r_1$. Note that agent 1 is indifferent between the three pieces when we cut at $a$ and $b$. Furthermore, by monotonicy of $v_2$, agent $2$ prefers the middle piece $[a,b]$. Thus, $a$ and $b$ satisfy the desiderata and we can move to the second step of the protocol.
    \item By renaming the agents if needed, we have $\ell_1 \leq \ell_2 \leq \ell_3$ and $r_1 \leq r_2 \leq r_3$. In that case, we let $a := \ell_2$ and $b := r_2$. Clearly, agent 2 is indifferent between the three pieces. Note that by monotonicity of $v_1$, agent 1 likes the leftmost piece at least as much as the rightmost piece. Similarly, agent 3 likes the rightmost piece at least as much as the leftmost piece. If one of them prefers the middle piece, $a$ and $b$ satisfy the desiderata and we proceed with the second step. Otherwise, giving the leftmost piece to agent 1, the rightmost piece to agent 3, and the middle piece to agent 2 yields an envy-free allocation.
\end{itemize}

Given $a$ and $b$ satisfying the desiderata, the second step of the protocol proceeds as follows. After renaming the agents if needed, we have that agent 1 is indifferent between the three pieces, while agents 2 prefers the middle piece. Let $c_1$ denote the midpoint of the valuation of agent 1, i.e., $v_1(0,c) = v_1(c,1)$, and $c_2$ denote the midpoint of the valuation of agent 1, i.e., $v_2(0,c) = v_2(c,1)$. Note that $c_1, c_2 \in [a,b]$.

For any $\ell \in [a,c_1]$, let $r(\ell)$ denote the cut satisfying $v_1(0,\ell) = v_1(r(\ell),1)$. Note that $r(a) = b$, $r(c_1) = c_1$, and $r(\cdot)$ decreases monotonically. Similarly, for any $\ell \in [a,c_2]$, let $r'(\ell)$ denote the cut satisfying $v_2(\ell,r'(\ell)) = \max\{v_2(0,\ell), v_2(r'(\ell),1)\}$. It can be shown that $r'(a) \leq b$, $r'(c_2) = 1$, and $r'(\cdot)$ increases monotonically. If $c_2 \leq c_1$, then $r(\cdot)$ and $r'(\cdot)$ must cross for some $\ell \in [a,c_2]$. If $c_1 \leq c_2$, then $r(\cdot)$ and $r'(\cdot)$ must cross for some $\ell \in [a,c_1]$. In either case, we can find such a crossing point (approximately) using $O(\log(1/\eps))$ communication by reducing to the monotone-crossing problem. We refer to \cite{BranzeiN19-cake-communication} for the details. It is then easy to see that if $r(\ell) \approx r'(\ell)$ the division $(\ell, r(\ell))$ must be $\eps$-envy-free.
\end{proof}

Our algorithm for four monotone agents from \cref{sec:algo}, together with the monotone-crossing problem of \cite{BranzeiN19-cake-communication}, yield the following result.

\begin{lemma}
For four agents with monotone valuations, there exists a $O(\log^2(1/\eps))$ communication protocol that finds an $\eps$-envy-free connected allocation.
\end{lemma}

\begin{proof}[Proof sketch]
We show that the algorithm presented in \cref{sec:algo} can be simulated by a communication protocol with total cost only $O(\log^2(1/\eps))$. We think of Agent 1 as running the algorithm and communicating with the other agents whenever needed.

For the first step of the algorithm, it suffices for Agent 1 to find the equipartition of the cake into four equal parts according to its own valuation, and no communication is needed for that.

For the second step, Agent 1 performs $O(\log(1/\eps))$ steps of binary search. Every step of binary search consists in checking whether Condition A or B holds at some value $\alpha$. We will argue below that this only requires $O(\log(1/\eps))$ communication. Thus, the total cost of the second step is $O(\log^2(1/\eps))$.

Finally, in the third step, Agent 1 just needs to output a division that satisfies Condition A or B at some given value $\alpha$. The arguments below will show that this only requires $O(\log(1/\eps))$ communication.

It remains to argue that given some value $\alpha$, we can, using only $O(\log(1/\eps))$ communication, check whether Condition A or B holds at value $\alpha$, and, if so, output a division that satisfies it. This can be shown by considering all the possible cases as in the proof of \cref{lem:algo-queries}, and checking that each of them can be handled using only $O(\log(1/\eps))$ communication.

For Condition A, Agent 1 can determine the positions of all the cuts by itself, then use $O(\log(1/\eps))$ communication to share these positions with the other agents, who can then provide all the necessary information to decide whether the condition holds using $O(1)$ communication.

For Condition B, the challenge is to determine the positions of the cuts using only $O(\log(1/\eps))$ communication. Once these positions are known, $O(1)$ communication is again sufficient to check the condition. Below, we use the same notation as in the proof of \cref{lem:algo-queries} to argue that these cuts can be found using only $O(\log(1/\eps))$ communication.

If the pieces $k$ and $k'$ are adjacent, then Agent 1 can figure out the positions of two cuts, and inform Agent $i$. Agent $i$ can then figure out the position of the last cut and inform everyone. This uses $O(\log(1/\eps))$ communication.

If there is one piece between pieces $k$ and $k'$, then Agent 1 can figure out the position of one cut, and inform Agent $i$. The positions of the other two cuts can then be found using $O(\log(1/\eps))$ communication by reducing to the monotone-crossing problem~\cite{BranzeiN19-cake-communication}. We omit the details.

Finally, if $k$ is the leftmost piece and $k'$ is the rightmost piece, then Agent 1 and Agent $i$ can find the positions of the leftmost and rightmost cut using $O(\log(1/\eps))$ communication again by reducing to the monotone-crossing problem~\cite{BranzeiN19-cake-communication}. We omit the details. The middle cut can then be found by Agent 1.
\end{proof}

\bigskip
\subsubsection*{Acknowledgments}
We would like to thank the reviewers for useful comments and suggestions, and in particular one reviewer for suggesting a simplification of the proof in Section 4. We are also grateful to Simina Brânzei for helpful discussions. The first author was supported by the Swiss State Secretariat for Education, Research and Innovation (SERI) under contract number MB22.00026. The second author was supported by NSF CCF-2112824, and a David and Lucile Packard Fellowship.

\bibliographystyle{alphaurl}
\bibliography{references}

\end{document}